\documentclass[12pt]{article}
\usepackage[latin9]{inputenc}
\usepackage{geometry}
\usepackage{units}
\usepackage{ifsym}
\usepackage{amsmath}
\usepackage{amssymb}
\usepackage{setspace}
\usepackage{wasysym}
\usepackage{esint}
\usepackage[authoryear]{natbib}

\usepackage{caption}
\usepackage{subcaption}

\makeatletter

\usepackage{amsfonts}\usepackage{amsthm}\usepackage{multirow}\usepackage{booktabs}\usepackage{tabularx}\usepackage{rotating}\usepackage{nicefrac}

\makeatother

\begin{document}
\newtheorem{observation}{Observation}

\singlespacing

\title{Functional Forms for Tractable Economic Models \\
and the Cost Structure of International Trade\thanks{This paper replaces its previous working versions
 ``Pass-Through and Demand Forms''/``A Tractable Approach to Pass-Through Patterns''/``The Average-Marginal Relationship and Tractable Equilibrium Forms''. 
We are grateful to many colleagues and seminar participants
for helpful comments. We appreciate the research assistance of Konstantin
Egorov, Eric Guan, Franklin Liu, Eva Lyubich, Yali Miao, Daichi Ueda,
Ryo Takahashi, Huan Wang, and Xichao Wang. This research was funded by the Kauffman
Foundation, the Becker Friedman Institute for Research in Economics,
the Japan Science and Technology Agency and the Japan Society for
the Promotion of Science to which we are grateful. We are particularly
indebted to Jeremy Bulow for detailed discussion and for inspiring
this work and to James Heckman for advice on relevant theorems in
duration analysis and nonparametric estimation. All errors are our
own.}\vspace{5mm}}

\author{Michal Fabinger\thanks{Graduate School of Economics, University of Tokyo, 7-3-1 Hongo, Bunkyo-ku,
Tokyo 113-0033, Japan and CERGE-EI, Prague, the Czech Republic: fabinger@e.u-tokyo.ac.jp.}\and E. Glen Weyl\thanks{Microsoft Research New York City, 641 Avenue of the Americas, New York, 10011 USA and Department of Economics, Princeton University: glenweyl@microsoft.com.}}

\date{August 2018\vspace{0cm} }
\maketitle\thispagestyle{empty}
\begin{abstract}
We present functional forms allowing a broader range of analytic solutions to common economic equilibrium problems. These can increase the realism of pen-and-paper solutions or speed large-scale numerical solutions as computational subroutines. We use the latter approach to build a tractable heterogeneous firm model of international trade accommodating economies of scale in export and diseconomies of scale in production, providing a natural, unified solution to several puzzles concerning trade costs. We briefly highlight applications in a range of other fields. Our method of generating analytic solutions is a discrete approximation to a logarithmically modified Laplace transform of equilibrium conditions.

\end{abstract}
 \onehalfspacing

\newpage{}

 \section{Introduction}\label{SectionIntroduction}

Analytic solutions have played a major role in many fields of economics. They are useful both as closed-form, pencil-and-paper solutions to applied
theory models, and as components (subroutines) of larger models, making them more computationally tractable.\footnote{This later type of use is particularly important in the closely allied computationally-intensive field of Bayesian statistics
where closed-form tractable priors are typically used to approximate otherwise computationally intractable probability models.}
In this paper, we substantially expand the class of known analytic solutions to a broad class of standard economic models. { }We then illustrate
the application of these ideas to a computationally-intensive model of international trade that helps resolve a long-standing puzzle about trade
costs by allowing more realistic functional forms of such costs.\footnote{In this main, computationally
intensive application we find that analytic characterization of the solutions of sub-problems in larger-scale models is particularly useful in conjunction
with analytic-differentiation software, graphics processing units (GPUs), and related optimization algorithms. GPU computing has witnessed dramatic
developments over the last few years, which our work benefitted from.}

We observe that most frequently used functional forms that lead to analytic solutions in economics, namely linear and constant-elasticity functions,
share a convenient property: They preserve functional forms under transformations that we refer to as {``}average-marginal transformations{''}. That
is, the functional form of the average value of the function is the same as that of its derivative. { }Formally, we say that a functional form class
is preserved by average-marginal transformations if for any function \(F(q)\) the class also contains any linear combination of \(F(q)\) and \(q
F'(q)\).\footnote{A simple economic interpretation would be to identify \(F(q)\) with the price
\(P(q)\) of a good sold by a monopolist, i.e. with the \textit{ average} revenue the firm receives per unit sold, in which case \(F(q)+q F'(q)=P(q)+q
P'(q)\) is the \textit{ marginal} revenue. The name of the transformation is chosen to be consistent with this and similar examples. The Bulow-Pfleiderer
demand class \citep{bulow} discussed later is also invariant under average-marginal transformations.}
We then find all functions that have such property. 

Within this class, we identify functional forms that have a given level of {``}algebraic tractability{''}, a property we define. These are linear
combinations of power functions satisfying certain conditions. When used to represent demand and cost curves they lead to economic optimization conditions
that may be transformed to polynomial equations of a degree smaller than 5. These, in turn, may be solved explicitly by the method of radicals. This
substantially generalizes the simple analytic solutions that economists are familiar with in the case of constant marginal cost and either linear
or constant-elasticity demand. Even beyond degree 5, precise solutions to such polynomial equations are available at minimal cost in standard mathematical
software.

We show that elements of functional form classes preserved by average-marginal transformations also have advantageous properties when applied to
aggregation over heterogeneous firms in monopolistically competitive models: They lead to closed-form aggregation integrals under very flexible assumptions.
This means that a problem with a continuum of heterogeneous firms may be reduced to a set of explicit equations at the macroeconomic level.

In our method, the existence of closed-form solutions to optimization conditions sometimes requires parameter restrictions involving parameters both
from the supply side and the demand side. These restrictions may or may not be approximately satisfied in a particular market. { }If they are not
satisfied, one may be tempted to conclude that our method is not applicable. Most likely this is the reason why economists have not found (or have
not attempted to find) the kind of solutions we discuss in our paper.

We explain, however, that the range of applicability of our method is larger than it may seem at first sight as this issue does not pose a large
problem. Even if the parameter restrictions are not satisfied, analytic solutions { }at other parameter values may be used to construct an interpolation
that covers parameter values of interest. In this way one can extend the usefulness of our analytic method. Another way is to realize that a given
demand or cost function may be \textit{ approximated} by functions that satisfy our restrictions, in which case the restrictions are satisfied by
choice.\footnote{Yet way of extending the usefulness of the solutions is to use Taylor series
expansions around them, which may be useful for certain types of models.}

While our approach is useful in many fields of economics, as we illustrate, the main application we focus on in this paper belongs to the field of
international trade. Analytic tractability has been important for international trade to the extent that almost all models assumed constant marginal
costs of both production and logistics/shipping. Under such assumptions trade models are much more straightforward to solve. Yet, as we discuss in
detail, these assumptions contrast with models of cost used by the logistics managers that economists are presumably attempting to describe. As we
show, our functional forms preserve analytical tractability while allowing the realism of matching such models. 

Our primary application in this paper shows how such more realistic models of cost can help resolve the trade cost puzzle in a model of world trade
flows with heterogeneous firms.\footnote{Even though we do need considerable computational power
to fit our model to the data, without the tractability of our functional forms the computations would be significantly harder and we would not have
attempted to obtain a calibration of world trade flows that we discuss below.} Standard models of international
trade attribute the observed rapid falloff of trade flows with distance to trade costs that increase dramatically with distance. But we have no reason
to believe that such dramatic distance dependence of trade costs exists in the real world. Container shipping charges depend on distance only modestly,
and in any case, represent only a tiny fraction of the value of the transported goods. { }A similar statement holds for the so-called iceberg trade
costs, i.e., the damage of goods during transportation: We know goods typically do not get damaged during transport, and if they do, the damage probability
is unlikely to strongly increase with the distance a shipping container traveled over the ocean. While trade costs may be sizeable, they are much
more likely to be associated with shipment preparation and coordination or with loading and unloading, rather than with the distance traveled over
the ocean. For this reason, the rapid falloff of trade with distance represents a puzzle from the point of view of standard models of international
trade.\footnote{This puzzle in various forms has been discussed by many authors; see\, \,\citet{disdier2008puzzling} for an overview and \citet{head2013separates} for an in-depth discussion of the problem.}

Our model resolves this puzzle in a very natural way. Firms find it costly to produce larger quantities due to increasing marginal costs of production.
At the same time, they find it beneficial to concentrate their exports to a few destinations due to economies of scale in shipping. With this cost
structure, even a small cost advantage of a particular destination will be enough to make the firm export there instead of other destinations. If
trade costs are slightly smaller for closer destinations, this cost advantage will lead to substantially larger trade flows at smaller distances
and substantially smaller trade flows at larger distances. 

The model also resolves a puzzle related to firm entry into export markets. Although it is not as widely discussed as the trade cost puzzle, it is
a clear empirical regularity that models with constant marginal costs cannot address in a natural way. In the data, one can often see two similar
firms, say, from China, one exporting to, say, Portugal and not to Greece and the other exporting to Greece and not to Portugal. To reconcile such
patterns with the assumption of constant marginal cost of production, standard international trade models introduce stochastic cost shocks specific
to each firm-destination pair. These cost shocks have to be dramatically large. For the second firm they need to offset the entire profit the first
firm makes from its exports to Portugal. In the absence of any real-world phenomenon that could lead to cost shocks of this kind, this represents
a puzzle.\footnote{We discuss alternative mechanisms in Section \ref{SectionWorldTrade}.}

Our model explains this puzzle in a straightforward way. With increasing marginal costs of production and economies of scale in shipping, firms need
to solve a combinatorially difficult problem of choosing export destinations.\footnote{In economics
there are many combinatorially difficult problems, and we expect our methods to be useful there.} Different
approximate solutions of this choice problem can lead to different sets of export destinations, even if the maximized profits are almost the same.
One approximately optimal set of export destinations may include Portugal, while another one may include Greece. 

We solve the model using an iterative method involving an outer loop and an inner loop. The outer loop requires an evaluation of firms{'} profit
functions for many discrete choices of export destinations. Our functional forms bring a tractability advantage that makes these evaluations fast.
In the inner loop, we solve for a general equilibrium of the world economy keeping the discrete choices fixed. There using our functional forms is
helpful because it allows for an analytic calculation of derivatives that are needed for accelerated gradient descent algorithms. 

Separately, we develop many other applications of the proposed functional forms. For the model of outsourcing decisions in a sequential supply chain
constructed by  \citet{antras}, we reformulate the theory
to simplify the analysis and use this new formulation to apply our functional forms. This allows us to show that for more realistic demand functions,
outsourcing occurs at both the early (viz.$\, $raw materials) and late (viz.$\, $final commercial sales) stages of production, while intermediate
stages are performed in-house, corresponding to common observation of outsourcing patterns. For a model of labor bargaining by 
\citet{stole,stole2}, we tractably generalize the closed-form solutions that have been found
for linear or constant-elasticity demand and show that for realistic demand patterns the employment effects of bargaining have interesting and intuitive
cyclical patterns. We also discuss applications to imperfectly competitive supply chains, two-sided platforms, selection markets, auctions, and,
extensively, monopolistic competition.

Finally, we show that our method may be thought of as a discrete approximation to a logarithmically modified Laplace transform. It may also be thought
of as a sieve method of non-parametric econometrics. In addition, the transformed variables reveal economic properties of demand functions that would
appear accidental otherwise.

The next section provides a quick illustration of our functional forms with a focus on modeling demand under income inequality. Section 
\ref{sec:FunctionalFormsForAveragesAndMarginals} presents our main theoretical results. Section
 \ref{SectionWorldTrade} focuses on modeling world trade.
Section  \ref{SectionBreadthOfApplication} discusses
other applications. Section  \ref{SectionArbitraryDemandAndCostFunctions}
develops the theory connecting our tractable functions to a logarithmically modified Laplace transform. The paper also includes an appendix and supplementary
material.

 \section{Example: Replacing Constant-Elasticity Demand}\label{SectionReplacingConstantElasticityDemand} 

 \subsection{Constant-elasticity demand and its flexible replacement}\label{ConstantElasticityDemandAndItsFlexibleReplacement} 

The most canonical and widely-used demand form in economic analysis is the constant-elasticity specification, corresponding to inverse demand \(P(q)=a
q^{-b}\). It is frequently used because of its analytic tractability. Historically, it appeared in the economic literature because in discrete-choice
{ }settings it is plausible that product{'}s valuations follow the income distribution and the income distribution was believed to be approximately
Pareto, i.e., power-law.\footnote{In cases where each individual can consume at most one unit
of an indivisible product, the inverse demand function equals the reversed quantile function of the distribution of valuations (willingness to pay),
up to constant rescaling. The reversed income quantile function here refers to the function that maps a given quantile \(q\) measured starting at
the top of the income distribution to the corresponding valuation level. Note also that, of course, we do not wish to say that the most important
property of constant-elasticity demand lies in the context in which it first appeared. We are merely using this example as an illustration of our
approach to demand functions.} Modern data of income, however, led to different conclusions on the shape
of the income distribution.\footnote{The origin of constant-elasticity demand historically appears
to be the argument by \citet{saybook} that willingness to pay for a typical discrete-choice product is likely to be proportional
to income, and thus that the distribution of the willingness to pay should have the same shape as the income distribution. { }(Say{'}s assumption
is likely to be approximately correct for example for products that save a fixed amount of time to the owner, independently of their wealth.) Since
{ }early probate measurements of top incomes exhibited power laws (i.e., Pareto distributions) \citep{garnier,say}, by extrapolation
\citet{dupuit} and \citet{Mill} suggested that demand would have a constant elasticity. { }This observation appears
to be the origin of the modern focus on constant elasticity demand form \citep{dupuithistory,cobborigins}. However evidence on
broader income distributions that became available in the 20th century as the tax base expanded \citep{capital} shows that, beyond
the top incomes that were visible in 19th century data, the income distribution is roughly lognormal through the mid-range and thus has a probability
density function that is bell-shaped, rather than power-law. { }Distributions that accurately match income distributions throughout their full range
\citep{reedjorgensen,doublepareto,doubleparetolognormal} have a similar bell shape but incorporate the Pareto tail measured in
the 19th century data.} 

 In this section we discuss another demand form that is also highly analytically tractable but has more flexibility. This flexibility allows us to
get a much better match to the income distribution.\footnote{Similarly, this flexibility could
allow us to get a better match to a distribution of valuations in cases where it differs from the exact income distribution.}
As an illustration, we show that our proposed demand form leads to substantially different policy implications than the constant-elasticity form
in the socially important case of bias of innovative technical progress.

 \begin{figure}[t]
\includegraphics[width=3.4in]{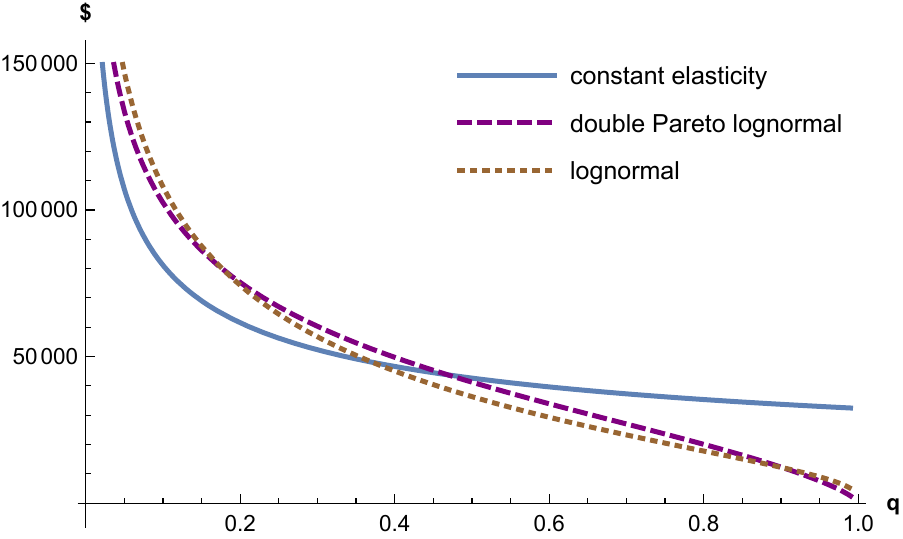}
 \includegraphics[width=3.4in]{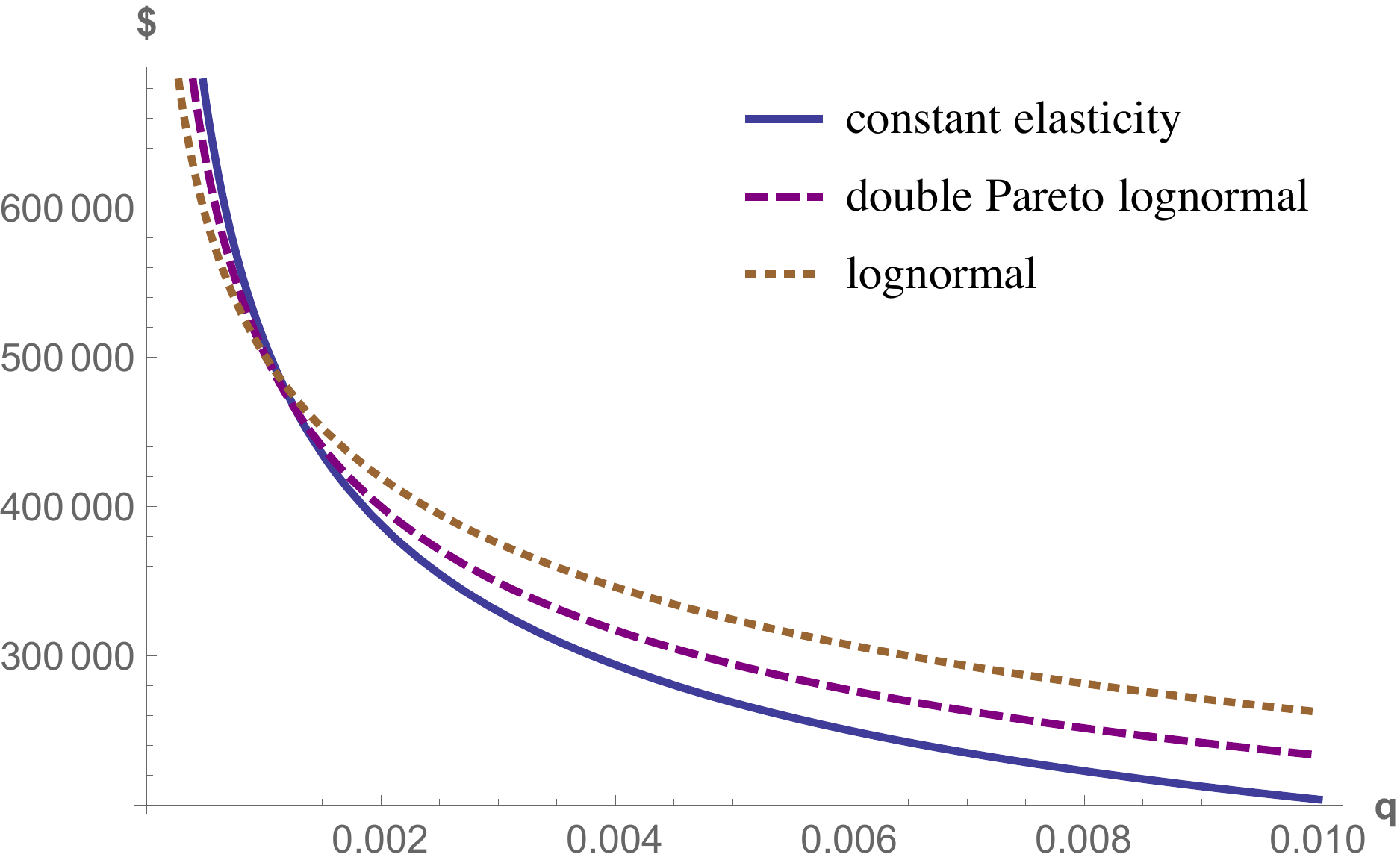}
 \caption{Comparing the fit of the best-fit lognormal and the best-fit constant-elasticity form
to a double-Pareto lognormal estimation of the 2012 US income distribution, represented as a demand function (reversed quantile function). { }Dollars
at any reversed quantile represent the income of the corresponding individual. { }On the left is the fit for the full income distribution, while
the right shows the upper tail. { }We used a standard calibration of a double Pareto lognormal proposed by \citet{reed} and used
the generalized method of moments to find the constant-elasticity demand function that best fits this throughout the full range of the income distribution.
In the upper tail the constant elasticity approximation is a bit better of a fit than the lognormal. However, in the rest of the income distribution
its fit is terrible, while the lognormal fits quite well (although in economic models it is harder to work with).}\label{CEfigure}
 \end{figure} 

 \begin{figure}[t]
\includegraphics[width=3.4in]{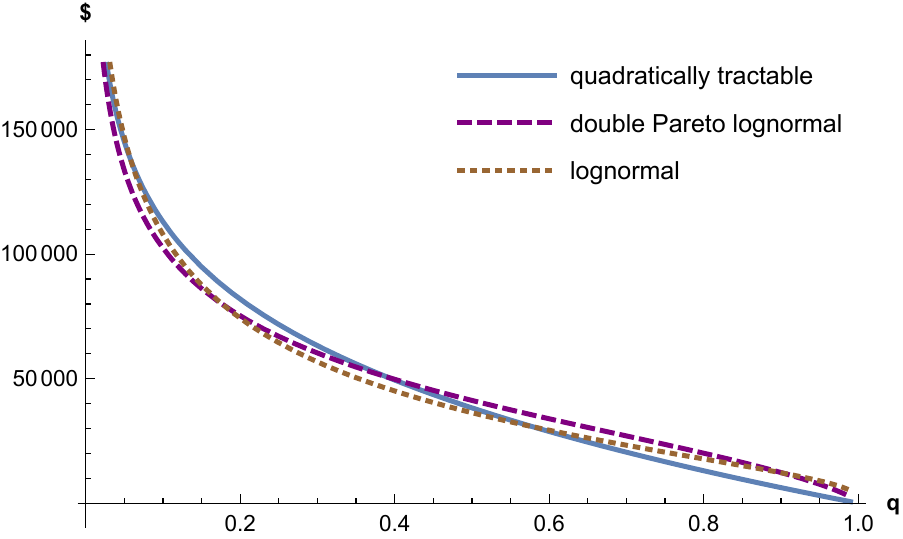}
 \includegraphics[width=3.4in]{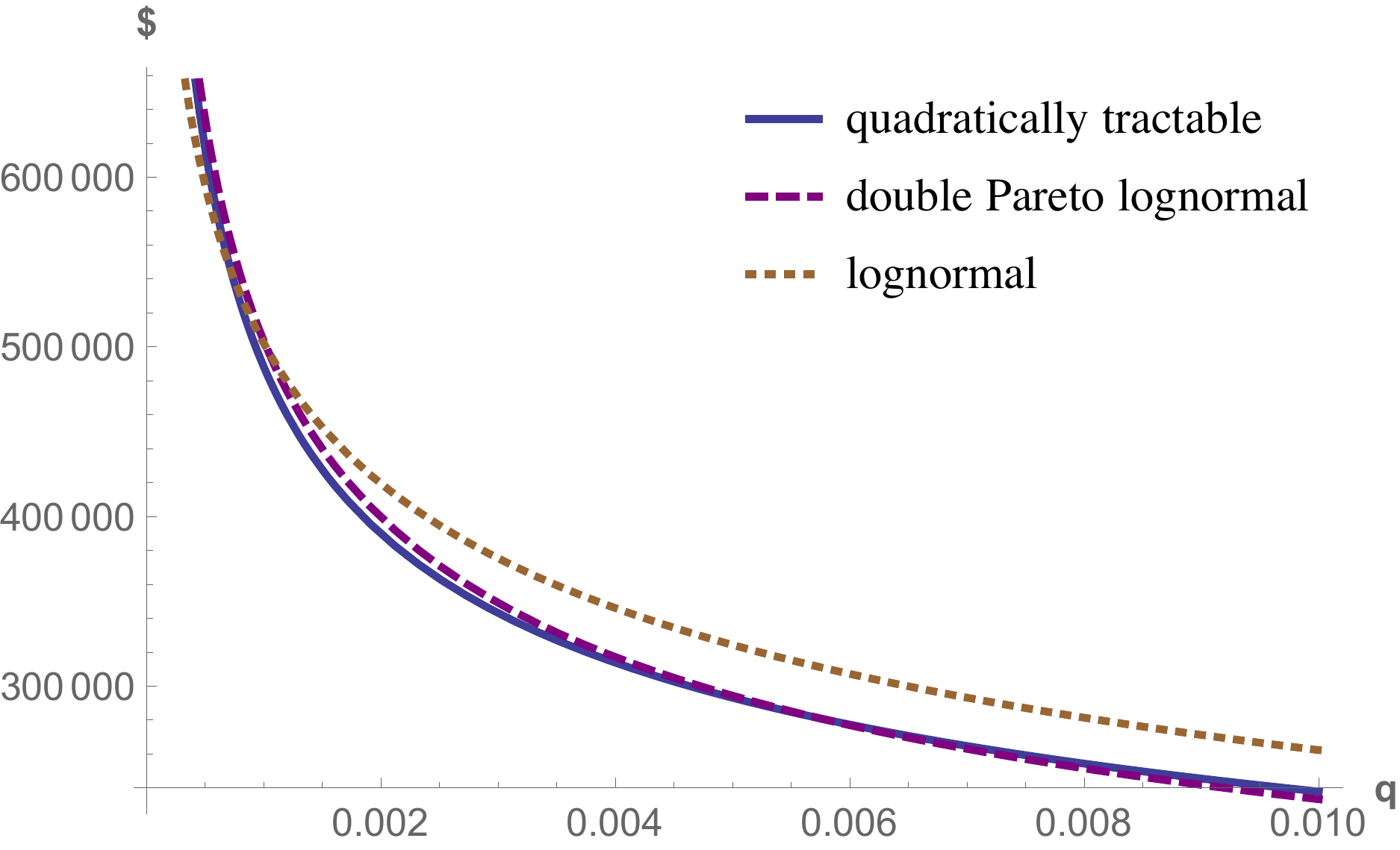}
 \caption{Comparing the fit of the best-fit lognormal and the best-fit quadratically solvable
form to a double-Pareto lognormal estimation of the US Income distribution, represented as a demand function (reversed quantile function). { }Dollars
at any reversed quantile represent the income of the corresponding individual. { }On the left is the fit for the full income distribution, while
the right shows the upper tail.}\label{ourfigure}
 \end{figure}

Consider the task of representing the empirical income distribution using a corresponding demand function. Constant-elasticity demand fails this
purpose, as illustrated in Figure  \ref{CEfigure}. This
is because the income distribution is not Pareto but approximately double Pareto lognormal  \citep{reedjorgensen,doublepareto,doubleparetolognormal}. { }Working with the double Pareto lognormal distribution (or with the more loosely fitting lognormal distribution) in economic
models would be quite difficult.\footnote{See Footnote \ref{FootnoteLognormalDistributionAnalyticDifficulty}.} { }To overcome this difficulty, we propose a functional form that allows for the same basic shape, but leads to calculations
almost as easy as for the constant-elasticity form:

 \begin{equation}
 P(q)=m-\\
 ma\\_{-}\left(\frac{q}{q\\_0}\right)^{-b}-\\
 ma\\_{+}\left(\frac{q}{q\\_0}\right)^b,\\
 \, \, \, \, \, \, \,\\
\hspace*{0.5ex}  \, \, \, \, \, \,\\
 a\\_{-}\equiv 1-a\\_{+}.\\
 \label{EquationQuadraticallyTractableFormForIncomeDistribution}
 \end{equation}
A set of parameter values that matches the income distribution (for the US in 2012) very well is \(a_-=-1/2, a_+=5/2,
b=2/5\). We obtained these values by a generalized-method-of-moments fit and rounding the results. The match is illustrated in Figure 
\ref{ourfigure}.

 \subsection{Bias in technological progress}\label{BiasInTechnologicalProgress} 

As a simple, illustrative application, we discuss the case of bias in technological progress described in 
\citet{kremer, worstcasebounds}. First, we do that for the case of constant-elasticity demand
and explain why it is highly tractable. Then we turn to our proposed demand form in Equation  \ref{EquationQuadraticallyTractableFormForIncomeDistribution} and show that it preserves a high degree of tractability that constant-elasticity demand has. 

When the private sector decides which products to develop, it chooses profit-making products, not necessarily those products that create the greatest
social value. This bias in technological progress depends on the discrepancy between private and social gains. 
\citet{kremer, worstcasebounds} consider the fraction of the social gains from creating a
new product that may be appropriated by a monopolist\footnote{\citet*{budishroinwilliams}
studied this problem recently in a different context. }, referred to as the appropriability ratio, and show
that the maximal fraction of potential surplus that may be lost due to imperfect appropriability is equal to one minus this appropriability ratio.
They compare different demand functions since they lead to different bias in research and development, but always assume no costs. Here we assume
a fixed demand function and consider biases at different levels of marginal production cost. We walk quite didactically through the process of solving
the model in order to illustrate the source of the tractability of the constant-elasticity form and why it carries over to our proposed generalized
form but not to the lognormal distribution form. { }We then follow  \citet{worstcasebounds}{'}s argument that a sensible demand function is one matching the world income distribution and use this as motivation for
using our form to study the impact of cost on the appropriability ratio, which is very different under our form than under constant elasticity.

Consider a monopolist with a constant marginal cost \(c\) and constant-elasticity (inverse) demand \(P(q)=a q^{-b}\). Her marginal revenue is \(P(q)+P'(q)q\).
Under the constant elasticity form, \(P'(q)q=-a b q^{-b}\), which has the same form as \(P(q)\), just a different multiplicative constant out front.
For this reason, the marginal revenue has the same form as well: \(\text{MR}(q)=a (1-b) q^{-b}\) . The monopolist optimally equates it to the marginal
cost, so the optimal quantity may be determined by solving the linear equation \(a(1-b)x=c\) with \(x\equiv q^{-b}\), yielding \(q=(a(1-b)/c)^{1/b}\).
From this it follows by substitution that the firm{'}s absolute markup is \(\overline{\text{PS}}\equiv \text{PS}/q=c b/(1-b)\), where \(\text{PS}\)
stands for the producer surplus. Furthermore, the average consumer surplus also has the same form as \(P(q)\), differing only by a multiplicative
constant: \[\overline{\text{CS}}\equiv \frac{\text{CS}}{q}=\frac{\int _0^qP\left(\tilde{q}\right)d\tilde{q}-P(q)q}{q}=\frac{\frac{a}{1-b}q^{1-b}-a q^{1-b}}{q}=\frac{a
b}{1-b}q^{-b}.\]
 Evaluated at the optimal quantity, the average consumer surplus is \(\overline{\text{CS}}= c b(1-b)^{-2}\). The appropriability
ratio, i.e., the ratio of producer surplus and the total surplus, may be evaluated as \(\overline{\text{PS}}/\left(\overline{\text{PS}}+\overline{\text{CS}}\right)=(1-b)/(2-b)\),
which is a constant independent of cost. Thus all products have precisely the same appropriability ratio, and cost is irrelevant to the bias of investments
in research and development.

If we tried to investigate this problem in a tractable way for more general demand functions that have been used in the economic literature, we could
use the Bulow-Pfleiderer demand introduced in the next section, which includes both constant-elasticity and linear demand as special cases. However,
we would again find that the cost \(c\) has no impact on the bias of technical progress.

If instead, we tried to investigate the implications of demand curves corresponding to other distributions of product valuations, such as the lognormal
distribution or the double-Pareto lognormal distribution (which fits the income distribution), we would quickly find that such investigation cannot
be carried out analytically.\footnote{\label{FootnoteLognormalDistributionAnalyticDifficulty}For
a lognormal distribution with mean \(\mu _{\ell }\) and standard deviation \(\sigma _{\ell }\) of the exponent, the inverse demand is \(P(q)=\exp
\left(\sigma _{\ell } \Phi ^{-1}(1-q)+\mu _{\ell }\right)\), where \(\Phi\) is the standard normal cumulative distribution function and where we
normalized maximum demand to 1. There is no analytic solution to the monopolist{'}s optimization condition \(\text{MR}=c\) because the following
expression is too complicated: \(P'(q)q=-\sqrt{2\pi }\sigma _{\ell } q \exp \left(\sigma _{\ell } \Phi ^{-1}(1-q)+\mu _{\ell }+\frac{1}{2}\left[
\Phi ^{-1}(1-q)\right]^2\right)=-P(q)\sqrt{2\pi }\sigma _{\ell } q \exp \left(\frac{1}{2}\left[\Phi ^{-1}(1-q)\right]^2\right)\). The more realistic
double-Pareto lognormal distribution leads to even more complicated expressions. { }Clearly, if demand functions of this kind were used inside larger
models, the absence of analytic tractability could quickly become a significant obstacle.}

Here we show that working with the demand form in Equation  \ref{EquationQuadraticallyTractableFormForIncomeDistribution} is much easier and elegant and leads to substantive economic results. Its marginal form \(P'(q)q\) has the same functional
form as \(P(q)\) itself:\[P'(q)q= m b a_-\left(\frac{q}{q_0}\right)^{-b}-m b a_+\left(\frac{q}{q_0}\right)^b.\]
If we introduce the notation\[a_{n,-}\equiv (1-b)^na_-,\text{$\, \, \, \, \, \, $ }a_{n,+}\equiv (1+b)^na_+,\text{$\, \, \, \, \, $ }x \equiv  \left(\frac{q}{q_0}\right)^b,\]
the monopolist{'}s first-order condition is just the quadratic equation\[-\text{  }a_{1,-}+\left(1-\frac{c}{m}\right) x-\text{  }a_{1,+}x^2=0.\]
This leads to the closed-form solution\[q=q_0x^{1/b},\, \, \, \, \, \, x=\frac{1}{2 a_{1,+}}\left(1-\frac{c}{m}+\sqrt{\left(1-\frac{c}{m}\right)^2-4 a_{1,-} a_{1,+}}\right).\]
The per-unit consumer and producer surplus again take the same functional form:\[\overline{\text{CS}}=- m b\text{  }\tilde{a}_-\left(\frac{q}{q_0}\right)^{-b}+m b\text{  }\tilde{a}_+ \left(\frac{q}{q_0}\right)^b\text{  
},\, \, \, \, \, \, \overline{\text{PS}}=m-c-m a_-\left(\frac{q}{q_0}\right)^{-b} -m a_+ \left(\frac{q}{q_0}\right)^b .\]
The appropriability ratio is then\[\frac{\overline{\text{PS}}}{\overline{\text{PS}}+\overline{\text{CS}}}= \frac{m-c-m\text{  }a_-\left(q\left/q_0\right.\right)^{-b} -m\text{ 
}a_+\left(q\left/q_0\right.\right)^b }{m-c+m a_{-1,-}\left(q\left/q_0\right.\right)^{-b} +m a_{-1,+}\left(q\left/q_0\right.\right)^b }=\frac{\left(1-b^2\right)
\left(-a_-+a_+ \left(q\left/q_0\right.\right)^{2 b}\right)}{-\left(2+b-b^2\right) a_-+\left(2-b-b^2\right) a_+ \left(q\left/q_0\right.\right)^{2
b}},\]
where the last equality was obtained by substituting for the marginal cost from the first-order condition. { }Substituting
the parameter values we specified right after Equation  \ref{EquationQuadraticallyTractableFormForIncomeDistribution} gives\[\frac{\overline{\text{PS}}}{\overline{\text{PS}}+\overline{\text{CS}}}=\frac{ 21+105\left(q\left/q_0\right.\right)^{4/5}}{56+180 \left(q\left/q_0\right.\right)^{4/5}}.\]
This equals \(21/56\approx 37.5\%\) for \(q=0\) (when the product serves a tiny fraction of the population) and monotonically
increases in q to { }\(53/118\approx 53.4\%\) for \(q=q_0\) (when most of the population is served). { }This suggests a bias towards cheap, mass-market
products and away from expensive products that mostly cater to the rich; of course, all this analysis is based, like 
\citeauthor{kremer}{'}s, on aggregate surplus and might well reverse if distributional concerns
were incorporated.$\quad $

While we focused here on biases from the appropriability ratio, it can be shown (in closed-form) that many other aspects of standard intellectual
property policy differ substantially under our form from the results under the constant pass-through class. For example, under our form the ratio
of consumer surplus to monopoly deadweight loss is much greater (usually by several times) than under the Bulow-Pfleiderer class so that patents
are more desirable and optimal patent protection is greater than under the standard forms. Similarly allowing pharmaceutical producers to price discriminate
often increases deadweight loss under the standard forms  \citep{acv},
while it is always beneficial under our form. Thus the standard forms are substantively misleading on a number of issues and the added complexity
of using our form is minimal.

\section{Central Results\label{sec:FunctionalFormsForAveragesAndMarginals}}

In the previous section we focused on a particular functional form
derived from our theory, a particular calibration target (the US income
distribution) and a particular application. However, our approach
applies much more broadly. We characterize all functional
forms that have the useful property of the form above: namely that, 
in the language of demand curves, linear combinations of marginal
revenue and inverse demand take the same form as inverse demand itself. Within these we then identify
functional forms that lead to closed-form solutions utilizing power functions
and the method of radicals.

\subsection{Form preservation under the average-marginal transformation}

Let us denote by $F\left(q\right)$ the average of an economic variable
that depends on $q$, where a baseline interpretation of $q$ is a quantity
of a good. The marginal variable is then $\left(qF\left(q\right)\right)'=F\left(q\right)+qF'\left(q\right)$.
We now formally define what it means for these two variables to have
the same functional form, as we alluded to in the previous section.

\newtheorem{definition}{Definition}

\begin{definition} \textbf{\emph{(Form Preservation)}} We say that
a functional form class $\mathcal{C}$ is {\em form-preserving} under average-marginal
transformations if for any function $F(q)\mathcal{\in C}$, the class
also contains any linear combination of $F(q)$ and $qF'\left(q\right)$.
In other words, $F\mathcal{\in C}\Rightarrow\mathcal{\forall}\left(a,b\right)\in\mathbb{R}^{2}:\ aF+bqF'\in\mathcal{C}$.
In economic terms, we interpret $F\left(q\right)$ as the average
of the variable $qF\left(q\right)$, such as revenue or cost, and
$F\left(q\right)+qF'\left(q\right)$ as its marginal counterpart.
This definition thus states that any linear combination of the average
and marginal variables belong to the defined class of functions.\footnote{Note
that any form-preserving class is also form-preserving under multiple
applications of operators of this type.}

\end{definition}

Obviously, if $\mathcal{\mathcal{C}}$ is taken to be a sufficiently
large (e.g. infinite-dimensional) class of functions it may be form-preserving in a fairly mechanical
way. For example, if it is the set of all analytic functions with the domain
$\left(0,\overline{q}\right)$ for some $\overline{q}$ then we know
that $a F\left(q\right)+b q F'\left(q\right)$ is also analytic and has
at least as large a domain. This observation is not very useful for
the purposes of tractability because the set of all analytic functions
with this domain contains many that, as we discussed in the previous
section, are not tractable using standard analytic and computational
methods. 

Thus we will naturally wish to consider smaller classes. It is, therefore,
useful to identify the most general set of finite-dimensional functional
form classes that are form-preserving under the average-marginal transformations
$F\rightarrow aF+bqF'$. Before stating the characterization theorem,
let us briefly clarify what we mean by the dimensionality of a functional
form class. For example, a functional form class $a_{1}e^{-a_{2}q}$,
where $a_{1}$ and $a_{2}$ are continuously varying real numbers
is two-dimensional, while $a_{1}e^{-a_{2}q}q^{-a_{3}}$ with continuously
varying real $a_{1}$, $a_{2}$, and $a_{3}$ is three-dimensional.\footnote{While this intuitive description is sufficient for practical purposes,
more formally we say that an \emph{$m$-dimensional functional form
class} is a subset of a space of functions (of a scalar, continuous
variable) that is homeomorphic to an $m$-dimensional manifold, possibly
with a boundary. Such manifold, with or without a boundary, is often
referred to as the \emph{moduli space}.}

\newtheorem{theorem}{Theorem}

\begin{theorem}\label{formpreserve}

\textbf{\emph{(Characterization of Form-Preserving Functions)}} Any
real finite-dimensional functional form class with domain $(0,\infty)$ 
(or an open subinterval of it)
that is form-preserving under average-marginal transformations must be a set of linear combinations of 
\[
\left(\log q\right)^{a_{jk}}q^{-t_{j}},\quad a_{jk}=0,1,...,n_{j}^{(1)},\quad j=1,2,...,N_{1},
\]
\[
\left(\log q\right)^{b_{jk}}\cos\left(\tilde{t}_{j}\log q\right)q^{-\hat{t}_{j}},\quad b_{jk}=0,1,...,n_{j}^{(2)},\quad j=1,2,...,N_{2},
\]
\[
\left(\log q\right)^{c_{jk}}\sin\left(\tilde{t}_{j}\log q\right)q^{-\hat{t}_{j}},\quad c_{jk}=0,1,...,n_{j}^{(2)},\quad j=1,2,...,N_{2},
\]
where $\left\{ t_{j}\right\} _{j=1}^{N_{1}}$, $\left\{ \tilde{t}_{j}\right\} _{j=1}^{N_{2}}$,
and $\left\{ \hat{t}_{j}\right\} _{j=1}^{N_{2}}$ are fixed sets of
real numbers and $N_1,N_2\in \mathbb N$ . If we exclude functions oscillating as $q\rightarrow0_{+}$,
only the functions in the first row are allowed. In that case the
most general form is the set of linear combinations of
\[
q^{-t_{j}},\quad q^{-t_{j}}\log q,\quad q^{-t_{j}}\left(\log q\right)^{2},\quad...\quad,q^{-t_{j}}\left(\log q\right)^{n_{j}},\quad j=1,2,...,N_{1}.
\]
\end{theorem}

The proof is provided in Appendix \ref{AppendixProofsOfTheorems}.

\subsection{Tractability}

We now provide a specific formal definition of ``tractability'' 
that allows
us to characterize the class of form-preserving functional forms that
have various levels of such tractability. While the term tractability is
constantly invoked in economics papers to justify various ``simplifying'' assumptions,
it is almost never defined formally.\footnote{Of course,
in other contexts the word ``tractability" may have other meanings that are
also useful. We specify below what we mean by ``tractability" in this paper.} 

A potential reason for this is that there is no standard, clear definition within applied mathematics 
of the notion of tractability
of the solution of mathematical equations. The theory of polynomial equations
establishes that generic
polynomial equations of degree at most four have solutions in terms
of ``the method of radicals'' (roots of different orders) and that
generic polynomial equations of higher degree have no such solutions, 
according to the Abel--Ruffini theorem.
But this theory does not imply that one could not extend the list of 
``closed-form'' functions by adding some other functions (other than roots)
to provide solutions to higher order polynomials.
In practice,
polynomial equations of any reasonably low order (say less than a hundred)
can be solved extremely rapidly by standard mathematical software
\citep{allpolynomial}.\footnote{Of course, the notion of 
``tractability" and ``closed-form solutions" is subjective to some extent. Equations whose
solutions may be expressed in terms of functions that are familiar enough
are often said to have closed-form solutions. That does not imply, however, 
that such notion is meaningless. Familiar functions are easier to work with 
for researchers thanks to existing intuition, as well as thanks to their implementation 
in symbolic or numerical software. In this paper we made definite choices to
resolve the terminological ambiguity.}

For this reason, we use a specific definition of tractability, which we call \emph{algebraic
tractability}, that is very simplistic: an equation is algebraically
tractable at some level $k$ if it can be solved using power functions and 
a solution
to a polynomial equation of degree no greater than $k$. While this
definition eliminates many other functions with known solutions, it
does a good job capturing existing forms that are widely considered
tractable while allowing an extension to richer forms in a pragmatic
manner given the ease with which polynomial equations can be solved
both analytically and computationally \citep{algebraic}.

An important feature of the (non-oscillating) class of functional
forms in Theorem \ref{formpreserve} is that if we include terms with
powers of logarithms we must also include all terms with powers of
logarithms below this. That is, if the class includes linear combinations
of $q\left(\log q\right)^{2}$ and $q^{-1/2}\left(\log q\right)^{2}$
it must also include linear combinations $q\log q$, $q^{-1/2}\log q$,
$q$ and $q^{-1/2}$. With a small number of (explicitly enumerable) exceptions, classes of
functional forms like this can rarely be solved in closed-form because
of the combination of power and logarithmic terms.\footnote{The most notable exception is the case when only a single power of
$q$ is used which can be divided out of the equation to yield a polynomial
in $\log q$. While this class is of some interest, we do not focus
on it here because it has the unappealing property that if one wishes
to include a constant term (which is often desirable as we discuss
below) one is limited to a small number of powers of logarithms and
all other parameters are set. There are other specific exceptions
and exploring the use of these is an interesting direction for future
research, but none offers the flexibility afforded by power functions
that we focus on below. This is likely why they have formed the basis
of so much prior work. We thus see the logarithm-based forms instead
as limits of the power forms that are worth including but not focusing
on.}

On the other hand, the even-simpler class of sums of power functions
nests all frequently-used tractable forms in the economic literature,
namely constant-elasticity demand combined with constant marginal
cost, linear demand combined with linear marginal cost as in \citet{horizontal},
and the  ``constant pass-through'' demand of  \citet{bulow} (henceforth BP)
with constant marginal cost.\footnote{In this section, for expositional purposes,
we discuss tractability from the point of view of monopoly problems. But it is worth noting
that tractability considerations would be exactly the same for Cournot oligopoly
and very similar in the many applications we discuss in this paper.}$^{,}$\footnote{The BP demand, defined below, 
gives constant pass-through rate of specific taxes to monopolist's prices
only in the case of a constant marginal cost. For this reason, we prefer to use the term Bulow-Pfleiderer (BP) demand,
instead of the frequently used term ``constant pass-through demand''.}As a result, we focus on functional form classes composed of linear
combinations of power functions $q^{-t_{j}}$.\footnote{There are a few cases not nested
in the forms of Theorem \ref{formpreserve} for which the firm's first-order condition may be solved. Hyperbolic demand curves
used by \citet{simonovska2015income} are one of them. Cases where the solution is in terms of the  Lambert W function
are the exponential utility function of \citet{newapproach,behrens} and single-product versions of 
the Almost Ideal Demand System and of translog demand.}

\begin{figure}
\begin{center}
\includegraphics[width=4in]{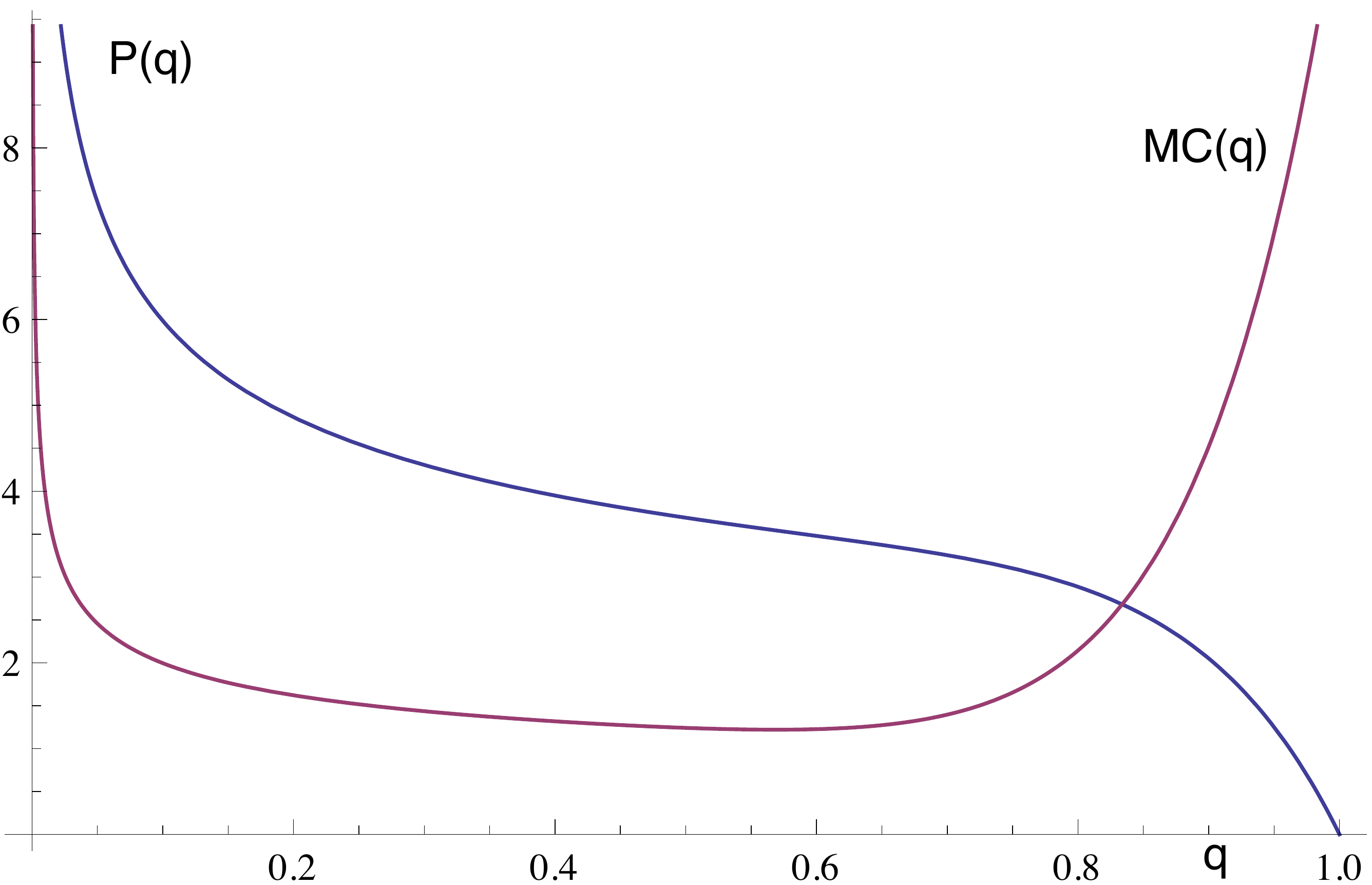}
\end{center}
\caption{Example of a bell-shaped-distribution-generated demand and U-shaped cost curve contributing to equilibrium
conditions that can be solved linearly: $P(q)=3\left(q^{-0.3}-q^{10}\right)$ and $MC(q)=q^{-0.3}+10q^{10}$.}\label{linearsolve}
\end{figure}

The BP demand corresponds to $P\left(q\right)=p_{0}+p_{t}q^{-t}$
for some real constants $t$, $p_{0}$ and $p_{t}$, not necessarily
all positive. In a monopolist's first-order condition, constant marginal cost
enters in the same way as $-p_{0}$. In this sense,
constant marginal cost is compatible with this demand side 
specification. (Similarly, linear marginal cost would be fully 
compatible with the demand side in the special case of
$t=-1$, i.e., linear demand.)

Using the BP demand form with constant marginal cost leads both to
tractability and to an important substantive implication:
the constancy of the pass-through rate of the constant marginal cost
to price. However, it is clearly possible to preserve the former property
without the latter. For example, consider inverse demand and average
cost of the form $P(q)=p_{s}q^{-s}+p_{t}q^{-t}$ 
and $AC(q)=ac_{s}q^{-s}+ac_{t}q^{-t}$.\footnote{\citet{demanding} 
studied the properties such bi-power form applied to inverse demand functions in combination
with constant marginal cost. Their goal was not to obtain closed-form solutions.}
Then the monopolist's first-order condition gives
\[
\left(p_{s}-ac_{s}\right)\left(1-s\right)q^{-s}+\left(p_{t}-ac_{t}\right)(1-t)q^{-t}=0\implies q=\left(-\frac{\left(p_{s}-ac_{s}\right)\left(1-s\right)}{\left(p_{t}-ac_{t}\right)(1-t)}\right)^{\frac{1}{s-t}}.
\]
This more general form thus still leads to a closed-form solution but offers
substantially more flexibility. For example, it can accommodate simultaneously
U-shaped cost curves and demand curves generated by a bell-shaped 
valuation distribution (in the sense of discrete choice).
 Figure \ref{linearsolve} provides an
example. A disadvantage of this form, however, is that it does not
include a constant term. A constant term would have been useful
for studying the pass-through rate and similar comparative statics.
Another
disadvantage is the absence of an explicit expression for the 
direct demand $Q(p)=P^{-1}(p).$

It is thus useful to look beyond systems that lead to a linear equation
(after a substitution using a power function).
Quadratic, cubic and quartic equations also yield closed-form solutions
by the method of radicals. Furthermore, polynomials of higher, but
still small, order can be solved extremely quickly by most mathematical
software without resorting to numerical search. For this reason, we define tractability
in terms of the degree of polynomial solution a form admits.

\begin{definition}\label{DefinitionAlgebraicTractability}

\textbf{\emph{(Tractability)}} We say that an economic problem involving
a scalar $q$ is \emph{algebraically tractable} \emph{at level $k$}
if a definite power of $q$ is the solution of a polynomial equation
of order $k$. For short we often refer to this simply as ``tractability''
and use adverbial forms for low $k$ (e.g. linearly or quadratically
tractable). By classical results of the theory of polynomial equations,
 only for $k\leq4$ can
such an equation be explicitly solved by the method of radicals and
thus we refer to economic problems that are algebraically tractable
at level $k\leq4$ as \emph{analytically tractable}. 

\end{definition}

We now characterize the set of functional forms from the power class
that are tractable at level $k$ for any positive integer $k$. A 
very naive conjecture
 based on the above discussion is that
this is simply the set of forms that can be written as the sum of
$k+1$ powers. To see why this is wrong, consider the equation 
\[
q+1+q^{-\nicefrac{1}{2}}=0.
\]
This does not admit a quadratic solution, but can be solved cubically
by defining $x\equiv q^{-\nicefrac{1}{2}}$, transforming the equation
into 
\[
x^{-2}+1+x=0\iff x^{3}+x^{2}+1=0.
\]
While the quadratic solution fails here, the cubic succeeds, because
the gap between the power of the first and second term ($1-0=1$) is
not equal to that between the second and third term ($0-\left(-\nicefrac{1}{2}\right)=\nicefrac{1}{2}$);
instead it is twice the second gap, implying that there is a ``missing''
term $q^{\nicefrac{1}{2}}$ in the equation. On the other hand, the
equation 
\[
q^{\nicefrac{1}{2}}+1+q^{-\nicefrac{1}{2}}=0
\]
 \emph{is} quadratically tractable because the gap between the first
and second powers equals that between the second and third. More broadly
the number of such \emph{evenly-spaced} powers sufficient to represent
the class determines its level of tractability.

\begin{theorem}\label{TheoremClosedFormSolutions}

\textbf{\emph{(Closed-Form Solutions)}} A functional form class $\mathcal{\mathcal{C}}$
composed of all linear combinations of a finite set of powers of $q$
is algebraically tractable at level $k$ for generic linear coefficients
if and only if the powers included are $\left\{ a+bi\right\} _{i\in J}$
for some fixed real numbers $a$ and $b$ and some fixed set of integers
$J\subseteq\left\{ 0,\ldots,j\right\} $ for a fixed integer $j\leq k$.
More informally, a class of sum of power laws is tractable at level
$k$ if it consists of at most $k+1$ evenly-spaced powers of $q$.

\end{theorem}

\begin{table}
{\tiny
\begin{tabular}{|c|c|c|c|c|}
\hline 
Form & Tractability properties & Flexibility & Special cases &  Historical notes \\
\hline
$F(q)=f_0+f_{-1}q$ & \begin{tabular}{c} Linearly tractable \\ Linearly invertible \end{tabular}  &  Linear $MC$  & Constant $MC$  & \citet{horizontal} \\
\hline
$F(q)=f_0 + f_{t} q^{-t}$ & \begin{tabular}{c} Linearly tractable \\ Linearly invertible \end{tabular} & \begin{tabular}{c} Any constant \\ pass-through \end{tabular} & \begin{tabular}{c} Linear \\ Constant elasticity \\ Exponential\end{tabular}  & \begin{tabular}{c} BP \\ constant pass-through demand \end{tabular} \\
\hline
$F(q)=f_tq^{-t}+f_sq^{-s}$ & Linearly tractable & \begin{tabular}{c} Demand generated by \\ bell-shaped distribution\\ U-shaped cost \end{tabular}  & BP & \begin{tabular}{c} \citet{demanding} \\ bi-power demand \end{tabular} \\
\hline
$F(q)=f_tq^{-t}+f_0+f_{-t}q^{t}$ &  \begin{tabular}{c} Quadratically tractable \\ Quadratically invertible \end{tabular} &  \begin{tabular}{c} Income distribution \\ U-shaped cost \end{tabular}  & BP &  This paper  \\ 
\hline
$F(q)= f_0+f_{t}q^{-t}+f_{2t}q^{-2t}$ & \begin{tabular}{c} Quadratically tractable \\ Quadratically invertible \end{tabular}  & \begin{tabular}{c} Demand generated by \\ bell-shaped distribution \\ U-shaped cost \end{tabular} & BP & \begin{tabular}{c} \citet{demandforms} \\ APT demand \end{tabular}  \\
\hline

\end{tabular}}

\caption{Various classes of linearly or  quadratically tractable, form-preserving equilibrium systems discussed in this or previous papers. }
\label{tractableforms}
\end{table}

One example of applying this theorem was given in the previous section:
our tractable form involves 3 evenly spaced power laws and thus is
quadratically tractable. Table \ref{tractableforms} summarizes
a rich set of other possibilities covered by this theorem.  The demand side of some of these has appeared in previous literature as we cite in the paper, though only in the case of \citet{horizontal} are we aware of authors harnessing the accompanying cost-side flexibility.

\subsection{Aggregation over heterogeneous firms}\label{AggregationOverHeterogeneousFirms}

Models of international trade involving firm heterogeneity frequently use the 
framework of \citet{melitz} or \citet{melitzottaviano}, which assume respectively constant elasticity and linear demand.  While these forms clearly play a role in the tractability of those models, the models are not always explicitly solvable even under these forms.  Instead, the key properties these allow is that  the firms' optimization problems 
may be solved explicitly and that aggregation integrals over heterogeneous firms 
may be expressed in closed form, assuming Pareto-distributed firm productivity.

We present a theorem that shows that substantial generalizations of these models can still lead to
closed-form aggregation.
We defer a full model set-up to Supplementary Material \ref{AppendixFlexibleMelitzModel}, but it may be thought of simply
as the \citet{melitz} model with relaxed functional form assumptions on the shape of demand, supply, and firm productivity 
distributions.

\begin{theorem}\label{TheoremAggregation}

\textbf{\emph{(Aggregation)}} Suppose that the utility structure implies an inverse demand curve \(P(q)\) and that firms have 
 marginal cost functions \(\text{MC}(q)=a \text{MC}_1(q)+\text{MC}_0(q)\), where \(a\) is an idiosyncratic parameter influencing the firm{'}s
productivity, distributed according to a cumulative distribution function \(G(a)\). Assume that \(P\), \(\text{MC}_0\), \(\text{MC}_1\), and \(G\)
are linear combinations of powers of their arguments, with the second order condition for the firm{'}s profit maximization satisfied. Furthermore, suppose that the powers are
such that \(\text{MC}_1\) and the difference between marginal revenue and \(\text{MC}_0\) are both of the form \(q^{\beta } N\left(q^{\alpha }\right)\)
with common $\alpha $, but possibly differing $\beta $ and polynomials \(N\). Then the aggregation integrals for the firms{'} revenue, cost, and profit may
be performed explicitly. The resulting expressions may contain special functions, namely the standard hypergeometric function, the standard Appell
function, or more generally Lauricella functions, and in the case of high-order polynomials (higher-order tractable specifications), increasingly high-degree polynomial root functions.  

\end{theorem}

While this result is closely related to our general theory and our other applications (in particular, because this aggregation is possible when the relevant variables have our proposed forms), there are also a few differences worth noting. First, aggregation is still possible when the heterogeneous component $MC_1(q)$ of marginal cost is ``shifted'' (in the exponent space) by a uniform multiplicative power factor relative to $MR(q)-MC_0(q)$. This corresponds to the ``possibly differing $\beta $'' in the statement of the theorem. Second, our results here are about aggregation, not solution, and the resulting functions are not, therefore, solutions to polynomial equations 
but rather various functions that may be exotic to some economists, but are widely used in mathematics and related applied fields. Finally, as the complexity of the forms rises, it is the complexity of these functions that rises.   

Closed-form aggregation is useful for at least three reasons. First of all, in the simplest cases the resulting aggregation integrals are just polynomial functions, which means that at the aggregate level the economic equations are relatively simple. Second, when the aggregation integrals lead to commonly used special functions, these are likely to be implemented in numerical software of the researcher's choice. The researcher gets a fast and numerically reliable implementation of these functions and their derivatives without spending time on approximation methods. Third, it is possible to take advantage of the properties of these functions that have been studied in the mathematical literature.

 \subsection{Interpolation between solutions}\label{SubsectionInterpolationBetweenSolutions}

We have discussed how to obtain closed-form solutions in economic modeling. We used linear combinations of power function and imposed conditions
on their exponents. It is natural to ask what happens if these conditions are not satisfied. Suppose we have a computationally intensive model whose
numerical solutions rely on closed-form solutions to its sub-problems. If we relax our assumptions on the exponents we just mentioned, the sub-problems
will not be solvable in closed form and obtaining numerical solutions to the full problem may require an excessive amount of time. Here we would
like to point out that in this case we have another way to proceed: we can solve the full problem at special loci where the conditions on exponents
are satisfied, and then interpolate between the resulting solutions.

Let us illustrate this approach with a toy example that is not computationally intensive. Consider a monopolistic firm with marginal revenue \(\text{MR}(q)=\text{MR}_0q^{-b_R}\),
\(b_R>0\) and marginal cost \(\text{MC}(q)=\text{MC}_0+\text{MC}_1q^{-b_C}\), { }\(\text{MC}_1>0\). After the substitution \(x=q^{-b_R}\), the firm{'}s
first-order condition becomes \(\text{MC}_0+\text{MC}_1x^b=\text{MR}_0x\), with \(b\equiv b_C/b_R\). This equation admits closed-form solutions by
the method of radicals for \(b\in\)$\{$\(-3\), \(-2,-1,-\frac{1}{2},-\frac{1}{3},0,\)\(\frac{1}{4}\),\(\frac{1}{3}\),\(\frac{1}{2}\),1,2,3,4$\}$.
The second-order condition is satisfied only for \(b<1\), so we restrict our attention to the first 9 values. For illustration, consider the simple
goal of finding numerical values of \(q\) for \(b\) between these points. Instead of solving the first-order conditions by usual numerical methods,
we may interpolate between the closed-form solutions, say, using cubic splines. Figure  \ref{FigureInterpolationBetweenAnalyticSolutions} shows the result of such interpolation, as well as the true solutions to the first-order condition, for \(\text{MC}_0=\text{MC}_1=\text{MR}_0=1\).
The agreement is extremely good, with average absolute value of relative deviations equal to 0.00013 and maximum absolute value of relative deviations
equal to 0.00056.\footnote{The corresponding values for Mathematica{'}s Hermite polynomial interpolation
were 0.00069 and 0.0021.}

If variables of interest in large scale, computationally intensive problems are similarly well behaved, then clearly the interpolation method could
save remarkable amounts of computation time and research budgets. There are other ways to extend the usefulness of the closed-form solutions to other
parameter values. For example, it may be possible to perform a Taylor expansion around a given closed-form solution. Such approach may also be combined
with the interpolation method. 

More broadly, one can view our approach to economic modeling as resembling pragmatic approaches in Bayesian Statistics. In that literature, it is
usually impossible to compute the posterior probability distribution associated with most prior distributions given the observed, often large, data
set. { }It is therefore common to approximate the prior by one selected from a class of prior distributions which are known to update to another
prior within that class in closed form as this minimizes the computational requirements of updating. { }In a similar manner, our tractable equilibrium
forms may approximate arbitrary cost and demand curves, while allowing solutions in closed-form which allow nesting inside of computationally intensive
models.

 \begin{figure}[t]\begin{center}\includegraphics[width=4in]{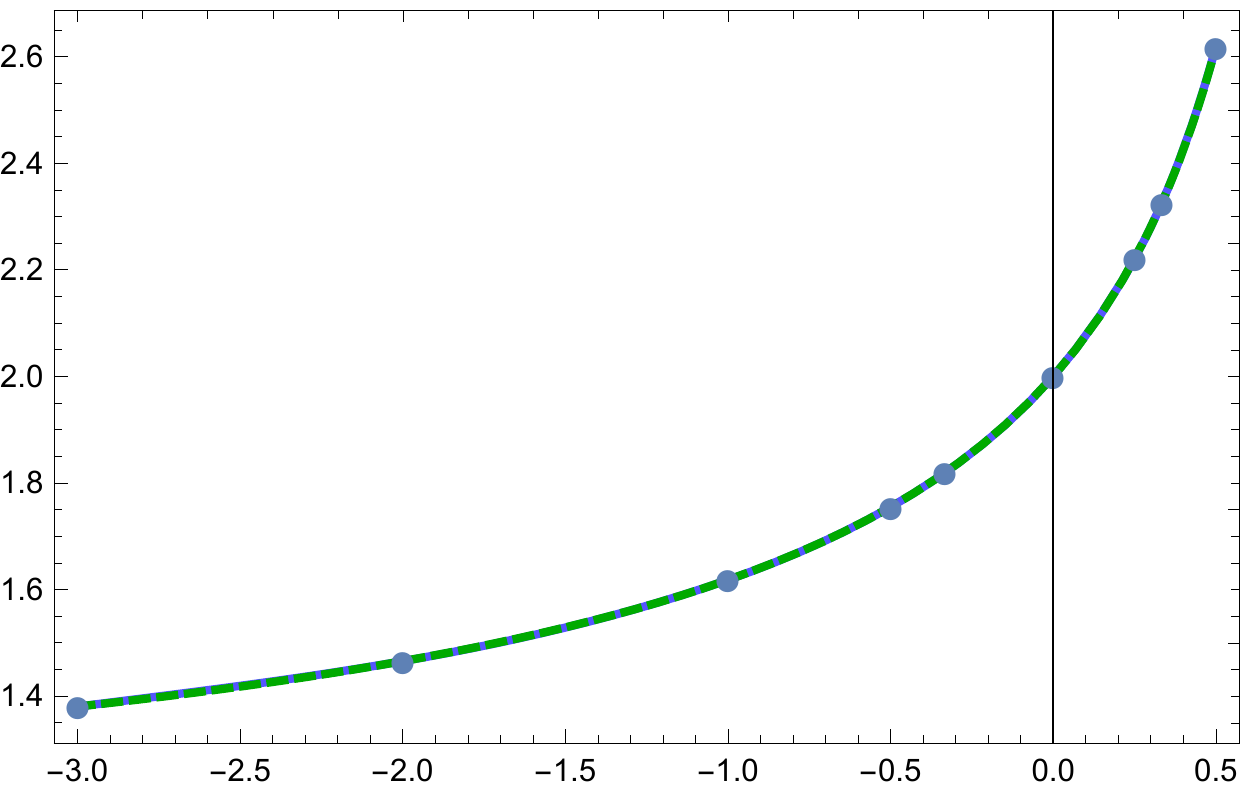}
 \caption{Comparison of an interpolation between analytic solutions and the correct values. The
blue dots represent analytic solutions. The blue solid line corresponds to an interpolation using cubic splines. The green dashed line represents
correct values.\\
}\label{FigureInterpolationBetweenAnalyticSolutions}
 \end{center}\end{figure}

 \vspace{2mm}In the next two sections we explore concrete
applications of our approach to closed-form solutions in economics. We will return to more theoretical matters in Section 
\ref{SectionArbitraryDemandAndCostFunctions}.

 \section{World Trade}\label{SectionWorldTrade}

 \subsection{Overview}\label{SubsectionOverview}

In this section we present a large-scale empirical application of our analytic approach to flexible functional forms in economics: a model of world
trade with a realistic cost structure for heterogeneous firms. 

International trade researchers almost always postulated constant marginal costs. Firms were assumed to have a constant marginal cost of production.
They were also assumed to face constant marginal costs of trade, either in the {``}iceberg{''} form (i.e., damage of goods as they are transported)
or in a per-unit form. Both of these assumptions are unrealistic. When we depart from them, we find an interpretation of world trade flows that is
dramatically different from the conventional view. The model{'}s parameters and predictions take realistic values, which resolves empirical puzzles
in the international trade literature.

Our model describes a world with multiple countries, with a general setup analogous to  \citet{melitz}. Our two important modifications are as follows. First, in addition to the usual iceberg cost, we allow for a specific cost
of trade that varies non-linearly with the traded quantity. Second, production is subject to increasing marginal cost, designed to capture the difficulty
of scaling the firm, e.g. due to internal agency problems. After discussing computational considerations and, separately, the two economic cost effects,
we return to the setup of the main model in Subsection  \ref{SubsectionModelSetup}.

 \subsection{Computational considerations}\label{SubsectionComputationalConsiderations} 

In applied fields of economics, such as the study of international trade, researchers can quickly reach the limits of what is computationally feasible
because of the number of economic agents and the high dimensionality of their choice sets (or state spaces). In our case, we study trade flows between
many countries involving heterogeneous firms, each of which is facing a combinatorially difficult decision problem. With powerful hardware, software,
and efficient algorithms, we were able to get a model fit for given parameter values in about a month and at a non-trivial cost. We were utilizing
our analytic solutions to sub-problems, without which the computation would be substantially longer and more costly.

Our functional forms help us in two ways: First, to evaluate firms{'} sales decisions conditional on the level of their marginal cost of production
and export entry decisions, we just need to evaluate closed-form solutions. This is crucial { }for being able to quickly evaluate a large number
of alternative sales patterns a firm may consider, and to find some of the best ones. Second, conditional of all firms{'} export entry decisions,
we can solve for the resulting general equilibrium of world trade by accelerated gradient descent algorithms, the Adam algorithm in our case. For
this algorithm to be useful, we need to be able to calculate gradients of candidate solutions{'} loss functions (error functions) analytically. Because
of the scale of the problem, we do not perform the gradient calculation by hand. { }Instead, we rely on automatic analytic differentiation software,
namely the neural-network optimization package of PyTorch, which allows us to run all computations in a highly parallel fashion on graphics processing
units (GPUs).\footnote{PyTorch is an open source software framework developed by Facebook primarily
for deep learning in artificial neural networks. Its first version was released in 2017.}

 \subsection{Firm-level economies of scale in shipping: a generalized Economic Order Quantity
model}\label{SubsectionFirmLevelEconomiesOfScaleInShipping}

Most models of international trade assume that the costs of trade are of the {``}iceberg{''} type: a fraction of all goods transported is assumed
to be destroyed in transit. It seems implausible that most of true trade costs would scale with trade volume and value in this manner.\footnote{Tariffs would depend on value in the same way as iceberg trade costs, although the details of their impact would be different,
as goods are not destroyed and governments collect tariff revenue. That said, most trade costs modeled as iceberg trade costs in the literature are
not supposed to represent tariffs and we will not focus on tariffs in this paper, although, of course, they may be incorporated in our model.} A certain fraction of international trade papers, e.g.  \citet{melitzottaviano}, allow for constant marginal per unit costs of trade.\footnote{Per-unit
costs of trade seem more realistic than costs of trade proportional to the goods{'} value, as documented by, e.g., \citet{hummelsskiba}.
Note that this reference did not allow for non-linearity of trade costs.} However, the adoption of standardized
shipping containers has made such constant marginal per-unit costs of transportation extremely low relative to the trade costs necessary to explain
the rates of global trade flows.

We work with the assumption that most important trade costs come from coordination (shipment preparation) costs and inventory costs, which is why
the logistics literature focuses on them. These costs depend on the frequency of shipping. If a firm ships too infrequently, it will face large costs
associated with idle inventory. If it ships too frequently, shipment preparation costs will add up to a large number. Knowing this trade-off, the
firm will choose an optimal frequency of shipping that balances these two effects. The resulting effective cost of trade then exhibits economies
of scale: a firm wishing to ship only a small quantity on average per unit of time will find shipping to be costly per unit of quantity.

To gain empirical insight into the scale economies of international trade, we estimate a model of optimal shipping frequency using monthly international
shipment data. Our approach generalizes the classic Economic Order Quantity model of Ford W. Harris, which is widely taught in operations management
courses in business schools and applied by logistics planning managers in corporations.\footnote{Despite
not appearing in the international trade literature, the Economic Order Quantity (EOQ) model \citep{harris} is perhaps the most
classical model of trade costs in the operations research literature and is regularly taught to business students as a method of optimizing their
inventory decisions; see e.g. \citet{cardenas2014celebrating} for highlights of its importance. Judging from the absence of citations,
the academic international trade community is largely unaware of Harris{'} publication. When fixed costs per shipment are included in international
trade models, they are incorporated in theoretical models with different structures. Those models are similar in spirit, but do not strictly speaking
contain the EOQ model or its generalizations \citep{kropf2014fixed, hornok2015administrative}. These papers provide very useful
insights into shipment frequency issues, and so does the purely empirical paper \citet{hornok2015per}. Note also that economies
of scale in shipping were studied by \citet{anderson2014gravity} and \citet{forslid2016big}, but those approaches
were not based on shipping frequency and in the former case involved external (i.e., not within-firm) economies of scale.}
Consider a firm that produces a single good in one country and wishes to ship to a different country quantity \(q\) per year, on average. The firm
faces a tradeoff between inventory costs and coordination costs associated with frequent shipping. The average annual inventory cost \(C_i\) is linearly
proportional to \(q\) and to the time \(T\) a typical unit of the good needs to remain in storage. If the size of each shipment is \(q_s\), then
\(T\), in turn, is linearly proportional to \(\left.q_s\right/q\), implying \(C_i=\kappa _iq_s\), for some constant \(\kappa _i\). The coordination
cost \(C_s\) of each shipment is proportional to its size: \(C_s=\kappa _tq_s^{\gamma }\), $\gamma \in $[0,1). (In addition, we could assume an
additional term proportional to \(q_s\), but this would not affect the optimal choice of \(q_s\) for given \(q\).) The resulting average annual coordination
cost is \(C_t=C_sq\left/q_s\right.=\kappa _tq q_s^{\gamma -1}\). { }Minimizing the sum of the inventory cost and the coordination cost leads to
the optimal choice \(q_s=\left(q (1-\alpha ) \kappa _t \kappa _i^{-1}\right)^{\frac{1}{2-\gamma }}\)\(\text{}\), the minimized value \((2-\gamma
) (1-\gamma )^{-\frac{1-\gamma }{2-\gamma }} \kappa _i^{\frac{1-\gamma }{2-\gamma }} \kappa _t^{\frac{1}{2-\gamma }} q^{\frac{1}{2-\gamma }}\), 
 and the optimal frequency of shipping \(f_s=\)\(q\left/q_s\right.\) equal to\[f_s=(1-\alpha )^{-\frac{1}{2-\gamma }} \kappa _i^{\frac{1}{2-\gamma }} \kappa _t^{-\frac{1}{2-\gamma }} q^{\frac{1-\gamma }{2-\gamma }}.\]
 This result implies that we can infer the coordination cost exponent $\gamma $ by examining the relationship
between the average annual quantity shipped and the frequency of shipping. If we regress the logarithm of shipping frequency \(f_s\) on the logarithm
of average annual quantity \(q\), the resulting slope coefficient should equal \(\beta  \equiv  (1-\gamma )/(2 - \gamma )\). The model predicts that
this coefficient always lies between 0 and 1/2, since $\gamma \in $[0,1).

Our simple model of shipping frequency choice nests two important extreme cases. The original Economic Order Quantity model, in which the cost per
shipment is fixed, { }corresponds to \(\gamma =0\) and \(\beta =1/2\), implying effective cost of trade (here inventory and coordination) proportional
to \(\sqrt{q}\). The other extreme case has \(\gamma \to  1\) and \(\beta \to  0\) and corresponds to effective cost of trade linearly proportional
to \(q\), i.e., constant marginal cost of trade, as assumed in almost all of the international trade literature.

To estimate $\beta $ and to test the prediction that $\beta \in $(0,1/2], we used a dataset on monthly shipments from China to Japan during years
2000-2006. { }We focus on firms in one narrowly-defined product category.\footnote{We selected
firms by requiring that they specialize in one product category (one 8-digit HS code). The exporting firm had to be active for more than two years
to be included in our estimation sample. We selected industries that included at least 10 firms meeting these criteria, in order to work with industries
that allow for a precise estimate of \(\beta\). We were also careful to take into account potential effects of seasonality, which could affect our
estimates. We constructed a measure of seasonal variations of exports for individual industries. Our estimates of \(\beta\) did not differ almost
at all between industries with larger and smaller seasonality. We discuss more details in Appendix \ref{AppendixDetailsOfTheGeneralizedEOQModelEstimation}.} { }Our point estimate of $\beta $ (averaged across industries) is 0.39 with a 95$\%$ confidence interval of [0.36,0.42].\footnote{The confidence interval corresponds to a simple statistical model in which \(\beta\) for different industries is drawn from
a normal distribution.} { }We can thus clearly reject the null hypothesis that \(\gamma = 1\) and \(\beta
= 0\), which would correspond to trade costs being linearly proportional to quantity shipped, as assumed in the vast majority of the international
trade literature. We can also reject the original EOQ model, which would correspond to \(\gamma =0\) and \(\beta =1/2\). We see, however, that the
original EOQ model is closer to reality than the linear proportionality assumption. { }

In our main trade model, we round the resulting value for $\beta $ from 0.39 to 0.4. This estimate implies that increasing quantity by 10$\%$ reduces
(the variable component of) the marginal cost of trade by 4$\%$. { }We refer to the effective cost of trade as {``}cost of shipping{''}, remembering
that it arises from per-shipment coordination costs and from inventory costs with optimally chosen shipping frequency. In the rest of this section,
we use the notation \(\nu _{\text{LT}}\) for what was \(1-\beta =1/(2-\gamma )\) here, and set \(\nu _{\text{LT}}=0.6\). 

 \subsection{Export quantity determination}\label{SubsectionExportQuantityDetermination} 

For clarity of exposition, let us now consider the problem of export quantity determination for a firm that faces trade costs found in the previous
subsection. Similar ingredients will appear also in our main model described in Subsection  \ref{SubsectionModelSetup}. 

A firm considers exporting to one foreign country. If it delivers quantity \(q_f\) there, it will receive revenue \(R\left(q_f\right)\), for which
we choose the form { } \(R\left(q_f\right)=\nu _R^{-1}\kappa _R q_f^{\nu _R}\), where \(\nu _R=1-1/\sigma\) and \(\sigma =5\). The elasticity
of demand { }\(\sigma =5\) is consistent with the typical range in the trade literature. The firm faces an iceberg trade cost factor $\tau $, meaning
that it needs to send \(\tau  q_f\) in order for \(q_f\) to arrive. The shipping requires \(L_T\left(q_f\right)\) units of labor, which translates
into a cost \(w L_T\left(q_f\right)\). We choose \(L_T(q)=\nu _{\text{LT}}^{-1} \kappa _{\text{LT}} q^{\nu _{\text{LT}}}\), with \(\nu _{\text{LT}}=3/5\),
in agreement with the previous subsection. In this illustrative example, we assume constant marginal cost \(\text{MC}\). 

The second derivative of the profit function \(R\left(q_f\right)-\tau  \text{MC} q_f+w L_T\left(q_f\right)\) is \(\frac{2 }{5 }w \kappa _{\text{LT}}q_f^{-7/5}-\frac{1}{5
}\kappa _Rq_f^{-6/5}\), so the profit function is convex for \(q_f\in \left(0,32 w^5 \kappa _{\text{LT}}^5\kappa _R^{-5}\right)\) and concave for
\(q_f\in \left(32 w^5 \kappa _{\text{LT}}^5\kappa _R^{-5},\infty \right)\). To identify the maximum, we just need to find potential local maxima
in the second region and to check whether they are larger than zero. This is because at \(q_f=0\) the profit is zero, and as \(q_f\to  \infty\) it
goes to \(-\infty\).

The firm{'}s first-order condition is\[R'\left(q_f\right)- \tau  \text{MC}-w L_T'\left(q_f\right)=0\, \, \, \, \, \Longrightarrow \, \, \, \, \, -\frac{w \kappa _{\text{LT}}}{\tau }
q_f^{-\frac{2}{5}}+\frac{\kappa _R}{\tau } q_f^{-\frac{1}{5}}-\text{MC}=0.\]
We recognize that the function of \(q_f\) on the left-hand side is one of our proposed tractable functional forms. We
can, therefore, solve the first-order condition in closed-form, in this case using the quadratic formula. { }If \(\text{MC}> \kappa _R^2/\left(4
w \tau  \kappa _{\text{LT}}\right)\) { }there is no real solution and the firm will choose not to export. If \(\text{MC}\leq \text{  }\kappa _R^2/\left(4
w \tau  \kappa _{\text{LT}}\right)\), the solution that lies in the \(\left.\left[32 w^5 \kappa _{\text{LT}}^5\kappa _R^{-5},\infty \right.\right)\)
region equals\[q_f=\left(\frac{\kappa _R+\sqrt{\kappa _R^2-4 \tau  \text{MC} \kappa _{\text{LT}} w}}{2 \tau  \text{MC}}\right)^5.\]
Plugging this position of the local maximum into the profit function gives\[\frac{1}{192 \text{MC}^4 \tau ^4} \left(\kappa _R+\sqrt{\kappa _R^2-4 \text{MC} w \tau  \kappa _{\text{LT}}}\right)^3 \left(-16 \text{MC} w \tau
 \kappa _{\text{LT}}+3 \kappa _R \left(\kappa _R+\sqrt{\kappa _R^2-4 \text{MC} w \tau  \kappa _{\text{LT}}}\right)\right)\]
The first two factors are positive, and the last one is positive if and only if \(\text{MC} <15\kappa _R^2/\left(64 w
\tau  \kappa _{\text{LT}}\right)\approx 0.234 \kappa _R^2/\left(w \tau  \kappa _{\text{LT}}\right)\). If this condition is satisfied, the firm will
export the quantity satisfying the first-order condition, otherwise, it will export zero.\footnote{If
exporting leads to zero profit just like not exporting, we specify that the firm chooses not to export.}
Thus for any level of marginal cost, the quantity chosen by the firm may be written compactly as\[q_f=\left(\frac{\kappa _R+\sqrt{\kappa _R^2-4 \tau  \text{MC} \kappa _{\text{LT}} w}}{2 \tau  \text{MC}}\right)^51_{\text{MC}<\frac{15 \kappa
_R^2}{64 w \tau  \kappa _{\text{LT}}}},\]
where the second factor represents an indicator function. We see that the functional form allowed for a very simple and
straightforward analysis.
We also see that exporting may not be profitable even if there is no fixed cost of exporting. This implies that such model
with constant elasticity of demand can generate an export cutoff without fixed costs of exporting.

 In our main model described in Subsection  \ref{SubsectionModelSetup},
which no longer assumes that the marginal cost of production is constant, we still benefit from the closed-form characterization of the solution
to the first-order condition in terms of the level of marginal production cost. This is both for the evaluation of the solution and for taking derivatives
of the solution, as needed by gradient descent algorithms. Of course, the degree of the benefit grows in proportion to the number of potential export
destinations.

 \subsection{Increasing marginal cost of production}\label{SubsectionIncreasingMarginalCostOfProduction} 

Economies of scale, modeled using fixed costs of production, are present in most models of firms in the international trade literature. By contrast,
diseconomies of scale almost never appear in that literature. Yet there are many reasons to believe that increasing marginal costs of production
are similarly important in shaping firms{'} behavior. This is presumably why introductory economics classes frequently illustrate increasing marginal
cost schedules. Beyond short-to-medium term capacity constraints and adjustment costs usually discussed in such courses, even in the longer term
if a company decides to scale up its production an order of magnitude, it needs to introduce an additional layer of management hierarchy, which brings
with it non-trivial agency problems. In a large organization incentives are diluted, and maintaining motivation, discipline, and output quality becomes
harder.\footnote{The restaurant industry is an obvious example: few people would associate chain
restaurants with outstanding culinary experience. Another fairly obvious example is the automobile industry: there are many automakers in the world,
each having a relatively small market share, very stable over time, even though cars produced by different automakers are highly substitutable from
costumers{'} perspective. With constant marginal costs of production this would require a remarkably small dispersion of marginal costs across firms,
which is especially hard to rationalize given the large observed fluctuations of currency exchange rates. Also, the increasing nature of marginal
costs of production was one of the reasons why socialist economies were unsuccessful: state-controlled monopolies avoid duplication of effort in
product design and other fixed costs of production, yet they suffer from severe agency problems that private sector competition can mitigate. Although
here we emphasize increasing marginal costs of production in the long term, they are also interesting at short time scales; see \citet*{almunia2018venting}
and references therein.} Of course, managers of firms are intuitively aware of these problems, at least
to some extent, and take them to account when shaping the structure of the firm. 

Issues of this kind are the subject of interest to vast literature within organizational economics, which includes
\citet{williamson1967hierarchical},  \citet{calvo1978supervision}, and  \citet{tirole1986hierarchies}.\footnote{Oliver E. Williamson's Nobel lecture (\citet{williamsonnobellecture}) provides an excellent, compact discussion
of the many things that may go wrong in a large organization. For a related discussion, see \citet{tirole1988theory}.} 

Estimating how much marginal costs increase with production volume is non-trivial since both economies and diseconomies of scale play a role in firm
behavior. Our model provides a unique opportunity to obtain such estimates by matching firm-level multi-destination export data with world trade
model solutions.

 \subsection{Model setup}\label{SubsectionModelSetup}

Apart from the cost structure of the firms, our model is closely analogous to  \citet{melitz}, which many readers are familiar with. For this reason, we keep the description of the modeling setup succinct.

The world consists of { }\(N_c\) countries, indexed by \(k\). In each country, different varieties \(\omega\) of a differentiated good are produced
by monopolistically competitive heterogeneous single-product firms using a single factor of production, for simplicity referred to as labor.

 Consider a firm located in country \(k\) and identified by an index \(i\). In order to produce a quantity \(q_i\), the firm needs to pay a variable
cost of { }\(\frac{1}{1+\alpha } \kappa _{C,i} w_k q_i^{1+\alpha }\), where \(w_k\) is the competitive wage the firm{'}s country \(k\) and { }\(\kappa
_{C,i}\) is a positive constant that depends on the firm. Importantly, the constant $\alpha $ { }determines how { }quickly marginal costs increase
when any firm decides to scale up production; it is the elasticity of the marginal cost of production with respect to quantity. { }In addition to
the variable cost, there is a fixed cost \(f_o\) of operation and a fixed cost { }\(f_x\) of exporting to a destination country \(k_d\), expressed
in units of domestic labor.\footnote{In general, we can allow for country-dependence of these
costs: \(f_{o,k}\) and \(f_{x,k,k_d}\). We chose to make them country-independent for simplicity, not for tractability or computational feasibility
reasons.} { }

Entry into the industry is unrestricted, but involves a sunk cost of entry \(f_e\), again in units of domestic labor. Only after the entry cost has
been paid does the firm learn its variable cost parameter \(\kappa _{C,i}\), drawn from a distribution with cumulative distribution function \(\tilde{G}\left(\kappa
_C\right)\). When the value of \(\kappa _{C,i}\) is revealed, the firm decides whether or not to exit the industry, and if it does not exit, whether
to export to any of the other countries. In addition to endogenous exit, with a probability of \(\delta _e\) per period the firm is exogenously forced
to exit (starting from the end of the first period).

Trade costs have two components. The first corresponds to standard iceberg trade costs: in order of one unit of the good to arrive in the destination
country \(k_d\), \(\tau _{k,k_d}\) units need to be shipped.\footnote{Including also per-unit
trade costs would not affect the computational feasibility of the model.} The second component requires
using an amount of labor given by \(L_{T,k,k_d}(q)=\nu _{\text{LT}}^{-1} \kappa _{\text{LT},k,k_d} q^{\nu _{\text{LT}}}\), where we set \(\nu _{\text{LT}}=\frac{3}{5}\)
to be consistent with the empirical value, as in Subsection  \ref{SubsectionExportQuantityDetermination}.\footnote{The cost \(L_{T,k,k_d}(q)\) is associated with coordination/shipment
preparation tasks and with inventory costs. Its form is motivated by the empirical results of Subsection \ref{SubsectionFirmLevelEconomiesOfScaleInShipping}.}

Consumers in each country have a CES utility function \(U=(\int q_{\omega }^{1-\frac{1}{\sigma }}d\omega )^{\frac{\sigma }{\sigma -1}}\) that depends
on the quantity \(q_{\omega }\) of each variety $\omega $ consumed. As in Subsection  \ref{SubsectionExportQuantityDetermination}, we set the elasticity of substitution $\sigma $ equal to 5, which is consistent with the typical range in the existing
empirical literature of about 4 to 8. This exact choice is motivated by analytic tractability. Each country \(k\) has an endowment of labor \(L_{E,k}\),
which is supplied at a country-specific competitive wage rate \(w_k\) mentioned above.

The revenue a firm can earn by selling a quantity \(q\) in a given market is \(R_{k_d}(q)=\frac{\kappa _{R,k_d}}{\nu _R} q^{\nu _R}\), where \(\nu
_R=1-\frac{1}{\sigma }\). The factor \(\kappa _{R,k_d}\) is endogenously determined and depends on the price index and the consumption expenditures
in the destination country. 

The firm may choose to exit the industry (to save on the fixed cost \(f_o\)) or to operate and sell its product in a number of countries, earning
a non-negative profit $\pi $ per period of operation. In expectation, an entrant needs to break even: \(\delta _e f_e=\text{E$\pi $}\), which determines
the equilibrium measure of firms in each country.\footnote{The model has no explicit discounting
of future utility, but \(\delta _e\) plays a role similar to a discount rate.} Similarly, labor markets
in each country \(k\) need to clear, which means that the total labor demanded by firms at wage \(w_k\) needs to equal the labor endowment \(L_{E,k}\).
If we impose balanced budget conditions, consumers{'} expenditures equal their wage earnings, as firms earn zero ex-ante profits.\footnote{In our empirical setting we allow for budget imbalances that reflects similar imbalances in the data.}

 \subsection{The exporting firm{'}s problem}\label{SubsectionTheExportingFirmsProblem} 

Let us discuss the nature of the exporting firm{'}s problem. Increasing marginal costs will limit the scale of the firm{'}s production. Since trade
is subject to decreasing marginal costs, the firm will concentrate its exports into a limited number of countries. The overall production level \(q_i\)
of firm \(i\) as well as export market entry decisions are endogenous. For now let us consider the relation between of export quantities and \(q_i\),
conditional on having paid fixed costs of exporting to a number of countries.

The first-order condition for choosing the quantity \(q_{f,i,k_d}\) that should reach a foreign market \(k_d\) equates the marginal revenue and the
comprehensive marginal cost that depends on the overall production level \(q_i\):\[R_{k_d}'\left(q_{f,i,k_d}\right)=\tau _{k,k_d} \text{MC}_i\left(q_i\right)+w_k L_{T,k,k_d}'\left(q_{f,i,k_d}\right)\Rightarrow \frac{\kappa _{R,k_d}}{\tau
_{k,k_d}} q_{f,i,k_d}^{-\frac{1}{5}}=\text{MC}_i\left(q_i\right)+\frac{w_k \kappa _{\text{LT},k,k_d}}{\tau _{k,k_d}} q_{f,i,k_d}^{-\frac{2}{5}},\]
 in analogy with Subsection  \ref{SubsectionExportQuantityDetermination}. The solution for \(q_{f,i,k_d}\) given \(q_i\) is:\[q_{f,i,k_d}=\frac{1}{\left(2 \tau _{k,k_d} \text{MC}_i\left(q_i\right)\right)^5} \left(\kappa _{R,k_d}+\sqrt{\kappa _{R,k_d}^2-4 w_k \kappa _{\text{LT},k,k_d}
\tau _{k,k_d} \text{MC}_i\left(q_i\right)}\right)^5\]
 If the marginal cost of production \(\text{MC}_i\left(q_i\right)\) exceeds \(\kappa _{R,k_d}^2\)/(\(4 w_k \kappa _{\text{LT},k,k_d}
\tau _{k,k_d}\)), the first-order condition cannot be satisfied. For domestic sales we assume \(\kappa _{\text{LT},k,k}=0\) and \(\tau _{k,k}=1\),
so the optimal quantity sold domestically is simply \(q_{i,k}=\left(\kappa _{R,k}/\text{MC}_i\left(q_i\right)\right)^5\).

The total quantity \(q_i\) produced should equal the total of quantity sold domestically and sent abroad: \(q_i=q_{i,k}+\sum _{k_d\neq  k}\tau _{k,k_d}
q_{i,k,k_d}\), with \(q_{f,i,k_d}\) { }given by the formula above. This represents one equation for one unknown: \(q_i\). Each root of this equation
represents a candidate optimum for the firm.\footnote{Mathematically, the firm{'}s choice of
destinations in order to maximize profit is a submodular function maximization. This is because serving an additional set of markets \(A\) is less
attractive if the initial set of markets \(S_l\) is larger: \\
\(\pi \left(S_2\cup A\right)-\pi \left(S_2\right)\leq \pi \left(S_1\cup A\right)-\pi \left(S_1\right)\) for { }\(S_1\subseteq  S_2\) and \(A\cap
S_2=\emptyset\). Here \(\pi (S)\) denotes the optimal profit a firm can earn if it serves a set of markets \(S\). If instead our problem was supermodular
function maximization, it would be algorithmically easy. International trade papers such as \citet{antras2017margins} take advantage
of supermodular function maximization being straightforward.} The profit-maximizing choice(s) of destinations
may then be found by evaluating total profits at these candidate optima. For a small number of countries this is simple, but for large \(N_c\) the
problem becomes combinatorially difficult.\footnote{For an in-depth discussion of combinatorial
discrete choice problems in economics, see \citet{eckert2017combinatorial}. The method that \citeauthor{eckert2017combinatorial}
propose would be useful for us if the number of countries we consider were substantially smaller. This is because the method reduces the exponent
of an exponentially difficult problem, but does not change its exponential nature; submodular function maximization is NP-hard in general.} For this reason, when we solve the model for a large number of countries, we use approximate algorithms instead of an exhaustive
search.\footnote{We should clarify that even conditional on having made export fixed-cost payments,
the number of candidate optima is still combinatorially large. This is because for some destinations it may be impossible to satisfy the FOC and
in those cases we allow the firm to export zero amount there. To avoid this difficulty, { }when we consider candidate optima, we restrict attention
to those that satisfy a particular ordering condition, without loss of generality. We rank export destinations by the level of (constant) marginal
cost that would make them a profitable destination, in descending order. Then we require that if a firm exports a positive amount to a given destination,
it also exports to all preceding destinations. Imposing this condition is without loss of generality because if a firm decides to export zero amount
to a destination, it should not have paid the associated fixed cost of exporting in the first place.} 

 \subsection{Solution strategy}\label{SolutionStrategy}

We solve the model using an iterative algorithm that has an outer loop and an inner loop.\footnote{Due
to its combinatorial nature, the exact version of our model is computationally extremely difficult. It may seem natural to try to obtain approximate
solutions by first fixing aggregate variables in the model, solving for firm decisions given these aggregates, and then updating the aggregate variables
based on the firms{'} behavior. We attempted to do that, but could not get results within a reasonable amount of time and budget. This is because
for any values of aggregates, we needed to solve separate discrete choice problems by many firms, which requires a lot of time. For this reason,
we used a different nesting of loops: we moved all discrete choice decisions into an outer loop of an iterative algorithm, and given these discrete
choices, we solved for all continuous variables in an inner loop.} In the outer loop firms decide whether
or not they pay fixed costs of operation and fixed costs of exporting and commit to their decision. In the inner loop, we then solve for the general
equilibrium of the world economy given these fixed-cost decisions.

Finding this general equilibrium without the tractable functional forms is computationally difficult since a multi-level nested iteration is very
time-consuming. However, thanks to the analytic nature of our model, we were able to obtain the general equilibrium much faster using accelerated
gradient descent in a space parametrized by quantities \(q\), wages \(w\), measures of firms \(M\), price-index related variables \(\kappa _R\),
and country-level expenditures \(E\). We used the Adam optimizer of  \citet{kingma2014adam}, as implemented in PyTorch, a neural network optimization software for GPU computing.\footnote{We tried several accelerated and non-accelerated gradient descent algorithms. Adam performed the best.}
The gradients are computed analytically by automatic differentiation (autograd, in this case) and backpropagation.\footnote{Our model, as detailed in the next subsection, had more than 20,000 variables and described 2 million potential trade flows.
Newton{'}s method would not be feasible here, given that the Hessian of the loss function would have 400 million entries, although light-weight second
order methods, such as L-BFGS, could potentially be useful. They would again benefit from the analytic nature of our model. For an overview of optimization
algorithms, see the excellent book by \citet*{goodfellow2016deep}.}\(\, ^,\,\)\footnote{Given firms{'} sunk cost decisions, we need to solve for the general equilibrium of the world economy, i.e., we need to solve
for wages, price indexes, and the measure of firms of each type in each country, as well as for production levels of each firm. What makes our calculation
fast is the fact that we have explicit formulas for quantities sent to individual destinations conditional on the firm{'}s marginal cost, and that
these formulas and their derivatives may be evaluated extremely fast.}

Given a solution to the general equilibrium problem, we then let firms reconsider their fixed cost payments. For numerical stability, we do not update
at once the fixed cost payment decisions of all firms. Instead, for each productivity level in a country we introduce \(N_v=10\) versions (copies)
of firms, which can differ by their fixed cost commitments. Updating fixed-cost commitments then proceeds in cohorts. In one iteration of the outer
loop, version 1 firms will be able to reconsider the fixed cost payment. In the second iteration, version 2 firms will do so, etc. Keeping different
versions of firms comes at a computational cost, of course, but we found this necessary.

Finding the best fixed cost decision is a combinatorially difficult problem. Given that there are \(N_c-1\) potential export destinations, this leads
to \(2^{N_c-1}\) possibilities for the exports. With \(N_c=100\), this is more than \(10^{29}\). To obtain an approximate optimum, we use Algorithm
{ }2 of  \citet{buchbinder2015tight}, which is stochastic
in nature. { }We consider 9 (random) runs of that algorithm, and if the best of them is better than the firm{'}s previous fixed cost decision, we
update it. After the update, we again solve for a new general equilibrium involving continuous variables.

 \subsection{Fitting the model}\label{FittingTheModel}

We work with { }\(N_c=100\) countries. { }This choice is motivated by data availability and parameter fit considerations: For a substantially larger
number of countries, the trade data would be too noisy and unreliable. For a substantially smaller number, it would be impossible to read off the
elasticity of the marginal production cost from the firms{'} export pattern using our method.

The labor endowment in the model is interpreted { }as an efficiency-adjusted number of units of a single production factor, which in practice would
include labor, capital, and the related productivity. This effective labor endowment and the trade cost prefactors { }\(\kappa _{\text{LT},k,k_d}\)
are recovered by fitting the model to data on country GDP and world trade flows for the year 2006, as described below.

 To match the typical empirical firm size distribution, which we take as Pareto distribution with Pareto index \(\mu _R=1.05\), we choose { }the
firm size distribution to be another Pareto distribution with Pareto index \(\left.\mu _R(\sigma -1)\right/(1+\sigma  \alpha )\).\footnote{The value of 1.05 for the Pareto index of the firm size distribution has empirical support in \citet{aoyama2010econophysics},
at least for the advanced economies studied there. In our open-economy model, there is no simple closed-form expression for the firm size (revenue)
as a function of firm productivity. For this reason we use a formula that would hold exactly for closed economies, as well as in the absence of trade
costs for the world. Simple algebra shows that the required value of the { }productivity Pareto index is { }\(\left.\mu _R(\sigma -1)\right/(1+\sigma
 \alpha )\) and we use this value. The calibration results in a good match for the firm size distribution of Chinese exporters (both for all firms
and for single-HSID firms, as these have the same empirical shape). Given this encouraging result, we have not explored other specifications for
the productivity distribution.} The productivity distribution is the same for every country in the model;
any real-world overall firm productivity differences across countries are represented by adjustments to the countries{'} effective labor endowments.

 For computational purposes we discretize the productivity distribution to \(N_p=20\) discrete values, each representing the same probability mass.\footnote{Initially, we tried \(N_p=10\), but such crude discretization led to numerical errors that were too large. Also note that
even though for simplicity we sometimes refer to the probability masses as {``}firms{''}, they really represent collections of firms in monopolistic
competition, not a few discrete firms in an oligopoly model.} 

In addition, we need to specify (the flow value of) the costs of entry, fixed costs of production, export market entry costs, as well as iceberg
trade costs. We make these choices as simple as possible, independent of the country or country pair. Their values are given in Table 
\ref{tractableforms2}. The flow value of the cost of entry is set to one half of the fixed
cost of operation. The fixed cost of exporting is set to be negligible. The iceberg trade cost \(\tau -1\) is non-zero but small enough to be consistent
with prices firms in practice pay for insurance or as tariffs. In general, the parameters are chosen to reflect a long-term interpretation of the
model, with timescales of many years.\footnote{We do not attempt to model high-frequency phenomena
in international trade (except that shipping frequency considerations provide a micro-foundation for our trade costs). For studying month-to-month
or year-to-years changes, it would not be appropriate to assume that the sunk fixed cost of exporting is fairly negligible.}

 \begin{table}\begin{center}
{\begin{tabular}{|c|c|c|c|c|c|c|c|c|c|}
\hline \(N_c\) & 100 & \(N_p\) & 20 & \(N_v\) & 10 { }& \(\alpha\) & multiple values\\
\hline \(\sigma\) & 5 { }& \(\nu _R\) & 0.8 &\(\nu _{\text{LT}}\) & 0.6 & \(\mu _R\) & 1.05 \\
\hline  \(\delta _e\)\(f_e\)   & 0.05 & \(f_o\)
& 0.1 { }& \(f_x\) & \(10^{-5}\) & \(\tau\) & 1.05 \\
\hline \end{tabular}}
\caption{Calibration parameter values }
\label{tractableforms2}
\end{center}\end{table}

We use importer-reported data on international trade flows for the year 2006 from the UN Comtrade database. We select 100 countries/economies with
the largest GDP, as reported by the IMF in its World Economic Outlook database, subject to trade and GDP data availability. We adjust the countries{'}
GDP for tradability using the United Nations{'} gross value added database; see Appendix  \ref{AppendixWorldTradeFlows}.

 \subsection{Elasticity of the marginal cost of production}

We solve for the model fit for different values of the parameter $\alpha $, the quantity-elasticity of the marginal cost of production. Then we compare
the resulting pattern of firm trade with that of Chinese firm-level export data for 2006 in order to find what value of $\alpha $ leads to a good
agreement. 

We obtained fits to the data on world trade flows and adjusted GDP levels for values of $\alpha $ ranging from 0.15 to 0.3; see Figure 
\ref{FigureMedianSizeOfChineseFirmByDestinationCountry}. In each case, we computed power-law
best-fit curves that describe the dependence of the median size of Chinese firms that export to a destination as a function of the popularity rank
of that export destination.\footnote{More precisely, we use the generalized method of moments
to fit functions of the form \(c_0\, (\text{rank})^{-c_1}+c_2\).} The popularity is computed as the fraction
of Chinese firms in the data that choose to export to the given destination. We also computed such best-fit curve for the data. The results are intuitive:
For smaller \(\alpha\), the difference between the median (log) size of firms exporting to unpopular destinations and to popular destinations is
larger because in this case the most productive firms will dominate world trade, and if a less productive firm decides to export at all, it will
choose a few of the popular destinations.

The data corresponds roughly to \(\alpha \approx 0.225\) if we consider all 99 export destinations when computing the best-fit curves, or to \(\alpha
\approx 0.25\), if we consider the first third of them by popularity rank. { }The first estimate has the advantage of taking into account a large
range of export destinations. We include the second estimate because the top third of the destinations account for a vast majority of Chinese export
and because the model{'}s precision is lower for very small countries. The values \(\alpha \approx 0.225\) or { }\(\alpha \approx 0.25\) would imply
that if a firm decides to scale up production by an order of magnitude, its marginal cost increases by about 68$\%$ or 78$\%$, respectively. { }These
values seem very realistic, given that such a dramatic expansion of the firm would require an additional layer of management hierarchy with related
principal-agent problems. Note that these inefficiencies would be partially offset by savings on the fixed cost of production.

 \begin{figure}[t]\begin{center}\includegraphics[width=12cm]{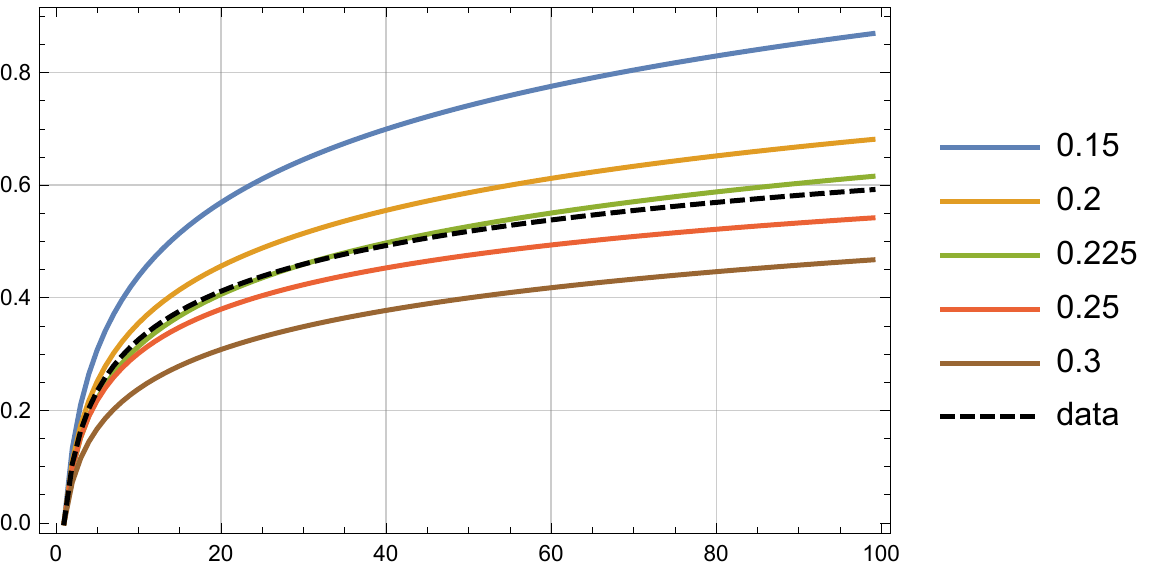}
\includegraphics[width=12cm]{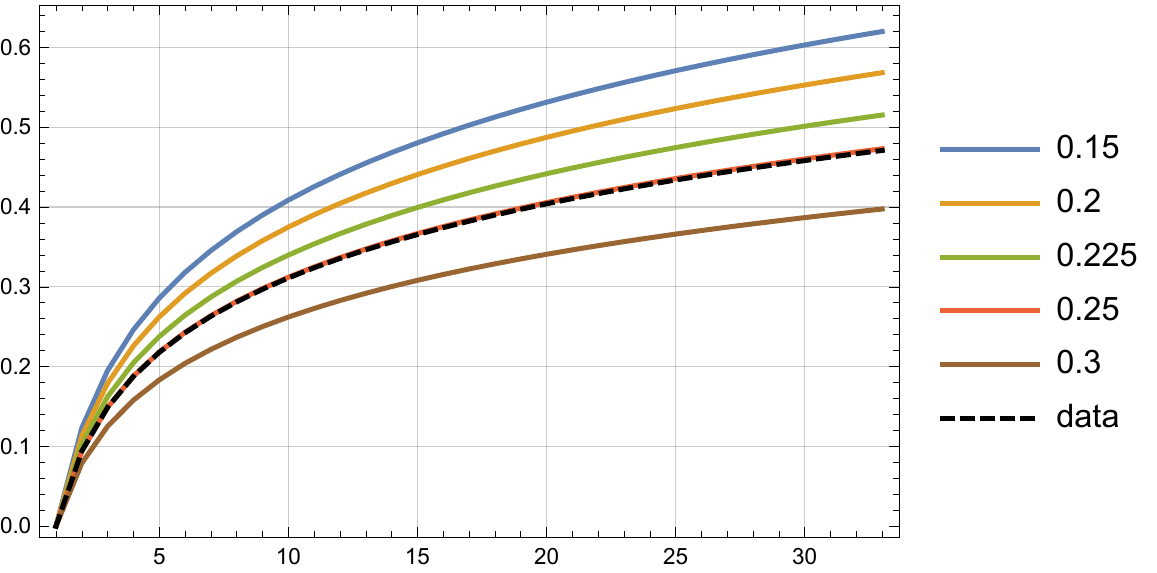}\end{center}
 \caption{Median revenue of Chinese exporting firm (base-10 log scale) by destination country
for different values of \(\alpha\) and for the observed values. For the first export destination (United States),
the median log revenue is normalized to 0 to make visual comparisons easier. The top figure corresponds to all 99 export destinations in the model,
while the bottom figure corresponds to the top third by export popularity.}\label{FigureMedianSizeOfChineseFirmByDestinationCountry}
 \end{figure}

 \subsection{The gravity equation of trade and the dependence of trade costs on distance}  \label{TheGravityEquationOfTradeAndTheDependenceOfTradeCostsOnDistance} 

The model fit results have important implications for the gravity equation of trade and for the trade cost puzzle (discussed in detail by 
\citet{disdier2008puzzling} and  \citet{head2013separates}).\footnote{See \citet{head2014gravity} for a recent
overview of the literature on the gravity equation of trade. Our purpose here is to highlight the consequences of our model{'}s mechanisms, so we
focus on the baseline gravity equation that describes the dependence of trade flows on distance and effective GDPs of countries. A comprehensive,
in-depth investigation of our model that parallels detailed studies in the gravity-equation literature will be reported separately. It is, of course,
worth investigating gravity equations with added controls, such as common language. Similarly, it is good to account for the {``}multilateral resistance{''}
phenomenon (i.e.more isolated countries being more eager to trade with a given partner) already when designing the regression/estimation equations
to study. Our model provides very different structural equations than other models, so the matter of multilateral resistance is quite involved. In
addition, it is good to explicitly consider trade flow zeros in constructing the regression/estimation equations, although that makes little difference
here as almost all trade flows are non-zero in our sample of 100 economies.} { }The gravity equation of
trade implied by the data\footnote{This is for 30 largest economies. For 100 economies the results
would be more noisy.} \[\log  x_{i j}\approx -0.77 \log  d_{i j}+1.12 \log  y_i+1.10 \log  y_j+\text{const.}\]
 matches well the gravity equation implied by the fitted model\footnote{Here
we used \(\alpha =0.225\).} \[\log  x_{i j}\approx -0.71 \log  d_{i j}+1.08 \log  y_i+1.02 \log  y_j+\text{const.}\]
 Of course, this is not surprising given that the world trade flows were a target of our model fit; if the fit was perfect,
the two equations would coincide. What is interesting, though, is that the trade cost prefactors (i.e. the factors \(\kappa _{\text{LT},k,k_d}\)
in \(L_{T,k,k_d}(q)=\nu _{\text{LT}}^{-1} \kappa _{\text{LT},k,k_d} q^{\nu _{\text{LT}}}\)) depend on distance only very weakly:\[\log  \kappa _{\text{LT}}\approx 0.048 \log  d_{i j}+0.032 \log  y_i+0.048 \log  y_j+\text{const.}\]
 We see that trade flows decrease rapidly with distance despite only a very mild increase of trade cost prefactors \(\kappa
_{\text{LT}}\) with distance. Although this may look surprising at first sight, there is clear intuition for this phenomenon: Due to increasing marginal
costs of production, firms effectively have only a limited output to sell and due to economies of scale in shipping, they need to concentrate their
exports to only a few countries. They choose close countries because the trade costs are slightly lower, which leads to strong effects for the decrease
of trade with distance.\footnote{We briefly discuss related literature and mechanisms in Appendix
\ref{AppendixWorldTradeFlowsRelatedLiterature}.\label{FootnotePointingToAppendixWorldTradeFlowsRelatedLiterature}} 

From the histogram in Figure  \ref{HistogramOfLogTradeCostPrefactors}
we see that the dispersion of \(\kappa _{\text{LT}}\) is very small, which is only possible with a very small dependence on distance. This small
dispersion is very much consistent with sea shipping over large distances being only mildly more expensive than over short distances. This provides
a very natural resolution to the trade cost puzzle in the international trade literature.

 \begin{figure}[t] \begin{center} \includegraphics[width=4.0in]{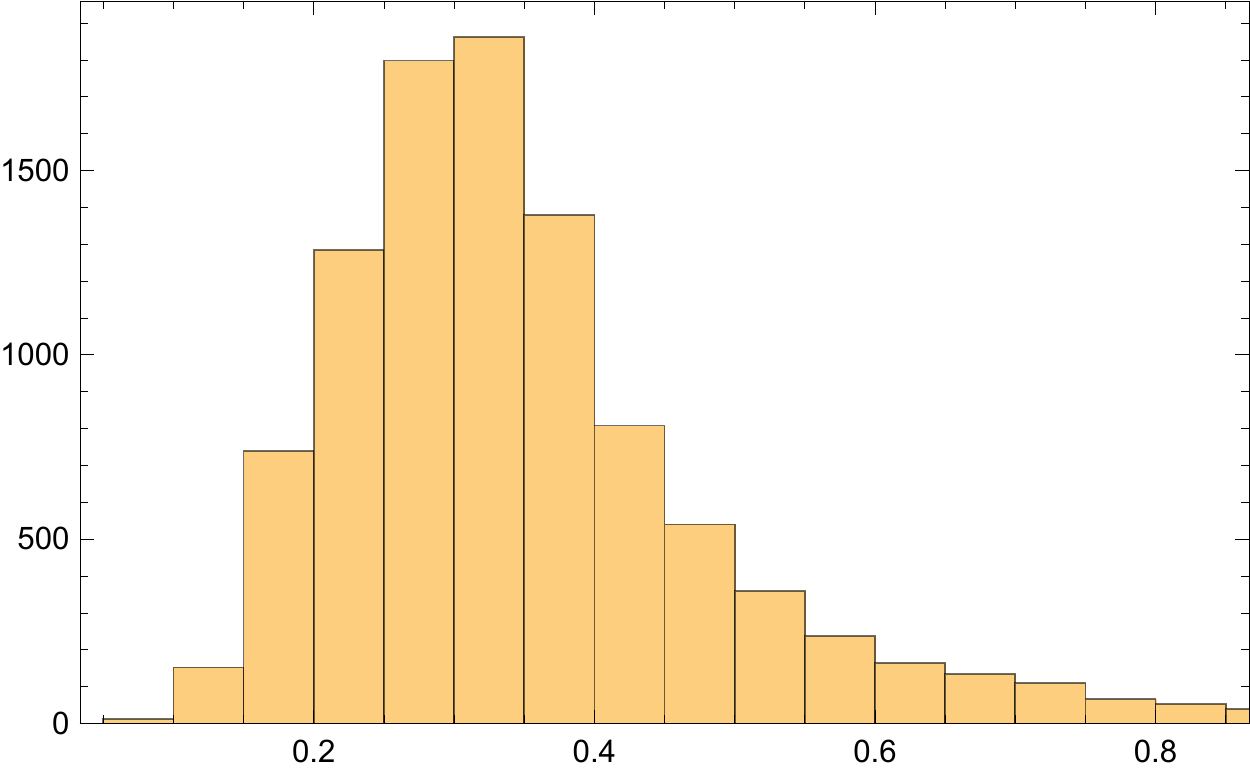}
\end{center}
 \caption{Histogram of the (symmetrized) trade cost prefactors.}\label{HistogramOfLogTradeCostPrefactors}
 \end{figure}

 \subsection{Choice of export destinations}\label{ChoiceOfExportDestinations} 

In the Introduction, we briefly discussed an empirical pattern of firm entry into export markets that would seem puzzling in standard models of international
trade. Our model naturally implies such pattern. Figure  \ref{FigureExportDestinationsByTwoIdenticalFirms} illustrates export market entry choices in the fitted model for pairs of identical firms, i.e. firms from the same country
and with the same productivity.\footnote{The choice of countries for the figure is not completely
arbitrary. China was chosen for the figure because it is a large country and we see patterns of this kind in its firm-level data. We chose the Czech
Republic since it is a small country with many neighbors and we know of patterns of this kind based on a series of interviews with Czech exporters
featured in \textit{ Hospodarske noviny}, a newspaper.} These would be impossible in a corresponding model
with constant marginal cost unless we introduced unrealistically large firm-destination specific cost shocks (or other similar shocks).\footnote{Although for identical firms this would be impossible if we introduce differences between the firms, there are other phenomena
that may play a role. We briefly discuss them in Appendix \ref{AppendixWorldTradeFlowsFirmExportPatterns}.\label{FootnotePointingToAppendixWorldTradeFlowsFirmExportPatterns}} It is straightforward to see why this is the case. With constant marginal costs, the decision of whether or not to enter
a particular export destination is independent of such decisions for other destinations, as long as the firm does not shut down. If there were no
firm-destination specific shocks, then two identical firms with the same constant marginal costs would reach the same conclusions about the profitability
of each export destination. In order to make one of the firms give up on a particular destination, we would have to introduce a firm-destination
specific shock that would offset all the profit the firm was about to make from selling at that destination.

Our model naturally delivers the export destination choice pattern that would seem puzzling otherwise. With increasing marginal costs of production,
a destination that is profitable for one firm may not be profitable for another identical firm, if that firm already serves other locations. Of course,
if there were no economies of scale in shipping (and no significant export entry fixed costs), firms would dilute their exports over more destinations
and would not face a combinatorial discrete choice problem. In that case, two identical firms would serve the same destinations, unless, again, there
were firm-destination specific shocks. For this reason, both increasing marginal costs of production and economies of scale in shipping are crucial
for our model{'}s ability to resolve the export destination choice puzzle.

A mechanism of this kind also has the potential to explain why personal relationships can play a large role in international trade. Just like shorter
distance, knowing someone trustworthy to cooperate with at a potential export destination can provide a mild profit advantage for exporting there
instead of other destinations. This modest advantage may then have a large effect on trade flows, given increasing marginal costs of production and
economies of scale in shipping.\footnote{Similarly, the mechanism may help explain why in the
long run trade liberalization can have dramatic effects on trade flows, as for example in the case of the 2001 US-Vietnam trade liberalization; see
\citet{mccaig2018export}.}

The subject of the export destination choice pattern is, of course, very rich and calls for an in-depth empirical and theoretical investigation,
which will be provided in a separate, monothematic paper.

 \begin{figure}[t]\begin{center}\includegraphics[width=12cm]{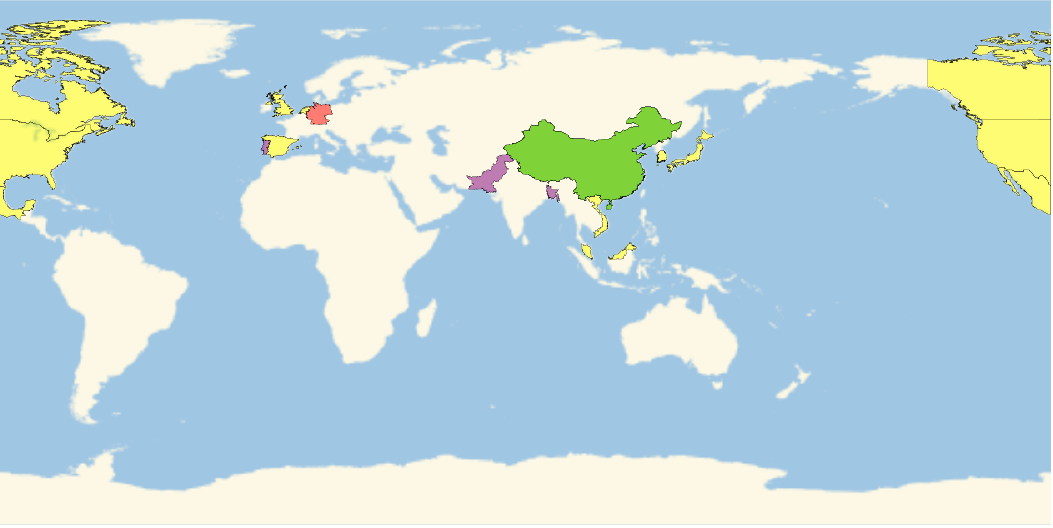}
 \includegraphics[width=12cm]{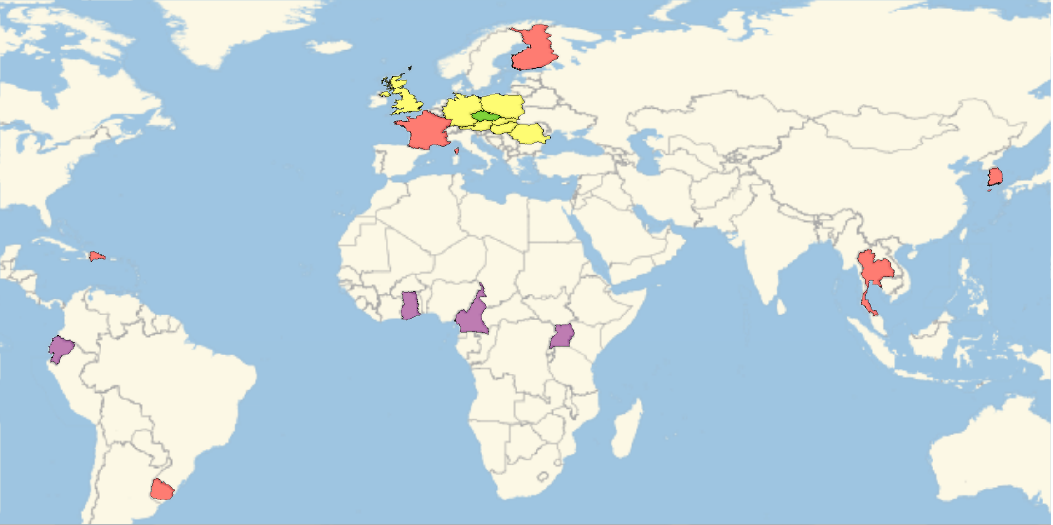}
 \end{center}
 \caption{The top map highlights export destinations of two Chinese firms in the fitted model.
Specifically, for two identical firms the map shows destinations to which both firms export (yellow), destinations to which only firm 1 exports (red),
destinations to which only firm 2 exports (purple), and the country of origin (green). The bottom map shows similar information for two firms from
the Czech Republic.}
 \label{FigureExportDestinationsByTwoIdenticalFirms}
 \end{figure}

 \subsection{Implications for modeling international trade}\label{SubsectionImplicationsForModelingInternationalTrade} 

A vast majority of models of international trade (and spatial economics) assume constant marginal cost of production within firm, even though empirical
evidence for such constancy is lacking and even though organizational economics is telling us that scaling up a firm is highly nontrivial, given
all the internal agency problems. { }An obvious reason for making the assumption of constant marginal cost is that it decouples firms{'} behavior
in different export destinations and makes the models easy to solve. Without such decoupling we need to solve combinatorial discrete choice problems,
which are hard in the case of submodular function maximization (corresponding to increasing marginal cost).

We have seen that working directly with increasing marginal costs leads to a dramatically different perspective on quantitative and qualitative aspects
of international trade. It is computationally challenging, but the results are worth it. Some of the puzzles are no longer puzzling, as trade costs
behave the way we would expect. { }

In the future, working with trade models that impose constant marginal production costs will not be as appealing to us as before. It suddenly has
a flavor of the proverbial searching for keys under a streetlight. Once we accept the idea of working with models with increasing marginal costs,
there are many questions to address. It would be good to re-think many topics in international trade, such as the impact of various policies, interventions
or technological changes on the equilibria and on the welfare of different agents in the economy.\footnote{For example, for welfare consequences we can no longer use simple, elegant formulas such as those derived by \citet*{arkolakis2012new}.} The model we worked with is very parsimonious, but including multiple locations per country, multiple sectors, supply chains,
and/or foreign direct investment would be desirable. Some of these ingredients would bring their own combinatorial discrete choice problems. The
models will certainly be even more computationally intensive than the one we worked with. Improvements in algorithms and hardware, hopefully, will
make solving the models feasible.

 \section{Breadth of Application}\label{SectionBreadthOfApplication} 

 \subsection{Overview of applications} 

In this section we provide a brief overview of numerous other applications.

 \subsubsection{Supply chains with hold-up \citep{antras}\label{SubsectionSupplyChainsWIthHoldup}}

 \citet{antras} develop a model of continuum sequential
supply chains where a main firm organizing its production needs to decide whether to outsource or insource (i.e. perform in-house) each stage of
the production process. Production requires relationship-specific investment, which leads to a hold-up problem in the spirit of { }
\citet{grossmanhart}. Outsourcing a production stage has the advantage of giving high-powered
incentives to the producers, while insourcing has the advantage of mitigating the hold-up problem.

The paper works with constant-elasticity demand and concludes that there can be only one production stage at which the main firm switches production
mode; depending on the parameter values, either all of the upstream or all of the downstream (but not both) is outsourced and the rest is insourced.
This, of course, clashes with the fact that for many manufacturing supply chains both the upstream (say, elementary components) and the downstream
(say, retail) are outsourced, while the core of the production process is insourced.

In Appendix  \ref{AppendixSupplyChainsWithHoldup} and
in Supplementary Material  \ref{AppendixSupplyChainsWithHoldupDetails},
we introduce a transformation of economic variables that makes the mathematics of the model dramatically simpler, in particular connecting the analysis
to the classical monopsony problem, whose cost-side aspects are analogous to the demand-side aspects of a monopoly problem.\footnote{Using our transformed variables would have saved at least 10 pages of the original paper \citet{antras}.
But of course, relative to these authors we have the benefit of hindsight.} This allows us to observe by
insights analogous to ours above that constant-elasticity demand may be replaced by our tractable functional forms without almost any loss of analytic
power. We find that for a realistic functional form of this kind (where the product has a {``}saturation point{''} in terms of quality), the model
implies that both upstream and downstream parts of the supply chain are optimally outsourced, while the middle (core) of the supply chain is optimally
insourced, as our intuition suggests in many real-world cases.

 \subsubsection{Labor bargaining without commitment (\citeauthor{stole}) \label{SubsectionLaborBargainingWithoutCommitment}}
Stole-Zwiebel bargaining, as introduced in  \citet{stole,stole2},
has become one of the standard ways of modeling labor bargaining in relation to unemployment. In their model, if an employee leaves the firm after
an unsuccessful wage bargaining, the remaining employees may renegotiate their wage, and they will choose to do so since the firm{'}s bargaining
position is weakened. For this reason, the firm will choose to employ more workers than it would if labor markets are competitive; employing an additional
worker lowers negotiated wages for the others. 

While this model appears to differ from previous examples we considered, as it is not a straightforward monopsony model, we show that behavior under
the Stole-Zwiebel model corresponds to a {``}partial{''} application of the marginal-average transformation ({``}partial monopolization{''}) and
thus remains tractable under our forms. Thus it is common to use standard, form-preserving tractable forms to analyze this model, especially constant-elasticity.
The downside of the assumption is that interesting effects are suppressed: the overemployment ratio (ratio of actual employment and employment under
competitive labor markets) is a constant independent of economic conditions. 

We introduce richer functional forms that preserve the tractability of the model. We find that for a plausible parameterization, changes in the overemployment
ratio can account for a non-trivial fraction of employment changes over the business cycle. These results are discussed in Supplementary Material
 \ref{AppendixLaborBargainingWithoutCommitment}.

 \subsubsection{Imperfectly competitive supply chains
\label{SubsectionImperfectlyCompetitiveSupplyChains}}

Imperfectly competitive supply chains, as described in  \citet{salinger},
are a very natural and popular way of modeling multi-stage production. We find that models of this kind may be solved in closed form not only for
linear or constant-elasticity demand but also for our proposed, much more flexible functional forms. Intuitively, behavior at each level of the supply
chain is derived by applying the marginal-average transformation to behavior at the preceding level, as each step of the supply chain forms the demand
for the level above it. We discuss this application in Appendix { } \ref{sub:Imperfectly-competitive-supply} and provide the details in Supplementary Material  \ref{sequential}.

 \subsubsection{Two-sided platforms {\` a} la \citet{rt2003}\label{SubsectionTwoSidedPlatforms}}

 \citet{rt2003} developed a model of two-sided platforms
that allows for understanding pricing decisions for the two sides of the market and their surplus consequences. The model used linear demand. We
find that our more flexible functional forms preserve the tractability of the model. This can lead to very different conclusions, as discussed in
Supplementary Material  \ref{AppendixTwoSidedPlatforms}.

 \subsubsection{Auction Theory\label{SubsectionAuctionTheory}}

 \paragraph{Symmetric independent private values first-price auctions.}\label{SymmetricIndependentPrivateValuesFirstPriceAuctions}

First price auctions with symmetric independent private values may be solved explicitly for uniform or Pareto value distributions. We find that the
tractable functional forms we propose still lead to closed-form solutions, and at the same time they allow for more realistic (i.e. bell-shaped)
value distributions. We discuss these results in Supplementary Material  \ref{AppendixSymmetricIndependentPrivateValuesFirstPriceAuctions}.

 \paragraph{Auctions v. posted prices \citep*{onlineauctions}.} \label{AuctionsVPostedPrices}

 \citet{onlineauctions} develop a model in which online
sellers choose either auctions or posted prices. They use a uniform distribution in their model. We find that our proposed functional forms also
lead to tractable models but allow a richer set of possibilities for the sellers{'} optimal behavior that better match the data. We explain the details
in Supplementary Material  \ref{AppendixAuctionsVPostedPrices}.

 \subsubsection{Selection markets}
\label{SelectionMarkets}

In selection markets (markets with adverse or advantageous selection) as in  \citet{imperfectcomp}{'}s generalization of  \citet{einav2010estimating}
and  \citet{einav2011selection}, the equilibrium conditions
are such that again our proposed tractable functional forms lead to closed-form solutions. This allows for modeling possibilities that provide a
better match to the empirical evidence, as explained in Supplementary Material  \ref{AppendixSelectionMarkets}.

 \subsubsection{Monopolistic competition}
\label{SelectionMarkets}

Tractable functional forms are very useful in the case of monopolistic competition beyond what we discussed in the previous section. Supplementary
Material  \ref{AppendixMonopolisticCompetition} contains
an extensive discussion of other possible modeling choices that generalize, say, the Melitz model or the Krugman model. These calculations may be
used as a basis for new research projects on international trade.

\section{General Approximation and the Laplace-Log Transform\label{SectionArbitraryDemandAndCostFunctions}}

In most of the examples in the previous sections,
we have focused on average-marginal form-preserving classes of relatively low dimensions that are tractable at low orders. While these
are useful in many applications and reasonably flexible, they
have limits in their ability to fit arbitrary equilibrium systems.
In this section we show that this limitation arises from the desired tractability
of these forms, rather than any underlying rigidity of our average-marginal form-preserving classes. 
Under weak conditions we formulate here, arbitrary (univariate) equilibrium
forms can be approximated arbitrarily well by members of form-preserving classes.
The
limit of this approximation is the inverse Laplace-log transform of the equilibrium condition. Highly tractable
forms may thus be seen as ones with ``simple'' inverse Laplace-log transforms. We
show how the special, policy-relevant features of many common demand
forms can be characterized in terms of their transforms. Proofs
of the theorems in this, more abstract, section appear in Appendix \ref{AppendixProofsOfTheorems}. 
A number of these proofs are straightforward adaptations of theorems in the existing
literature. We include those theorems here for completeness and for the reader's convenience.

In the next subsection we provide definitions of the Laplace-log transform, utilizing existing mathematical
literature. Identifying the most important connections between what is useful in economics and
the mathematical literature is non-trivial. While a reasonable number of economists are familiar with Laplace transform
based on the Riemann-Stieljes integral, a theory based on that integral would exclude, say, the exponential
demand function, which is a popular modeling choice in the economics literature. For a more complete theory
we need to utilize the distribution theory by Laurent Schwartz, which has not been used in economics.

\subsection{The Laplace-log transform and arbitrary approximation}

Under quite general conditions, univariate equilibrium conditions may
be expressed as linear combinations of average-marginal form-preserving
functions. To make this statement precise, we focus
on the demand side here and write an inverse demand curve of interest
as $P\left(q\right)=U'\left(q\right)$, where $U(q)$ is a function primitive
to $P\left(q\right)$. We assume that $P\left(q\right)$ is non-increasing,
which implies that such primitive function exists. Depending on the
model of choice, $U(q)$ may or may not be proportional to the utility
of an agent, but to keep the terminology simple, here we refer to
$U\left(q\right)$ as the utility.\footnote{$U\left(q\right)$ would literally be a term in the utility function
$U\left(q\right)+\tilde{q}\tilde{P}$ in a model with two goods $q$
and $\tilde{q}$, where $\tilde{q}$ is treated as a num\' eraire good
with price $\tilde{P}$ normalized to 1. In this case the marginal
utility of $q$ equals its price $P\left(q\right)$.} Even though we explicitly discuss the demand side here, the mathematical
theorems below apply to the cost side as well, with a straightforward
reinterpretation. 

\label{completeness}We observe that virtually all shapes of demand
functions that are useful in economics may be associated with a utility
function of the form\footnote{The Laplace-log representation (\ref{eq:UtilityAsASymbolicIntegral})
of a given utility function $U\left(q\right)$ exists under various
conditions. Theorem 18b in Section VII.18 of \citet{widder2010laplace}
states general necessary and sufficient conditions on $U\left(q\right)$
for the existence of $u_{I}\left(t\right)$ such that (\ref{eq:UtilityAsRiemannStieljesIntegral})
is satisfied; \emph{almost all }utility functions we may encounter
in economic applications \emph{do satisfy these conditions}. Sections
VII.12-17 of \citet{widder2010laplace} provide conditions that guarantee
that $u_{I}\left(t\right)$ exists and has certain properties, such
as being of bounded variation, nondecreasing, or belonging to the
functional space $L^{p}$. Additional conditions may be found in Chapter
2 of the book by \citet{arendt2011vector}, which contains recent
developments in the theory. In situations when utility unbounded below
is desired, e.g. for constant demand elasticity smaller than 1, we
can depart from (\ref{eq:UtilityAsASymbolicIntegral}) and instead
use the bilateral specification $U\left(q\right)=\intop_{-\infty}^{\infty}u\left(t\right)q^{-t}dt$.
However this generalization requires the use of more technically involved
bilateral Laplace transforms and thus we do not discuss it in greater
detail here, though analogous results are available on request.}

\begin{equation}
U\left(q\right)=\intop_{-\infty}^{0}u\left(t\right)q^{-t}dt,\label{eq:UtilityAsASymbolicIntegral}
\end{equation}
for an appropriate $u\left(t\right)$, where we work on some arbitrarily
chosen finite interval $\left[0,\bar{q}\right]$. This integral may
be interpreted as a Laplace transform in terms of the variable $s\equiv\log q$, and for this reason, we refer to $u(t)$ as the inverse Laplace-log transform of $U(q)$.\footnote{Our use of $t$ for exponents throughout the text and our use of $s\equiv\log(q)$
here match the standard notation in the literature on Laplace transforms.}$^,$\footnote{
After an extensive literature search of hundreds of articles and talking
to numerous economists, including highly accomplished econometricians,
we concluded that this is almost certainly the first time  (inverse) Laplace
transform in log quantity is used in the economic literature. Note, however, that a
different transform, namely (inverse) Laplace transform in quantity, as opposed to log quantity, has been used in economics. These transforms have different properties and should not be confused. Note also that the way we use Laplace transform is different from, say, engineering fields in the sense that, because of the additional logarithm, functions of main interest for us in economics typically would not be of interest in engineering, and vice versa. For this reason, books containing detailed tables of Laplace transform were not of help to us. Except for trivial cases, we needed to derive the transforms listed in Supplementary Material \ref{sub:LaplaceInverseDemandFunctions} by ourselves.
} At the same time, the integral may be thought of as expressing $U\left(q\right)$
as a linear combination of form-preserving functions of Theorem \ref{formpreserve}.\smallskip{}

\noindent \textbf{Technical Clarification (Integral Definition).}\footnote{Note that in certain parts of the paper we need a more general definition
of the integral (\ref{eq:UtilityAsASymbolicIntegral}) than the definition
(\ref{eq:UtilityAsRiemannStieljesIntegral}). In those cases, e.g.
in the proof of Theorem \ref{formpreserve}, we use the Schwartz distribution
theory instead of the Riemann-Stieltjes integral theory.} Here we define the integral (\ref{eq:UtilityAsASymbolicIntegral})
to be the Riemann-Stieltjes integral 
\begin{equation}
U\left(q\right)=\intop_{-\infty}^{0}q^{-t}du_{I}\left(t\right)\label{eq:UtilityAsRiemannStieljesIntegral}
\end{equation}
for some function $u_{I}\left(t\right)$, not necessarily nonnegative,
such that the integral converges. If this function is differentiable,
its derivative $u'_{I}\left(t\right)$ is the $u\left(t\right)$ that
appears on the right-hand side of (\ref{eq:UtilityAsASymbolicIntegral}).
If $u_{I}\left(t\right)$ is only piecewise differentiable, then $u\left(t\right)$
is not an ordinary function but involves Dirac delta functions (i.e.
point masses) at the points of discontinuity of $u_{I}\left(t\right)$.

The corresponding inverse demand curve is $P\left(q\right)=U'\left(q\right)=-\intop_{-\infty}^{0}t\ u\left(t\right)q^{-t-1}dt$,
or
\begin{equation}
P\left(q\right)=\intop_{-\infty}^{1}p\left(t\right)q^{-t}dt,\label{eq:InverseDemandAsASymbolicIntegral}
\end{equation}
where we defined $p\left(t\right)\equiv\left(1-t\right)u\left(t-1\right)$.
We see that $P\left(q\right)$ is a linear combination of form-preserving
functions of Theorem \ref{formpreserve}. The following theorem summarizes
convenient properties of this approach to demand curves: uniqueness,
inclusion of linear combinations of power functions, approximations to arbitrary
functions, and analyticity.

\begin{theorem}\label{TheoremLaplaceTransformWithRiemannStieltjesIntegrals}

\textbf{\emph{(Laplace-log
Transform with Riemann-Stieltjes Integrals)}}\\
\textbf{\emph{(A)}} For each function $U\left(q\right)$ that may
be represented in the form (\ref{eq:UtilityAsASymbolicIntegral})
in the sense of (\ref{eq:UtilityAsRiemannStieljesIntegral}), there
exists just one normalized\footnote{Normalization here means that $u_{I}\left(0+\right)=0$ and $u_{I}\left(t\right)=(u_{I}\left(t+\right)+u_{I}\left(t-\right))/2.$
See Section I.6 of \citet{widder2010laplace}. } function $u_{I}\left(t\right)$ such that (\ref{eq:UtilityAsRiemannStieljesIntegral})
holds.\textbf{\emph{ (B)}} Any polynomial utility function may be
written in the form (\ref{eq:UtilityAsASymbolicIntegral}). \textbf{\emph{(C)}}
All functions of the form (\ref{eq:UtilityAsASymbolicIntegral}) are
analytic. In particular, their derivatives of any order exist. \textbf{\emph{(D)}}
An arbitrary utility function $\tilde{U}\left(q\right)$ continuous
on an interval $\left[0,\bar{q}\right]$ may be approximated
with an arbitrary precision by utility functions of the form (\ref{eq:UtilityAsASymbolicIntegral}),
in the sense of uniform convergence\footnote{By \emph{uniform convergence} we mean that for any continuous $\tilde{U}\left(q\right)$
there exists a sequence $\left\{ U_{j}\left(q\right),j\in\mathbb{N}\right\} $
of functions of the form (\ref{eq:UtilityAsASymbolicIntegral}) such
that for any $\epsilon>0$, all elements of the sequence after some
position $n_{\epsilon}$ satisfy $\sup_{q\in\left[0,\bar{q}\right]}|\tilde{U}\left(q\right)-U_{j}\left(q\right)|<\epsilon$.} on $\left[0,\bar{q}\right]$. 

\end{theorem}

Note that part D of this theorem is a simple consequence of the Weierstrass approximation theorem.\footnote{There
is also a related, more powerful theorem, the M{\"u}ntz-Sz\'asz theorem. \citet{barnett1983muntz} use it to
propose to write direct demand as M{\"u}ntz-Sz\'asz polynomials of prices. Here we write inverse demand as polynomials
of powers of quantities (times possibly another power of quantity), but the same logic would apply here: we could use 
M{\"u}ntz-Sz\'asz polynomials.} The reader may ask
why we do not instead work simply with polynomials in $q$ and use them as approximations. Even though this 
would be possible in principle, it would not be practical. This is because in economics we often need
flexibility in the $q\rightarrow 0_+$ limit behavior of the inverse demand function. With any (finite-order) polynomial, we would always get
finite $\lim_{q\rightarrow 0_+} P(q)$, i.e., a choke price; to allow for $\lim_{q\rightarrow 0_+} P(q) = \infty$, we could not stay within a finite-order approximation.

Theorem \ref{formpreserve} allowed for functions other than
linear combinations of power functions, such as $q^{-\alpha}(\log q)^n$ or $\log q$, that are also useful
in economics.\footnote{For example, $P(q) = a - b \log q$ corresponds to exponential demand, studied by many authors, including \citet{acv}. 
Similarly, inverse demand functions $P(q) = a - b (\log q)^n$ have interesting implications for market failure in sequential 
supply chains such as Cournot's multiple-marginalization problem.} Although according to part D of the last
theorem, such functions may be approximated by functions of the Riemann-Stieltjes
interpretation (\ref{eq:UtilityAsRiemannStieljesIntegral}) of (\ref{eq:UtilityAsASymbolicIntegral}),
it is convenient to be able to write them \emph{exactly} in the form
(\ref{eq:UtilityAsASymbolicIntegral}) by using a more powerful definition
of the integral. This is achieved by the following counterpart of
Theorem \ref{TheoremLaplaceTransformWithRiemannStieltjesIntegrals},
which goes beyond the theory of the Riemann-Stieltjes integral and
instead discusses Laplace transform of generalized functions based
on the distribution theory by Laurent Schwartz. In the following,
$\bar{s}$ is a real number smaller than $\log\bar{q}$.

\begin{theorem}\label{TheoremLaplaceTransformWithSchwartzIntegrals}

\textbf{\emph{(Laplace-log Transform with Schwartz Integrals)}}
A function $U\left(q\right)$ such that the related function $U_{_{\left[s\right]}}\left(s\right)\equiv U\left(e^{s}\right)$
considered in the half-complex-plane domain $\mathbb{C}_{\bar{s}}^{-}\equiv\left\{ s|\mbox{Re }s<\bar{s}\right\} $
is analytic (i.e. holomorphic) and bounded by a polynomial function
may be expressed in the form (\ref{eq:UtilityAsASymbolicIntegral})
with $u$ representing a distribution, i.e. a generalized function,
or more precisely an element of $\mathcal{D}'$ as defined by \citet{zemanian1965distribution}.\footnote{Here ``bounded by a polynomial'' refers to the absolute value of
$U_{_{\left[s\right]}}\left(s\right)$ being no greater than the absolute
value of some polynomial of \emph{$s$ }in the domain $\mathbb{C}_{\bar{s}}^{-}$.} This distribution is unique. Conversely, for any Laplace-transformable
distribution $u$, the integral (\ref{eq:UtilityAsASymbolicIntegral})
viewed as a function of $s\equiv\log q$ in the domain $\mathbb{C}_{\bar{s}}^{-}$
is analytic and bounded by a polynomial of $s$.

\end{theorem}

\begin{definition}\textbf{\emph{(Laplace Versions of Economic Variables)\label{DefinitionLaplaceAsAnAdjective}}}
For a variable $V\left(q\right)$ that may be expressed as an integral
of the form $V\left(q\right)=\intop_{a}^{b}v\left(t\right)q^{-t}dt$,
we use the adjective \emph{Laplace} to refer to $v\left(t\right)$.
For example, $u\left(t\right)$ of (\ref{eq:UtilityAsASymbolicIntegral})
would be referred to as \emph{Laplace utility}, and $p\left(t\right)$
of (\ref{eq:InverseDemandAsASymbolicIntegral}) as \emph{Laplace inverse
demand} or \emph{Laplace price}. 

\end{definition}

Here we present a theorem describing the relationship of the integral and its discrete approximation. Its proof is constructed using the Euler-Maclaurin
formula related to the trapezoidal rule for numerical integration. Following the same logic, it is possible to derive and prove other approximation
theorems by adapting numerous theorems on numerical integration that exist in the applied mathematics literature. 

 \begin{theorem}\label{TheoremDiscreteApproximation}
\textbf{\emph{(Discrete Approximation)}}
The Laplace-log transform of a function \(f(t)\) may be expressed as 

\[\int_{-\infty }^{t_{\max }} q^{-t} f(t) \, dt=\text{$\Delta $t} \sum _{t\in T} q^{-t} f(t)-\frac{1}{2} q^{-t_{\max }} \text{$\Delta $t} f\left(t_{\max
}\right)-\frac{1}{2} q^{-t_{\min }} \text{$\Delta $t} f\left(t_{\min }\right)+R,\]
 where \(T\equiv \left\{t_{\min }, t_{\min }+\text{$\Delta $t}, \text{...} , t_{\max }\right\}\) is an evenly spaced grid
with at least two points, \(m\) is an integer such that \(f\) is \((2m+1)\)-times continuously differentiable on \(\left[t_{\min },t_{\max }\right]\)
and where the remainder \(R\) is described below. \end{theorem}

The remainder in the theorem consists of three parts: \(R\equiv R_1+R_2+R_3\). The first part \(R_1\) is simply the difference of \(\int_{-\infty
}^{t_{\max }} q^{-t} f(t) \, dt\) and \(\int_{t_{\min }}^{t_{\max }} q^{-t} f(t) \, dt\), and can be made very small, since \(\int_{-\infty }^{t_{\min
}} q^{-t} f(t) \, dt=q^{-t_{\min }} \int_{-\infty }^0 q^{-t} f\left(t+t_{\min }\right) \, dt\), which is exponentially suppressed for \(t_{\min }\)
chosen sufficiently negative and for a well-behaved \(f(t)\). The second part \(R_2\) may be expressed using derivatives of \(h(t)\equiv f(t) q^{-t}\text{}\)
at \(t_{\min }\) and \(t_{\max }\): 
\[R_2=\sum _{k=1}^m \frac{B_{2 k}}{(2 k)!} \left(\text{$\Delta $t}^{2 k} h^{(2 k-1)}\left(t_{\min }\right)-\text{$\Delta $t}^{2 k} h^{(2 k-1)}\left(t_{\max
}\right)\right),\]
 where \(B_{2 k}\) represent Bernoulli numbers. These terms are suppressed by powers of $\Delta $t as well as by the factorial
in the denominator.\footnote{Moreover, it is possible to rescale \(q\) by a constant factor
to keep log \(q\) small in absolute value for the range of quantities of interest.} The last part \(R_3\)
may be expressed and bounded using integrals of high derivatives of \(h(t)\):
\[\begin{array}{lll}
 R_3=-\frac{\text{$\Delta $t}^{2 m+1}}{(1+2 m)!} \int_{t_{\min }}^{t_{\max }} P_{1+2 m}(t) h^{(1+2 m)}(t) \, dt & \text{,   } & | R_3 |\leq \frac{2
\zeta (2 m+1)\text{$\Delta $t}^{2 m+1}}{(2 \pi )^{2 m+1}} \int_{t_{\min }}^{t_{\max }} | h^{(2 m+1)}(t) | \, dt \\
\end{array}
,\]
 where $\zeta $ is the Riemann zeta function and \(P_{1+2 m}\) are periodic Bernoulli functions. 

Note that this theorem provides a prescription for the weights of the power terms that approximate the integral and gives a bound for the associated
error. Of course, by leaving the weights flexible and fitting them using a generalized method of moments, it is possible to get a better approximation
with a smaller error. It is also possible to use alternative prescribed weights that correspond to other numerical integration methods. The fact
that very different weight choices can all give good approximations is related to the fact that the problem of finding optimal weights is a case
of so-called ill-posed problems, for which regularization is typically used in the applied mathematics and econometrics literature.\footnote{As mentioned above, the validity of such approximations may be proved along the lines of the proof given here.}

\subsection{Complete monotonicity and pass-through behavior\label{sub:CompleteMonotonicity}}

Continuous representations of inverse demand functions introduced in the previous subsection provide more conceptual clarity than discrete approximations, which have their idiosyncrasies depending on precisely how many terms are included. These representations in terms of inverse Laplace-log transform can provide useful intuition. For example, if a researcher wishes to find a good discrete approximation to a particular inverse demand function, the researcher may compute the exact inverse Laplace-log transform (or consult Supplementary Material \ref{sub:LaplaceInverseDemandFunctions}) to see where the Laplace inverse demand function $p(t)$ is positive or negative. Choosing a few evenly spaced mass points with a similar positivity/negativity pattern is then likely to lead to a tractable approximation to the original inverse demand function that has similar qualitative properties.

Inverse Laplace-log transform representations of inverse demand functions are useful also for another reason: 
Many demand curves have economic properties (determining many policy implications) that are easily understood in terms of the inverse Laplace-log transform. To
develop the related theory, we start with a standard definition of
completely monotone functions and then discuss relations between complete
monotonicity, the form of Laplace inverse demand, and economic consequences
for the pass-through rate.\footnote{\citet{brockett1987class} also discuss relations between complete
monotonicity and a type of Laplace transform. The Laplace transform
used there is in terms of quantity $q$, whereas in our discussion,
it is in terms of the logarithm of quantity. These two transforms
are distinct and should not be confused. Similarly, the mathematical notion of complete monotonicity has very different economic manifestations in \citet{brockett1987class} and in our work.}  We classify many commonly used demand functions using this property, given that, as we discussed in the previous section, many policy questions turn on properties of the pass-through rate tied down by complete monotonicity.  

\begin{definition} \textbf{\emph{(Completely Monotone Function)}}
A function $f\left(x\right)$ is completely monotone iff for all $n\in\mathbb{N}$
its $n$th derivative exists and satisfies \emph{$\left(-1\right)^{n}f^{\left(n\right)}\left(x\right)\ge0$.}

\end{definition}

It turns out that many commonly used demand functions are such that
the consumer surplus is completely monotone as a function of negative
log quantity. For this reason, we make the following definition.

\begin{definition} \textbf{\emph{(Complete Monotonicity of the Demand
Specification)}}\footnote{In principle, it is possible to empirically test whether an empirical
demand curve satisfies the complete monotonicity criterion. The relevant
empirical test has been developed by \citet{heckman1990testing}.
It would just have to be translated from the duration analysis context
to our demand theory context.}

We say that a demand function (or a utility function) satisfies the
complete monotonicity criterion iff the associated consumer surplus is a completely
monotone function of $-s$, i.e. for all $n\in\mathbb{N}$,\emph{
$CS_{_{\!\left[s\right]}}^{\left(n\right)}\left(s\right)\ge0$,
}or equivalently\emph{}\footnote{The fact that these definitions are equivalent may be seen as follows:
With the marginal utility of the outside good normalized to one and
$U\left(0\right)$ is set to zero, we have $CS\left(q\right)=-qP\left(q\right)+\int_{0}^{q}P\left(q_{1}\right)dq_{1}=-qU'\left(q\right)+\int_{0}^{q}U'\left(q_{1}\right)dq_{1}=U\left(q\right)-qU'\left(q\right)$.
This translates into $CS_{_{\!\left[s\right]}}\left(s\right)=U_{_{\!\left[s\right]}}\left(s\right)-U_{_{\!\left[s\right]}}'\left(s\right)$, where we use the subscript $[s]$ to emphasize that the variable is to be treated as a function of $s$. 
The equivalence for any $n\in\mathbb{N}$ then follows by differentiation.}\emph{ 
$U_{_{\!\left[s\right]}}^{\left(n\right)}\left(s\right)-U_{_{\!\left[s\right]}}^{\left(n+1\right)}\left(s\right)\ge0$.
}Strict complete monotonicity criterion then refers to these inequalities
being strict.

\end{definition}

\begin{theorem}\label{TheoremNonnegativityOfLaplaceConsumerSurplus}

\textbf{\emph{(Nonnegativity
of Laplace Consumer Surplus) }}A (single-product) utility function
is bounded below and satisfies the complete monotonicity criterion
iff the Laplace consumer surplus $cs\left(t\right)$ is nonnegative
and supported on $\left(-\infty,0\right)$, i.e. $CS\left(q\right)=\int_{-\infty}^{0}$$cs\left(t\right)q^{-t}dt$
for some $cs\left(t\right)\ge0$.\emph{}

\end{theorem}

\begin{theorem}\label{TheoremMonotonicityOfThePassThroughRate}

\textbf{\emph{(Monotonicity of the Pass-Through Rate)}} The complete
monotonicity criterion for demand functions implies the pass-through
rate decreasing with quantity in the case of constant-marginal-cost
monopoly. The only exception is BP demand, for which
the pass-through rate is constant.

\end{theorem}

\begin{theorem}\label{TheoremCompleteMonotonicityOfDemandSpecification}

\textbf{\emph{(Complete Monotonicity of Demand Specification)}} The
following demand functions satisfy the complete monotonicity criterion:\footnote{The parameter names are chosen as in Mathematica.}\\
 Pareto/constant elasticity ($\epsilon>1$), BP ($\epsilon>1$),
logistic distribution, log-logistic distribution ($\gamma>1$), Gumbel
distribution ($\alpha>1$), Weibull distribution ($\alpha>1$), Fr\' echet
distribution ($\alpha>1$), gamma distribution ($\alpha>1$), Laplace
distribution\footnote{Each half of the distribution separately, or the full distribution
smoothed by arccosh to ensure the existence of the derivatives.}, Singh-Maddala distribution ($a>1$), Tukey lambda distribution ($\lambda<1$),
Wakeby distribution ($\beta>1$), generalized Pareto distribution
($\gamma<1$), Cauchy distribution. 

\end{theorem}

\noindent{\bf Corollary. (Monotonicity of the Pass-Through Rate)} \emph{ The last two theorems imply that the demand functions listed in
Theorem \ref{TheoremCompleteMonotonicityOfDemandSpecification} lead
to constant-marginal-cost pass-through rate decreasing in quantity, with
the exception of Pareto/constant elasticity as well as the more general
BP demand, which are known to lead to constant pass-through.}

\begin{theorem}\label{TheoremAbsenceOfCompleteMonotonicityOfDemandSpecification}

\textbf{\emph{(Absence of Complete Monotonicity of Demand Specification)}}
The following demand functions do \textbf{not} satisfy the complete
monotonicity criterion: normal distribution, lognormal distribution,
constant superelasticity (Klenow and Willis), Almost Ideal Demand
System (either with finite or infinite surplus), log-logistic distribution
($\gamma<1$), Fr\' echet distribution ($\alpha<1$), Weibull distribution
($\alpha<1$), Gumbel distribution ($\alpha<1$), Pareto/constant
elasticity ($\varepsilon>1$), gamma distribution ($\alpha<1$), Singh-Maddala
distribution ($a<1$), Tukey lambda distribution ($\lambda>1$), Wakeby
distribution ($\beta<1$), generalized Pareto distribution ($\gamma>1$).\end{theorem}

In our 
Supplementary Material \ref{forms} we provide a more complete taxonomy of pass-through properties of some of the demand forms mentioned here.  
Interestingly, the normal distribution has economic properties close to those of forms that satisfy the complete monotonicity criterion, since the non-complete monotonicity manifests itself only for very high-order derivatives.\footnote{In particular we found that  the normal distribution of consumer values has properties
very close to those satisfying the complete monotonicity criterion:
constant-marginal-cost pass-through is increasing in price (as we
show below), and low-order derivatives of $CS\left(s\right)$ with
respect to $-s$ are positive. We concluded that the complete monotonicity
criterion is not satisfied based on examining the sign on the tenth
derivative of $CS\left(s\right)$. The absence of complete monotonicity
is consistent with our expression to the corresponding Laplace inverse
demand, which does not seem to satisfy $t\,cs\left(t\right)\ge0$.
In most economic applications, the difference from completely monotone
problems is inconsequential because it manifests itself only in very
high derivatives of $CS\left(s\right)$.}  The lognormal distribution is not quite 
so well-behaved, but the more realistic income model (the double Pareto lognormal) behaves similarly for calibrated parameter values.

\section{Conclusion}

\label{conclusion}

We have shown that the set of analytic solutions to many common economic problems is substantially richer than typically assumed. 
They let economists work with flexible, realistic models, instead of imposing restrictive, unrealistic assumptions
in order to get analytic solutions of traditional kinds. Our approach to getting analytic solutions is also useful when
 applied to sub-problems of larger economic models. In those cases
it leads to the ability to solve those models numerically in a much more efficient way, as in our international trade
application.

The international trade model provides a perspective on the gravity equation of trade that is completely different
from the rest of the literature. The model resolves economic puzzles related to the cost of trade since its parameters take 
realistic values and at the same time the model matches well firm-level and country-level trade patterns.

Of course, there are many other applications of our method, some of which we briefly discussed here, some of which 
we will report in separate papers, and some of which, hopefully, the reader will develop on his/her own.

\singlespacing
\bibliographystyle{aer}
\bibliography{FunctionalFormsForTractableEconomicModels.bib}

\appendix

\begin{center}
{\LARGE{}{{Appendix}}} \singlespacing
\par\end{center}

\section{Proofs of Theorems}\label{AppendixProofsOfTheorems}

\textbf{Proof of Theorem \ref{formpreserve} (Characterization of
Form-Preserving Functions).} Here we present a constructive proof
of the theorem that exactly traces the steps we first used to derive 
the theorem's statement. It is instructive for readers familiar 
with Fourier transform or Laplace transform because it 
highlights the properties of functions we emphasize in this paper
and shows how using the transforms, calculations may be conveniently 
performed just in a couple of lines. Other readers may prefer
reading Supplementary Material 
\ref{AppendixDiscussionOfTheCharacterizationTheorem}, where 
we discuss how the theorem may be proven without 
functional transforms.

Here we derive the structure of $m$-dimensional 
functional form classes $\mathcal{C}$ that are invariant under
average-marginal transformations. We take as the domain 
of the functions an open interval $I$ of positive real numbers, which may
include all positive real numbers.\footnote{If functions of negative numbers
were of interest, we could simply switch to working in terms of $(-q)$ instead
of $q$ and derive analogous results.}
For convenience we express the (infinitely
differentiable) functions $F\left(q\right)$ on $I$
in terms of functions $G\left(s\right)$ defined in the corresponding logarithmically 
transformed domain,
with the identification $s\equiv\log q$, $F\left(q\right)\equiv G\left(\log q\right)$.
Consider a function $F\left(q\right)\mathcal{\in C}$ and its counterpart
$G\left(s\right)$. In terms of $G$, the average-marginal form-preservation
requires that the counterpart of $aG+bG'$ belong to the class $\mathcal{C}$,
if the counterpart of $G$ does so. For technical reasons, we will
work with $G\left(s\right)$ multiplied by the characteristic function
$1_{S}\left(s\right)$ of an arbitrarily chosen finite non-empty interval
$S\equiv\left(s_{1},s_{2}\right)\in I$, i.e. with $G_{S}\left(s\right)\equiv G\left(s\right)1_{S}\left(s\right)$.\footnote{If we worked with infinite intervals, 
the convergence of the integrals below would not be always guaranteed.}
We denote by $\hat{G}_{S}\left(\omega\right)$ the Fourier transform
of $G_{S}\left(s\right)$, which in turn may be expressed as the inverse
Fourier transform $G_{S}\left(s\right)=(2\pi)^{-1/2}\int_{-\infty}^{\infty}\hat{G}_{S}\left(\omega\right)e^{-i\omega s}d\omega$.\footnote{The Fourier transform used in the proof is equivalent to the Laplace
transform with imaginary $s$. Both transforms may be thought of as
parts of the holomorphic Fourier-Laplace transform.}

By iterating the defining property of average-marginal form-preservation,
we know that the class $\mathcal{C}$ contains also counterparts of
the derivatives $G^{\left(n\right)}\left(s\right)$. We will consider
the first $m$ of them, in addition to $G\left(s\right)$. For $n=1,2,...,m$,
we denote by $G_{S}^{\left(n\right)}\left(s\right)$ the truncation
of $G^{\left(n\right)}\left(s\right)$ to the interval $S$, i.e.
$G_{S}^{\left(n\right)}\left(s\right)\equiv G^{\left(n\right)}\left(s\right)1_{s\in S}$.
Inside the interval $S$,
\begin{equation}
G_{S}^{\left(n\right)}\left(s\right)=(2\pi)^{-1/2}\int_{-\infty}^{\infty}\left(-i\omega\right)^{n}\hat{G}_{S}\left(\omega\right)e^{-i\omega s}d\omega,\quad\mbox{for}\;s\in S,\;n\in\{0,1,2,...,m\}.\label{eq:NThDerivative}
\end{equation}
The $m+1$ functions $G_{S}\left(s\right),G_{S}^{\left(1\right)}\left(s\right),G_{S}^{\left(2\right)}\left(s\right),...,G_{S}^{\left(m\right)}\left(s\right)$
span a vector space with dimensionality $m+1$ or less. Dimensionality
equal to $m+1$ would contradict the assumption of having an $m$-dimensional
functional form class, which implies that the set of functions $G_{S}\left(s\right),G_{S}^{\left(1\right)}\left(s\right),G_{S}^{\left(2\right)}\left(s\right),$
$...,G_{S}^{\left(m\right)}\left(s\right)$ must be linearly dependent
on the interval $S$. As a result, there must exist a polynomial $T_{0}\left(.\right)$
(with real coefficients), such that 
\begin{equation}
\int_{-\infty}^{\infty}T_{0}\left(-i\omega\right)\hat{G}_{S}\left(\omega\right)e^{-i\omega s}d\omega\label{eq:T0GS}
\end{equation}
is zero for any $s\in S$. This expression vanishes not only for $s\in S\equiv\left(s_{1},s_{2}\right)$,
but also for $s\in\left(-\infty,s_{1}\right)$ and $s\in\left(s_{2},\infty\right)$.
This is because the right-hand-side of (\ref{eq:NThDerivative}) when
extended to arbitrary $s\in\mathbb{R}$ represents the $n$th derivative
of $G_{S}\left(s\right)$ in the sense of the Schwartz distribution
theory, and given that $G_{S}\left(s\right)$ vanishes for $s\in\left(-\infty,s_{1}\right)$
and $s\in\left(s_{2},\infty\right)$, so must its $n$th derivative.
Given that the expression (\ref{eq:T0GS}) is a generalized function\footnote{By a generalized function we mean an element of the space $\mathcal{S}'\left(\mathbb{R}\right)$
of distributions.} of $s$ that gives zero when integrated against any test function\footnote{A test function here refers to an element of the space $\mathcal{S}\left(\mathbb{R}\right)$
of space of rapidly decreasing functions.} supported on $(-\infty,s_{1}-\epsilon]\cup\left[s_{1}+\epsilon,s_{2}-\epsilon\right]\cup[s_{2}+\epsilon,\infty)$
for any $\epsilon>0$, we may write it as a linear combination of
Dirac delta functions and a finite number of their derivatives located
at $s_{1}$ and $s_{2}$. By computing its Fourier transform we find
that $T_{0}\left(-i\omega\right)\hat{G}_{S}\left(\omega\right)$ must
be of the form 
$
T_{1}\left(\omega\right)e^{is_{1}\omega}+T_{2}\left(\omega\right)e^{is_{2}\omega}
$
with some polynomials $T_{1}\left(\omega\right)$ and $T_{2}\left(\omega\right)$, with
complex coefficients in general.
Consequently, $\hat{G}_{S}\left(\omega\right)$ may be written as
\[
\hat{G}_{S}\left(\omega\right)=\frac{T_{1}\left(\omega\right)}{T_{0}\left(-i\omega\right)}e^{is_{1}\omega}+\frac{T_{2}\left(\omega\right)}{T_{0}\left(-i\omega\right)}e^{is_{2}\omega}.
\]
The polynomial $T_{0}\left(-i\omega\right)$ may have a common factor
with $T_{1}\left(\omega\right)$ or $T_{2}\left(\omega\right)$ or
both. If we cancel these common factors, we may rewrite the expression
as

\begin{equation}
\hat{G}_{S}\left(\omega\right)=\frac{T_{3}\left(\omega\right)}{T_{5}\left(\omega\right)}e^{is_{1}\omega}+\frac{T_{4}\left(\omega\right)}{T_{6}\left(\omega\right)}e^{is_{2}\omega}\label{eq:GHatS}
\end{equation}
for some polynomials $T_{3}$, $T_{4}$, $T_{5}$, and $T_{6}$, such
that $T_{3}$ has no common divisors with $T_{5}$ and similarly for
$T_{4}$ with $T_{6}$. Let us compute the inverse Fourier transform
of the last expression for $\hat{G}_{S}\left(\omega\right)$ using
the residue theorem. To perform the integration, we consider each
of the two terms in (\ref{eq:GHatS}) separately and specialize to
$s\in S$. We close the integration contour by semicircles at infinity
of the complex plane, correctly chosen so that their contribution
to the integral vanishes. The integral value is then equal to the
sum of the pole (residue) contributions, which give exponentials of
$s$ multiplied by polynomials of $s$. We see that for $s\in S$,
$
G_{S}\left(s\right)=\sum_{j=1}^{N}D_{j}\left(s\right)e^{-ist_{j}},
$
for some integer $N$, complex numbers $t_{j}$ and polynomials $D_{j}\left(s\right)$.
Since the interval $S$ was chosen arbitrarily, not just $G_{S}\left(s\right)$,
but also $G\left(s\right)$ itself must take this form. In the last
expression the constants may be complex. Without loss of generality,
we can assume that the first $N_{1}$ numbers $t_{j}$ are real and
the remaining ones have an imaginary part. By combining individual
terms into real contributions so that $G\left(s\right)$ is real,
we get 
\[
G\left(s\right)=\sum_{j=1}^{N_{1}}A_{j}\left(s\right)e^{-st_{j}}+\sum_{j=1}^{N_{2}}\left(B_{j}\left(s\right)\cos\tilde{t}_{j}s+C_{j}\left(s\right)\sin\tilde{t}_{j}s\right)e^{-\hat{t}_{j}s},
\]
where $A_{j}\left(s\right)$, $B_{j}\left(s\right)$, and $C_{j}\left(s\right)$
are polynomials, and $N_{1}+2N_{2}=N$. This form of $G\left(s\right)$
translates into the following form of $F\left(q\right)$: 
\begin{equation}
F\left(q\right)=\sum_{j=1}^{N_{1}}A_{j}\left(\log q\right)q^{-t_{j}}+\sum_{j=1}^{N_{2}}\left(B_{j}\left(\log q\right)\cos\left(\tilde{t}_{j}\log q\right)+C_{j}\left(\log q\right)\sin\left(\tilde{t}_{j}\log q\right)\right)q^{-\hat{t}_{j}}.\label{eq:AverageMarginalInvariantFunction}
\end{equation}
If we wish to exclude the possibility of oscillations, e.g. in economic
applications where we allow the functional form to be valid arbitrarily
close to $q=0$, we can set the polynomials $B_{j}$ and $C_{j}$
to zero and consider only functions of the form 
$
F\left(q\right)=\sum_{k=1}^{N_{1}}A_{j}\left(\log q\right)q^{-t_{j}}.
$
An example of functional forms of this kind is $aq^{-t}+bq^{-u}+cq^{-u}\log q+dq^{-u}(\log q)^{2}$.
The reader can easily verify that this is a four-dimensional functional
form class invariant under average-marginal transformations. In general,
it is now straightforward to check that the result (\ref{eq:AverageMarginalInvariantFunction})
implies the statement of the theorem.  \hfill{}\Square{}

\noindent \textbf{Proof of Theorem \ref{TheoremClosedFormSolutions} (Closed-Form Solutions).}  The proof is straightforward. By assumption, there exists some definite power $b$ such that $x \equiv q^b$ satisfies an algebraic equation of order $k$: $P_k(x)=0$, where $P_k(x)$ is a polynomial of order at most $k$. For this to be true, all elements of the functional form class must factorize as $q^a P_k(q^b)$ for some definite $a$. When expanded, the powers of $q$ in individual terms lie on the grid ${a, a+b, ... , a + b k }.$ \hfill{}\Square{}
\noindent \paragraph{Proof of Theorem \ref{TheoremAggregation}}
\hspace*{0.5ex} \textbf{(Aggregation)}.
 The firm{'}s revenue \(q P(q)\), cost \(\int \text{MC}(q) \, dq\), and profit are all linear combinations of powers of
\(q\). For this reason, it suffices to show that it is possible to perform explicitly aggregation integrals $\mathcal{I}$ for powers of \(q\) (the
quantity optimally chosen by a firm with productivity parameter \(a\)):
$\mathcal{I}\equiv \int q(a)^{\gamma _1} \, dG(a)$.
 Changing the integration variable to \(q\) gives:
$\mathcal{I}=\int q^{\gamma _1} G'(a(q)) a'(q) \, dq$. 
 The firm{'}s first-order condition equates the marginal revenue \(R'(q)=P(q)+q P'(q)\) to the marginal cost
\(\text{MC}_0(q)+a \text{MC}_1(q)\) and implies 
\[a=\frac{R'(q)-\text{MC}_0(q)}{\text{MC}_1(q)}\Rightarrow a'(q)=\frac{R''(q)-\text{MC}_0'(q)}{\text{MC}_1(q)}-\frac{R'(q)-\text{MC}_0(q)}{\text{MC}_1(q){}^2}
\text{MC}_1'(q).\]
 Substituting these expressions into the integral gives
\[\mathcal{I}=\int q^{\gamma _1} \left(\frac{R''(q)-\text{MC}_0'(q)}{\text{MC}_1(q)}-\frac{R'(q)-\text{MC}_0(q)}{\text{MC}_1(q){}^2} \text{MC}_1'(q)\right)
G'\left(\frac{R'(q)-\text{MC}_0(q)}{\text{MC}_1(q)}\right) \, dq.\]
 Since \(G'(a)\) is a mixture of powers of \(a\), and \(\left(R'(q)-\text{MC}_0(q)\right) \text{MC}_1'(q)\)
and { }\(R''(q)-\text{MC}_0'(q)\) are mixtures of powers of \(q\), the integral on the right-hand side may be written as a linear combination of
integrals of the type
\[\int q^{\gamma _5} \text{MC}_1(q){}^{\gamma _7} \left(-\text{MC}_0(q)+R'(q)\right){}^{\gamma _6} \, dq,\]
 where \(\gamma _7\) equals \(-\gamma _6-1\) or \(-\gamma _6-2\). Given our assumptions, up to a known multiplicative
constant this integral equals 
$\int q^{\gamma _8} N_1\left(q^{\alpha }\right){}^{\gamma _9} N_2\left(q^{\alpha }\right){}^{\gamma _{10}} \, dq$.
 If we change the integration variable to \(x\equiv q^{\alpha }\), the problem reduces to computing the integral
$\int x^{\gamma _{11}} N_1(x){}^{\gamma _{12}} N_2(x){}^{\gamma _{13}} \, dx$.
To complete the proof, it suffices to examine the structure of this intergral for different structures of the polynomials. 

Depending on the structure of the polynomials, the following six non-exclusive
cases may arise:

(1) If the polynomials \(N_1\) and \(N_2\) are trivial, the integral reduces to a power function of \(q\), without any special functions.

(2) If either \(N_1\) or \(N_2\) is trivial and the other polynomial is linear, the integral leads to the standard hypergeometric function, denoted
\(\text{}_2F_1\), since up to an additive constant
\[\int x^{\gamma _{11}} \left(1+\gamma _{14} x\right){}^{\gamma _{13}} \, dx=\frac{x^{1+\gamma _{11}}}{1+\gamma _{11}} \, _2F_1\left(1+\gamma _{11},-\gamma
_{13};2+\gamma _{11};-x \gamma _{14}\right)\]

(3) If both \(N_1\) and \(N_2\) are linear, the integral leads to the standard Appell function, denoted \(F_1\), since up to an additive constant
\[\int x^{\gamma _{11}} \left(1+\gamma _{18} x\right){}^{\gamma _{12}} \left(1+\gamma _{19} x\right){}^{\gamma _{13}} \, dx=\frac{x^{1+\gamma _{11}}}{1+\gamma
_{11}} F_1\left(1+\gamma _{11};-\gamma _{12},-\gamma _{13};2+\gamma _{11};-x \gamma _{18},-x \gamma _{19}\right)\]

(4) If either \(N_1\) and \(N_2\) is trivial and the other polynomial is quadratic, the integral again leads to the standard Appell function, denoted
\(F_1\):
\[\int x^{\gamma _{11}} \left(1+\gamma _{14} x+\gamma _{15} x^2\right){}^{\gamma _{13}} \, dx=\text{}\]

\[\frac{\gamma _{15}^{\gamma _{13}} x^{1+\gamma _{11}}}{1+\gamma _{11}} \left(\frac{1+x \gamma _{14}+x^2 \gamma _{15}}{\gamma _{15}+x \gamma _{14}
\gamma _{15}+x^2 \gamma _{15}^2}\right){}^{\gamma _{13}} F_1\left(1+\gamma _{11};-\gamma _{13},-\gamma _{13};2+\gamma _{11};\gamma _{16} x,\gamma
_{17} x\right)\]
 where \(\gamma _{16}=-2 \gamma _{15} \left(\gamma _{14}+\sqrt{\gamma _{14}^2-4 \gamma _{15}}\right){}^{-1}\),
and \(\gamma _{17}=2 \gamma _{15} \left(-\gamma _{14}+\sqrt{\gamma _{14}^2-4 \gamma _{15}}\right){}^{-1}\).

(5) If \(N_1\) and \(N_2\) are both of order less than five, we can factorize them into products of linear polynomials with the factorization performed
in closed form by the method of radicals. The resulting integral may be expressed using Lauricella functions. In particular, by the fundamental theorem
of algebra, \(N_1\) and \(N_2\) may be written as products of linear functions. This means that up to a multiplicative constant, \(x^{\gamma _{11}}
\left(1+\gamma _{18}x\right){}^{\gamma _{12}} \left(1+\gamma _{19}x\right){}^{\gamma _{13}}\) equals \(x^{b-1} \left(1-u_1x\right){}^{-b_1} \ldots
 \left(1-u_nx\right){}^{-b_n}\), where \(u_i\) represent the reciprocals of the roots of the polynomials. These roots, as well the constants \(b\),
\(b_1\), ..., \(b_n\) may be found explicitly using the standard formulas for solutions to quadratic, cubic, or quartic equations. Up to an additive
constant, the corresponding integral equals
\[\int x^{b-1} \left(1-u_1 x\right){}^{-b_1} \ldots  \left(1-u_n x\right){}^{-b_n} \, dx=\frac{x^b}{b} F_D{}^{(n)}\left(b,b_1,\ldots ,b_n,b+\text{1;}
u_1 x,\ldots ,u_n x\right)\]
 This is because in general the Lauricella function \(F_D{}^{(n)}\) is defined as
\[F_D{}^{(n)}\left(b,b_1,\ldots ,b_n,c ; x_1,\ldots ,x_n\right)=\frac{\Gamma (c)}{\Gamma (b) \Gamma (c-b)} \int_0^1 y^{b-1} (1-y)^{c-b-1} \left(1-x_1
y\right){}^{-b_1} \ldots  \left(1-x_n y\right){}^{-b_n} \, dy\]
 with \(\Gamma\) denoting the standard gamma function, and in the special case of \(c=b+1\) this definition
becomes
\[F_D{}^{(n)}\left(b,b_1,\ldots ,b_n,b+1 ; x_1,\ldots ,x_n\right)=b \int_0^1 y^{b-1} \left(1-x_1 y\right){}^{-b_1} \ldots  \left(1-x_n y\right){}^{-b_n}
\, dy\]
 Substituting \(y\to  \left.x_0\right/x, x_1\to  u_1x\) and \(x_n\to  u_nx\) then leads to the desired result
for the integral:
\[\int_0^x x_0^{b-1} \left(1-u_1 x_0\right){}^{-b_1} \ldots  \left(1-u_n x_0\right){}^{-b_n} \, dx_0=\frac{x^b}{b} F_D{}^{(n)}\left(b,b_1,\ldots
,b_n,b+1 ; u_1 x,\ldots ,u_n x\right)\]

(6) Finally, if either \(N_1\) or \(N_2\) is of order five or higher, the factorization involves root functions, since the method of radicals can
no longer be used. However, the resulting integral may still be expressed using Lauricella functions as described above.

 We conclude that the structure of the resulting expressions for the integral agrees with the statement of Theorem 3.
 \hfill{}\Square{}

\noindent \textbf{Proof of Theorem \ref{TheoremLaplaceTransformWithRiemannStieltjesIntegrals} (Laplace-log Transform with Riemann-Stieltjes Integrals).} 
\textbf{(A)} This follows from Theorem
I.6.3 of of \citet{widder2010laplace}.  \textbf{(B)} If we choose $u_{I}\left(t\right)$
appearing in Equation \ref{eq:UtilityAsRiemannStieljesIntegral} from the paper to be piecewise constant with
a finite number $N$ of points of discontinuity $\left\{ t_{j},j=1,2,...,N\right\} $,
the integral becomes 
$
U\left(q\right)=\sum_{j=1}^{N}a_{j}q^{-t_{j}},
$
where $a_{j}$ is the (signed) magnitude of the discontinuity at point
$t_{j}$, i.e. the magnitude of the mass that $u\left(t\right)$ has
at point $t_{j}$. If we choose $t_{j}$ to be nonpositive integers,
$U\left(q\right)$ will be a polynomial of $q$. By appropriate choices
of $N$ and $a_{j}$, any polynomial of $q$ may be expressed in this
way. \textbf{(C)} Given that polynomials
are included in Equation \ref{eq:UtilityAsASymbolicIntegral} from the paper, the theorem follows from
the Weierstrass approximation theorem, which states that polynomials
are dense in the space of continuous functions on a compact interval.
For a constructive proof of the theorem due to Bernstein, see e.g.
Section VII.2 of \citet{feller2008introduction}.  \textbf{(D)} This follows from Theorem
I.5a of \citet{widder2010laplace}. \hfill{}\Square{}

\noindent \textbf{Proof of Theorem \ref{TheoremLaplaceTransformWithSchwartzIntegrals} (Laplace-log Transform with Schwartz Integrals).} The three sentences of
the theorem are implied by the following statements in \citet{zemanian1965distribution}:
(1)\emph{ }Theorem 8.4-1 and Corollary 8.4-1a, (2) Theorem 8.3-1a,
(3) Theorem 8.3-2 and the text following Corollary 8.4-1a.\hfill{}\Square{}

\noindent \textbf{Proof of Theorem  \ref{TheoremDiscreteApproximation}
(Discrete Approximation).} This theorem follows straightforwardly from Theorem 4 of \citet{apostol1999elementary}. 
That theorem provides in its Equation 25 a convenient form of the Euler-Maclaurin formula, which may be written, after a small change
of notation, as: 
\begin{center}
\begin{tabular}{c}
$\sum _{k=1}^{n_T} F(k)=\int_1^{n_T} F(x) \, dx+\mathcal{C}(F)+E_F\left(n_T\right),$  \\
$\mathcal{C}(F)=\frac{1}{2} F(1)-\sum _{r=1}^m \frac{B_{2 r}}{(2 r)!} F^{(2 r-1)}(1)+\frac{1}{(2 m+1)!} \int_1^{\infty } P_{2 m+1}(x) F^{(2 m-1)}(x)
\, dx,$ \\
$E_F\left(n_T\right)=\frac{1}{2} F\left(n_T\right)-\sum _{r=1}^m \frac{B_{2 r}}{(2 r)!} F^{(2 r-1)}\left(n_T\right)+\frac{1}{(2 m+1)!} \int_{n_T}^{\infty
} P_{2 m+1}(x) F^{(2 m-1)}(x) \, dx.$
\end{tabular}
\end{center}
We can use this form of the Euler-Maclaurin formula to prove the discrete approximation theorem. The relationship we
want to prove is
\begin{center}
$\sum _{t\in T} q^{-t} f(t)=\frac{1}{\text{$\Delta $t}} \int q^{-t} f(t) \, dt+\frac{1}{2} q^{-t_{\min }} f\left(t_{\min }\right)+\frac{1}{2} q^{-t_{\max
}} f\left(t_{\max }\right)-\frac{R_1+R_2+R_3}{\text{$\Delta $t}},$
\end{center}
 where \(T\equiv \left\{t_{\min }, t_{\min }+\text{$\Delta $t}, \text{...} , t_{\max }\right\}\) and \(n_T\) is the number
of points in the grid \(T\). Equivalently,
\begin{center}
$\sum _{t\in T} q^{-t} f(t)=\frac{1}{\text{$\Delta $t}} \int_{t_{\min }}^{t_{\max }} q^{-t} f(t) \, dt+\frac{1}{2} q^{-t_{\min }} f\left(t_{\min
}\right)+\frac{1}{2} q^{-t_{\max }} f\left(t_{\max }\right)-\frac{R_2+R_3}{\text{$\Delta $t}}.$
\end{center}
If we use the notation
\[F(k)\equiv q^{-t_{\min }-k \text{$\Delta $t}} f\left(t_{\min }+(k-1) \text{$\Delta $t}\right)\]
 we can rewrite the individual terms in the desired formula as
 \begin{center}
\begin{tabular}{c}
$\sum _{t\in T} q^{-t} f(t)=\sum _{k=1}^{n_T} F(k),$  \\
$\frac{1}{2} q^{-t_{\min }} f\left(t_{\min }\right)+\frac{1}{2} q^{-t_{\max }} f\left(t_{\max }\right)=\frac{F(1)}{2}+\frac{F\left(n_T\right)}{2},$ \\
$\frac{R_2}{\text{$\Delta $t}}=\sum _{r=1}^m \frac{B_{2 r}}{(2 r)!} \left(F^{(-1+2 r)}(1)-F^{(-1+2 r)}\left(n_T\right)\right),$ \\
$\frac{R_3}{\text{$\Delta $t}}=-\frac{\text{$\Delta $t}^{2 m}}{(1+2 m)!} \int_{t_{\min }}^{t_{\max }} P_{1+2 m}(t) h^{(1+2 m)}(t) \, dt=-\frac{1}{(2
m+1)!} \int_1^n P_{1+2 m}(x) F^{(1+2 m)}(x) \, dx.$
\end{tabular}
\end{center}
 By comparing these expressions with those of Theorem 4 of  \citet{apostol1999elementary}, we see that the main statement of Theorem  \ref{TheoremDiscreteApproximation} is valid. The bound on \(R_3\) then simply follows from the formula \(| P_{2 m+1}(x) |\leq 2 (2 m+1)! (2 \pi )^{-2 m-1}\)
; see p. 538 of  \citet{lehmer1940maxima}.
\hfill{}\Square{} \\

\noindent \textbf{Proof of Theorem \ref{TheoremNonnegativityOfLaplaceConsumerSurplus} (Nonnegativity of Laplace Consumer Surplus).} This theorem follows from Bernstein's theorem on completely monotone
functions, formulated e.g. as Theorem IV.12a of \citet{widder2010laplace}
or Theorem 1.4 of \citet{schilling2012bernstein}. \hfill{}\Square{}

\noindent \textbf{Proof of Theorem \ref{TheoremMonotonicityOfThePassThroughRate} (Monotonicity of the Pass-Through Rate).}  Constant marginal cost monopoly
pass-through rate may be expressed as 
$
\rho={CS_{_{\!\left[s\right]}}'\left(s\right)}/{CS_{_{\!\left[s\right]}}''\left(s\right)},
$
which is straightforward to verify from the basic definitions. For a completely monotone problem, Laplace consumer surplus $cs\left(t\right)$
is nonnegative. For this reason, the inverse of $\rho$ may be expressed
as a weighted average of $t$ with nonnegative weight
$
w\left(t,s\right)\equiv{t\:cs\left(t\right)e^{-st}}/{\int_{-\infty}^{0}t\,cs\left(t\right)e^{-st}dt}
$
as follows 
\[
\frac{1}{\rho}=\frac{CS_{_{\!\left[s\right]}}''\left(s\right)}{CS_{_{\!\left[s\right]}}'\left(s\right)}=-\frac{\int_{-\infty}^{0}t^{2}cs\left(t\right)e^{-st}dt}{\int_{-\infty}^{0}t\ cs\left(t\right)e^{-st}dt}=-\int_{-\infty}^{0}t\,w\left(t,s\right)dt.
\]
In response to an increase in $s$, the weight gets shifted towards
more negative $t$,\footnote{In the same mathematical sense as in the definition of first order
stochastic dominance.} and $1/\rho$ decreases. We conclude that $\rho$ is decreasing in
$q$. Only if $t\,cs\left(t\right)$ is supported at one point will
there be no shift in weight and $\rho$ remains constant. That case
corresponds to BP demand.\hfill{}\Square{}

 \begin{figure}
\begin{centering}
\includegraphics[scale=0.95]{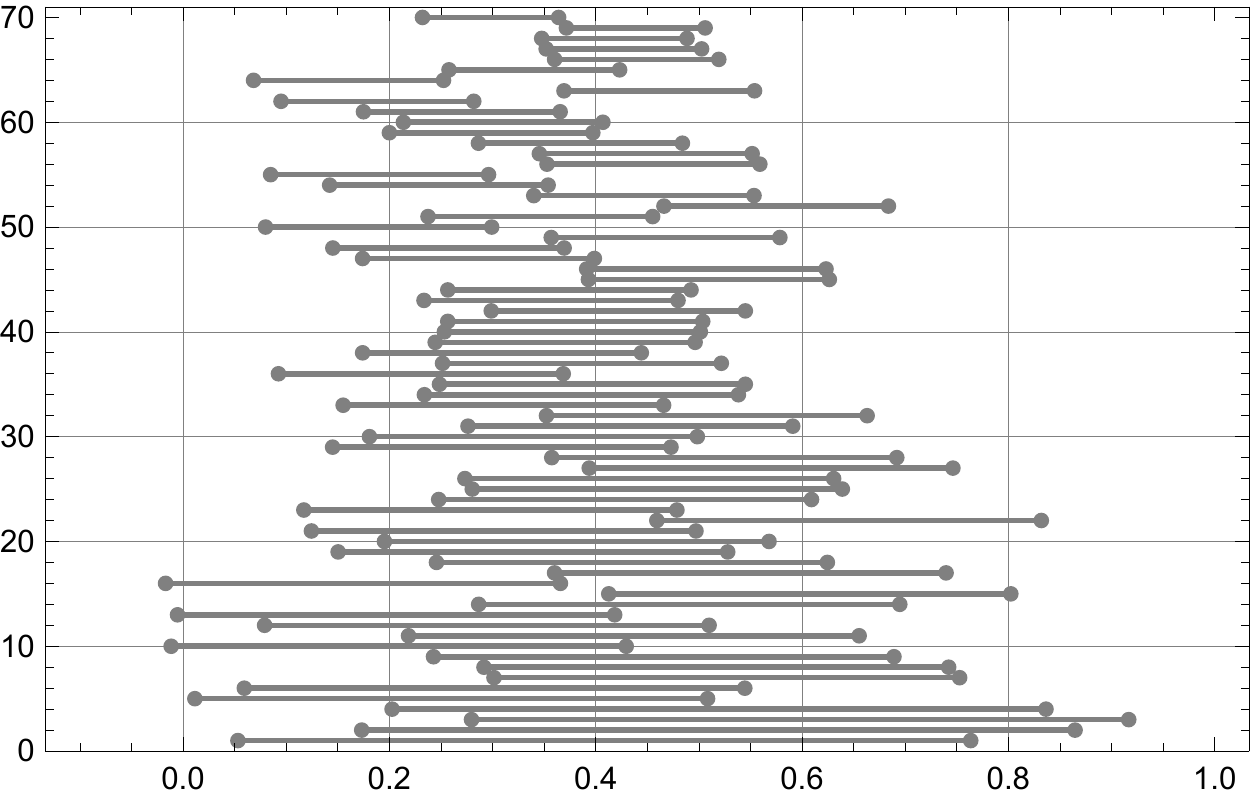}
\caption{Confidence intervals for the cost exponent \(\gamma\) for individual industries at the 95\% level. For visualization purposes,
the industries are ordered by the standard deviation of \(\gamma\) and stacked vertically.} \label{FigureConfidenceIntervalsForBeta}
\end{centering}
\end{figure}

\noindent \textbf{Proof of Theorem \ref{TheoremCompleteMonotonicityOfDemandSpecification} (Complete Monotonicity of Demand Specification).} The complete monotonicity properties follow by straightforwardly recognizing
that in these cases $t\,p\left(t\right)$ is nonnegative and supported
on $\left(-\infty,1\right)$, with the corresponding Laplace inverse
demand functions $p\left(t\right)$ listed in our 
Supplementary Material \ref{sub:LaplaceInverseDemandFunctions}, which also contains additional discussion. Note that for most of the inverse demand functions listed in the
theorem, it is also possible to prove complete monotonicity using
Theorems 1--6 of \citet{miller2001completely}.\hfill{}\Square{}

\noindent \textbf{Proof of Theorem \ref{TheoremAbsenceOfCompleteMonotonicityOfDemandSpecification} (Absence of Complete Monotonicity of Demand Specification) .} The statement of the theorem follows by inspection of the Laplace inverse
demand functions, as in the previous proof. Additional discussion may be found in 
Supplementary Material \ref{sub:LaplaceInverseDemandFunctions}.\hfill{}\Square{}

\section{Details of the Generalized EOQ Model Estimation}
\label{AppendixDetailsOfTheGeneralizedEOQModelEstimation}

Here we provide additional details of the estimation of the cost parameter \(\beta =(1-\gamma )/(2-\gamma )\). As mentioned in the main text,
we selected industries that included at least 10 firms satisfying our criteria. The corresponding confidence intervals corresponding to individual industries 
are plotted in Figure \ref{FigureConfidenceIntervalsForBeta}.

In principle, the value of average estimated $\beta $ could be 
sensitive to the cutoff on the number of firms per industry. Table 
\ref{TableSensitivityToCutoff} summarizes the dependence of the resulting average $\beta $
on the choice of the cutoff. It turns out that the average $\beta $ remains roughly the same even for large changes of the cutoff on the number of firms.

 \begin{table}
\begin{centering}
\begin{tabular}{|c|c|c|c|}
\hline \(N_{f,\min }\) & \(N_I\) & \(\beta\) & \(\sigma _{\beta }\)\tabularnewline
\hline \hline \
5 & 192 & 0.39 & 0.20 \
\tabularnewline\hline\
10 & 70 & 0.39 & 0.12 \
\tabularnewline \hline\
15 & 45 & 0.39 & 0.10 \
\tabularnewline \hline\
20 & 23 & 0.41 & 0.10 \
\tabularnewline \hline\
25 & 14 & 0.39 & 0.10 \
\tabularnewline \hline\
30 & 11 & 0.42 & 0.07 \
\tabularnewline \hline\
35 & 9 & 0.42 & 0.08 \\ \hline
\end{tabular}\caption{Sensitivity to the cutoff \(N_{f,\min }\) of the number of firms per industry. The cutoff influences
the number of industries \(N_I\) that satisfy the sample selection criteria and the resulting mean \(\beta\) and the corresponding standard deviation
\(\sigma _{\beta }\).}\label{TableSensitivityToCutoff}\end{centering}
\hspace*{0.5ex} \end{table}

The estimated value of \(\gamma\) could be, in principle, also influenced by seasonality patterns. To investigate this issue, we construct a measure of 
seasonality of individual industries. In particular, we calculate a Herfindahl-like seasonality index based on the shares of trade in individual months of the year, defined as \(H_s=\sum
_{i=1}^{12} v_i^2\), where \(v_i\) is the average share of month \(i\) in the average annual trade value. A high value of the index means that trade
flows are very unevenly distributed across months. Then we regress \(\gamma\) on this measure. We find that the 95\% confidence interval of the
slope coefficient is [-0.69,1.21] and the corresponding p-value is 0.58. For robustness, we change the cutoff to 5 firms, getting { }the confidence
interval [-0.91,0.30] and the p-value of 0.32. In both cases we do not reject the hypothesis that the slope coefficient is zero. The data is plotted
in Figure \ref{FigureInfluenceOfSeasonality}.

 \begin{figure}
\begin{subfigure}{.5\textwidth}
\begin{centering}
\includegraphics[scale=0.69]{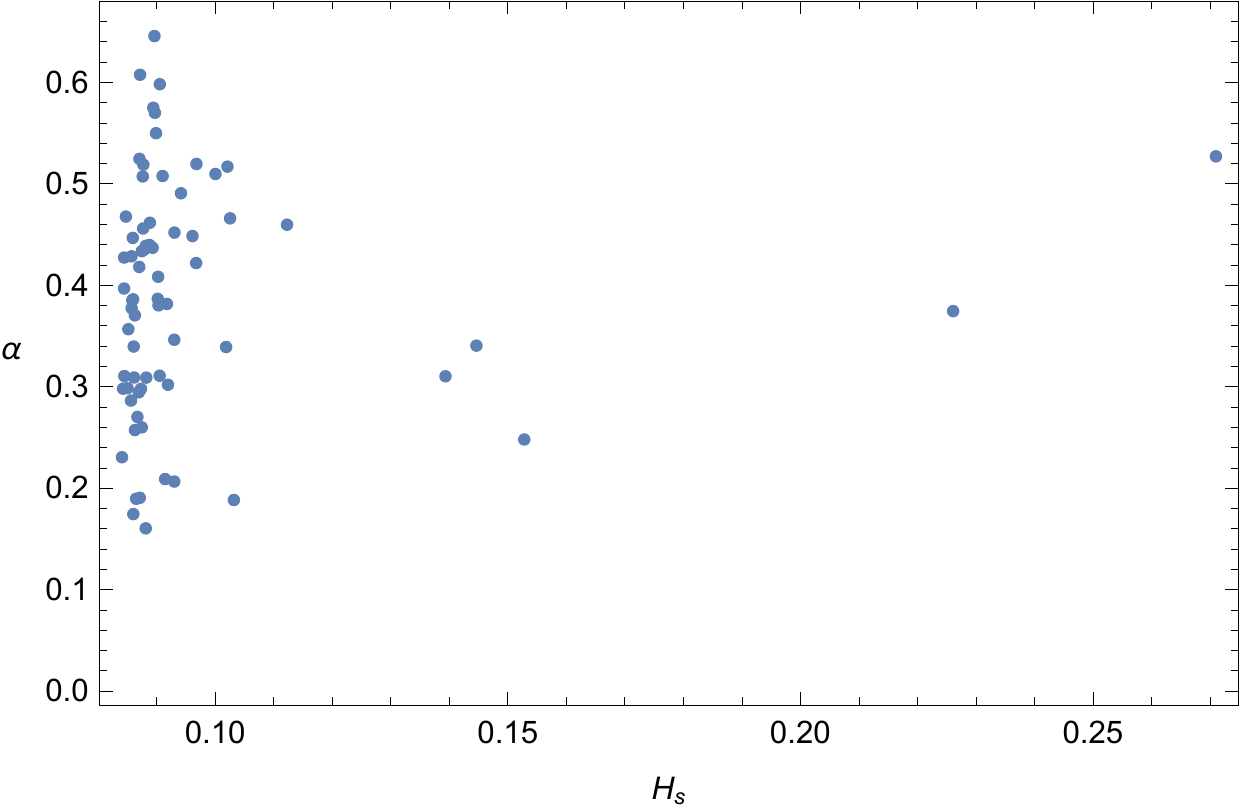}
\caption{}
\end{centering}
\end{subfigure}
\begin{subfigure}{.5\textwidth}
\begin{centering}
\includegraphics[scale=0.69]{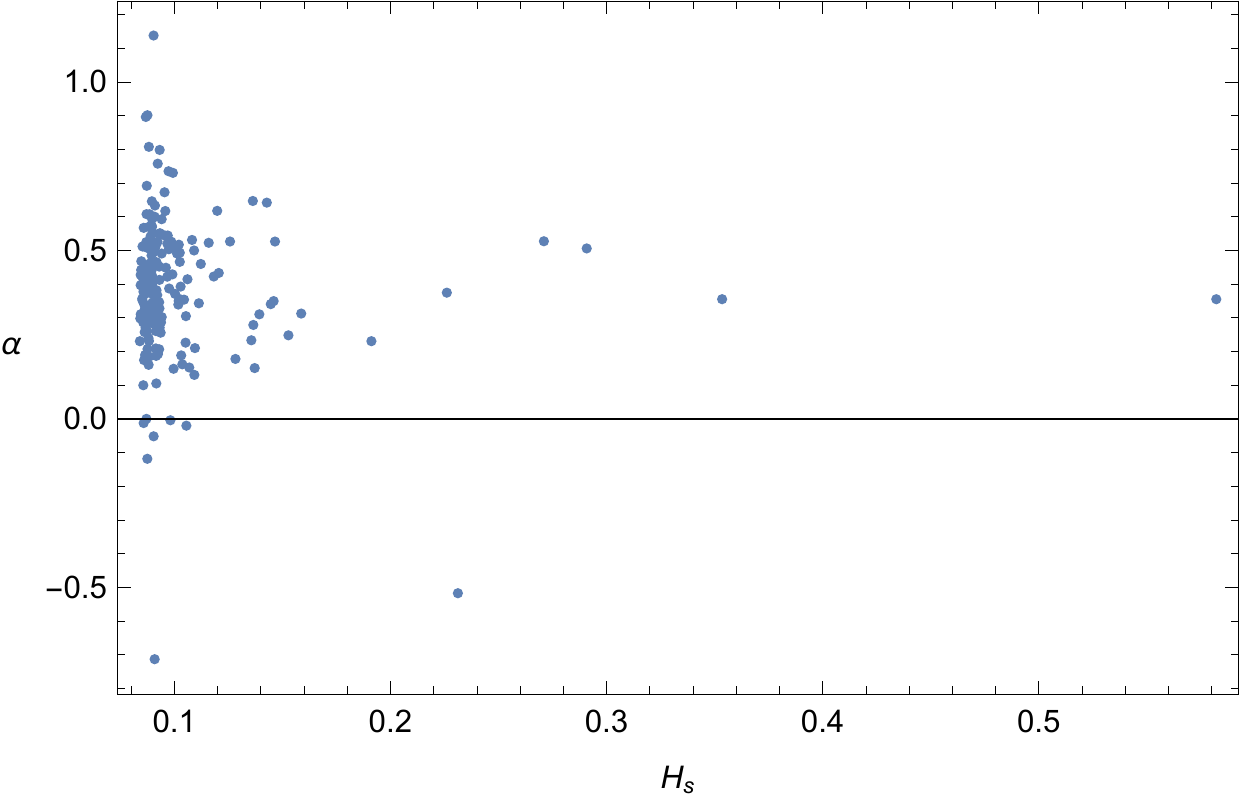}
\caption{}
\end{centering}
\end{subfigure}
\caption{The relationship of the cost exponent \(\alpha\) for specific industries and the industry seasonality index \(H_s\). Figure
(a) corresponds to the sample used for the main estimation, which is based on industries with at least 10 firms satisfying the sample selection criteria.
We do not observe any systematic pattern relating \(\alpha\) and \(H_s\). Figure (b) corresponds to a cutoff set to 5 firms as a robustness check.
Also, in this case the values of \(\alpha\) do not seem to be influenced by \(H_s\).}\label{FigureInfluenceOfSeasonality}
\end{figure}

 \section{World Trade}\label{AppendixWorldTradeFlows}

 \subsection{Details of data construction}\label{AppendixWorldTradeFlowsDetailsOfDataConstruction} 

Here we provide details of the data construction for Section  \ref{SectionWorldTrade}. The economies used to fit our model are, in descending order by 2006 GDP, United States, Japan, Germany, China, United
Kingdom, France, Italy, Canada, Spain, Brazil, Russia, South Korea, Mexico, India, Australia, Netherlands, Turkey, Switzerland, Sweden, Belgium,
Saudi Arabia, Norway, Poland, Austria, Denmark, Greece, South Africa, Iran, Argentina, Ireland, Nigeria, United Arab Emirates, Thailand, Finland,
Portugal, Hong Kong, Venezuela, Malaysia, Colombia, Czech Republic, Chile, Israel, Singapore, Pakistan, Romania, Algeria, Hungary, New Zealand, Kuwait,
Peru, Kazakhstan, Bangladesh, Morocco, Vietnam, Qatar, Slovakia, Croatia, Ecuador, Luxembourg, Slovenia, Dominican Republic, Oman, Belarus, Tunisia,
Bulgaria, Syria, Sri Lanka, Serbia/Serbia and Montenegro, Lithuania, Guatemala, Kenya, Costa Rica, Lebanon, Latvia, Azerbaijan, Cyprus, Ghana, Uruguay,
Yemen, Tanzania, El Salvador, Bahrain, Trinidad and Tobago, Panama, Cameroon, Ivory Coast, Iceland, Estonia, Ethiopia, Jordan, Macau, Zambia, Bosnia
and Herzegovina, Bolivia, Jamaica, Uganda, Honduras, Paraguay, Gabon, and Senegal. These countries were selected based on data availability. We computed
the tradable share (percentage) of GDP by selecting tradable sectors from United Nations gross value added database. We fit the GDP in the model
to the tradable portion of GDP, computed as GDP reported by IMF World Economic Outlook database multiplied by the tradable share of GDP.\footnote{In the model, all imports are consumed domestically, which implies that exports cannot be larger than tradable GDP. However,
such situation might arise for small, highly open economies. To avoid this discrepancy, when the calculated portion of tradable GDP that goes to
domestic consumption is smaller than five percent of the tradable GDP in the data, we increment it so that it reaches that level. This is done by
correspondingly increasing both the (adjusted) tradable GDP and the (adjusted) consumption in the economy. This criterion was satisfied for just
one economy, Hong Kong. Of course, a more realistic way of modeling this situation is to include multi-stage production and/or multi-stage transportation
in the model. This will require some additional research work, but it is a clear direction to pursue.} {
}This means that, for example, education revenue and education expenditures are not counted towards the model{'}s GDP and expenditures, which is
appropriate for a model designed to capture manufacturing and similar industries. Multi-sector extensions including services are, of course, possible.
Also, note that we exclude re-imports and re-exports from the trade flows data.

 \subsection{Related literature}\label{AppendixWorldTradeFlowsRelatedLiterature} 

Here we briefly discuss connections of the results of Subsection  \ref{TheGravityEquationOfTradeAndTheDependenceOfTradeCostsOnDistance} to related issues in the literature, as mentioned in Footnote  \ref{FootnotePointingToAppendixWorldTradeFlowsRelatedLiterature}.  \citet*{helpman2008estimating}
{ }studied the role of the extensive margin of trade for the estimation of the distance-dependence of trade costs based on world trade flows. The
authors found the distance effect to be 27 to 30 percent smaller than in benchmark estimates based on the gravity equation of trade without extensive
margin effects. Although this is an important correction, it is not enough to resolve the trade cost puzzle. We get much stronger effects because
of the increasing marginal cost of production. Moreover, unlike that paper we do not need unrealistically high export market entry costs that would
be inconsistent with the everyday experience that even sole entrepreneurs with very limited capital (for example, 25,000 USD) are able to start an
import/export business, a fact that is explained in many resources, such as  \citet{importexportbusiness} . 

Separately,  \citet{arkolakis2010market} { }builds an
elegant model of international trade where fixed costs of exporting are indeed negligible (and marginal costs of production are constant). Even though
the demand is CES, some firms choose not to export to a particular destination because before serving a customer, they need to pay a sizeable per-customer
advertising cost, which can make serving that customer unprofitable. An argument against this mechanism is that it would not work if targeted advertising
was possible. Empirical evidence in the industrial organization literature shows that the main portion of observed aggregate demand elasticity comes
from heterogeneity in the consumers{'} valuation of products, not from elasticity of demand by a given individual; an individual{'}s demand is quite
inelastic in the data. If firms could reach high-valuation customers and advertize directly to them, they would export to that destination. Especially
in recent years targeted advertising via the Internet is quite easy and widespread, so it is hard to justify the modeling assumption that it is impossible.
For this reason, it is better to think of the insightful paper  \citet{arkolakis2010market} { }in a more abstract way: as an investigation of situations where effective demand departs from CES. In principle, we could
remove economies of scale in shipping from our model and instead modify the demand. In this case, again, we could combine this with our assumption
of increasing marginal costs of production, and using our proposed tractable functional forms for demand we could proceed with computations in the
same way. But of course, we already have empirical evidence on the economies of scale in shipping, and we know that logistics costs as a proportion
of world GDP are very large. Note that the influential study of export decisions  \citet*{eaton2011anatomy} { }also uses the  \citet{arkolakis2010market}
{ }mechanism in theoretical modeling.

 \subsection{Firm export patterns}\label{AppendixWorldTradeFlowsFirmExportPatterns} 

Here we mention other possible mechanisms potentially leading to patterns similar to those in Figure \ref{FigureExportDestinationsByTwoIdenticalFirms} of Subsection 
\ref{ChoiceOfExportDestinations}, as referenced in Footnote 
\ref{FootnotePointingToAppendixWorldTradeFlowsFirmExportPatterns}. If the countries significantly
differ and we break the symmetry between the firms (in terms of how their products enter utility functions), it is possible to explain patterns resembling
those in Figure  \ref{FigureExportDestinationsByTwoIdenticalFirms}.
For example, windows imported by Finland are likely to be very different from windows imported by Portugal. If a firm specializes in only one kind
of windows, it is natural for them to export to only one of these destinations. Another possible phenomenon that could lead to similar patterns in
the data would be distribution centers in export destinations. For example, a firm may serve both Spain and Portugal from one distribution center
based in Spain. In that case international trade flow data would not record such sales in Portugal as exports to Portugal, but instead as exports
to Spain and then exports from Spain to Portugal. Yet another possibility is the case of very large firms. If these firms were so large that monopolistic
competition description of the market was inappropriate and we needed to model it as an oligopoly, there could be an alternative explanation for
choosing different export destinations. In this case strategic effects of market entry could potentially play a role. A firm may not choose to serve
Greece because Greece is already served by its rival and the market the is not profitable enough for two firms to enter. The puzzle would still remain
for smaller firms that cannot influence the entire industry. More generally, these three explanations may be valid in some cases but are not powerful
enough to explain the majority of the empirical regularity in the data, especially in the case of smaller firms that directly export goods that are
not geographically specialized. Case studies of individual exporters also make it clear that the export pattern is typically not explained by those
three explanations. A detailed investigation of these issues will be reported separately.

\section{Applications}\label{AppendixBreadthOfApplication}

\subsection{Supply chains with hold-up \citep{antras}\label{AppendixSupplyChainsWithHoldup}}

We consider a generalization of the supply chain model of 
\citet[henceforth AC]{antras}. Instead of the variables introduced
in the original paper, we use a different set of variables that makes
 the mathematics and intuition substantially simpler.\footnote{The relationship
 between our variables introduced in the next paragraph 
 (in a notation compatible with the rest of this paper)
 and the variables in \citet{antras} is as follows.
Let us use the symbol $\tilde{q}$ to refer to a quantity measure denoted
$q$ in AC, which is \emph{distinct} from what we call 
effective quality-adjusted quantity $q$.
In order to recover AC's original model as a special case, we identify
their output $\tilde{q}$ with $q^{1/\alpha}$, where $\alpha\in\left(0,1\right)$
is a constant defined there. For the present discussion we do not
need $q$ to be linearly proportional to the number of units produced.
It is just some measure of the output, which may or may not be quite
abstract. A similar statement applies to the customized intermediate
input. Our measure $q_{s}\left(j\right)$ of a particular input is
related to AC's measure $x$$\left(j\right)$ by $q_{s}\left(j\right)=\theta^{\alpha}(x\left(j\right))^{\alpha}$,
where $\theta$ is a positive productivity parameter defined in their
original paper.}

A firm produces a final good by sequentially 
using a continuum of customized inputs each provided by a different
supplier indexed by $j\in[0,1]$, with small $j$ representing initial stages of production (upstream) 
and large $j$
representing final stages (downstream). If production proceeds smoothly, the {\em effective}
quality-adjusted quantity $q$ of the final good is the integral of
the effective quality-adjusted quantity contributed by intermediate input $j$, which we denote
$q_{s}\left(j\right)$:  $q=\int_{0}^{1}q_{s}(j)\, dj$. This effective quantity
represents both the quantity of the good and its quality level. But we will
refer to it simply as ``quality'', since this will make the discussion sound more natural.
If production
is ``disrupted'' by the failure of some supplier $\overline{j}\in[0,1)$
to cooperate, then only the quality accumulated to that point in the
chain is available, with all further quality-enhancement impossible: 
$q=\int_{0}^{\overline{j}}q_{s}(j)\, dj$. The firm faces an inverse demand function $P(q)$, which 
does not necessarily have to be decreasing because, for example,
consumers may have little willingness-to-pay for an improperly
finished product. If there is no disruption in production, $q=q(1)$.

Following the property rights theory of the firm \citep{grossmanhart,hartmoore,antrasalone},
 input production requires relationship-specific
investments. The marginal revenue from additional quality brought by
supplier $j$, $MR\left(q(j)\right)q_{s}(j)$ is therefore split between
the firm and supplier $j$, where $MR=P+P'q$.\footnote{See AC's Subsection 3.1 for a discussion of why only marginal revenue, and not the full-downstream revenue, is the pie that is bargained over and an alternative micro-foundation of this model.} In particular, the
supplier receives a fraction $1-\beta(j)$ (its bargaining power).

The cost of producing quality $q_s(j)$ is homogeneous across suppliers and equal to $C(q_{s}(j))$, which is assumed 
strictly convex.\footnote{The AC model corresponds to $C(q_{s})=(q_{s})^{1/\alpha}c/\theta$,
where $c$ and $\theta$ are positive constants defined in their paper.
In our notation, the suppliers' cost is convex but their contributions
towards the final output are linear. In the original paper the suppliers'
cost is linear, but their contributions towards the final output have
diminishing effects. These are two alternative interpretations of
the same economic situation from the point of view of two different
systems of notation. As mentioned before, in our interpretation, the
product of a supplier is $q_{s}$, whereas in the original paper the
supplier's product is $x$, related to $q_{s}$ by $q_{s}\left(j\right)=\theta^{\alpha}(x\left(j\right))^{\alpha}$.
}
  Thus the first-order condition of supplier $j$ equates the share
of marginal revenue she bargains for with her marginal cost:
\begin{equation}
MC\left(q_{s}(j)\right)\equiv C'\left(q_{s}(j)\right)=\left[1-\beta\left(j\right)\right]MR\left(q\left(j\right)\right).\label{eq:SupplierFOC}
\end{equation}
The cost to the firm of obtaining a contribution $q_{s}(j)$ from
supplier $j$ is, therefore, the surplus it must leave in order to induce
$q_{s}(j)$ to be produced, $q_{s}MC\left(q_{s}(j)\right)$.%

The firm  chooses $\beta\left(j\right)$ through the nature of the contracting relationship optimally for each supplier
to maximize its profits. Following AC and \citet{antrashelpman1,antrashelpman2}, we mostly focus on the \emph{relaxed}
problem where $\beta\left(j\right)$ may be adjusted freely and continuously.
This provides most of the intuition for what happens when the firm
is constrained to choose between two discrete levels of $\beta$ corresponding to outsourcing (low $\beta$) and insourcing (high $\beta$) and may be more realistic given the complexity of real-world contracting \citep{boundaries}.
   Note that by convexity, $MC'>0$, while each $q_{s}$ makes a linearly
separable contribution to $q$. Thus for any fixed $q$ the firm wants
to achieve, it does so most cheaply by setting all $q_{s}=q$ by Jensen's
Inequality. Thus Equation \ref{eq:SupplierFOC}
becomes, at any optimum $q^{\star}$,
\begin{equation}
\beta^{*}\left(j\right)=1-\frac{MC\left(q^{\star}\right)}{MR\left(jq^{\star}\right)}.\label{betasolution}
\end{equation}
From this we immediately see that $\beta^{\star}$ is co-monotone
with $MR$: in regions where marginal revenue is increasing, $\beta^{\star}$
will be rising and conversely when marginal revenue is decreasing.
The marginal revenue associated with constant elasticity demand is
in a constant ratio to inverse demand.
This implies AC's principal result that when revenue elasticity is
less than unity the firm will tend to outsource upstream and when revenue
elasticity is less than unity the firm will tend to outsource downstream.  However, it seems natural
to think that $P(q)$ would initially rise, as consumers are willing
to pay very little for a product that is nowhere near completion,
and would eventually fall as the product is completed according to
the standard logic of downward-sloping demand. We now solve in an equally-simple form a model allowing this richer logic.

Equation \ref{betasolution} implies that the surplus left to each
supplier is $q_{s}MC(q)$ and thus total cost is $qMC(q)$. The
problem reduces to choosing $q$ to maximize revenue $qP(q)$ less
cost $qMC(q)$, giving first-order condition
\begin{equation}
MR(q)=MC\left(q\right)+q\, MC'\left(q\right).\label{firmfoc}
\end{equation}
This differs from the familiar neoclassical first-order condition
$MR\left(q\right)=MC\left(q\right)$ only by the presence of the (positive)
term $q\, MC'\left(q\right)$. Note that $MC+qMC'$ bears the same
relationship to $MC$ that $MC$ bears to $AC$; this equation therefore
similarly inherits the tractability properties of the standard monopoly
problem.  The reason is that the hold-up makes multi-part tariff pricing impossible, creating a linear-price monopsony structure by forcing the firm to pay suppliers the marginal cost of the last unit of quality for all units produced.

Let us now consider  $P(q)=p_{-t}q^{t}+p_{-u}q^{u}$ and $MC(q)=mc_{-t}q^{t}+mc_{-u}q^{u}$.
This includes AC's specification as the special case when $p_{-t}=0$
and $mc_{-u}=0$ so that each has constant elasticity.\footnote{In particular, in their notation, AC have $t=\frac{1}{\alpha}$, $u=1+\frac{\rho}{\alpha}$,
$mc_{-t}=\nicefrac{c}{\alpha\theta}$ and $p_{-u}=A^{1-\rho}$, where
$\theta$ and $\rho$ is are positive constants defined in AC, not
to be confused with the pass-through rate denoted by $\rho$ or the
conduct parameter denoted by $\theta$ in other parts of this paper.%
} However,  let us focus instead on the case when $t,u,mc_{-u},p_{-t}>0=mc_{-t}>p_{-u}$
and $u>t$ so that the first term of the inverse demand dominates
at small quantities while the second dominates at large quantities. The expression resulting for $\beta^{*}\left(j\right)$ is:

\begin{equation}
\beta^{*}\left(j\right)=1-\frac{1}{(1+u)\left[\left(1-\frac{p_{-u}}{mc_{-u}}\right)j^{t}+\frac{p_{-u}}{mc_{-u}}j^{u}\right]}.\label{eq:ResultingEffectiveBargainingPower}
\end{equation}
Note that because $mc_{-u}>0>p_{-u}$, the first denominator
term is positive and the second denominator term is
negative. This implies that at small $j$ (where $j^{t}$ dominates),
$\beta^{\star}$ increases in $j$, while at large $j$, it decreases
in $j$. In the AC complements case when $p_{-u}=0$, or even if $p_{-u}$
is sufficiently small, this large $j$ behavior is never manifested
and all outsourcing (low $\beta^{\star}$) occurs at early stages.
Also note that only the ratio of coefficients $\frac{p_{-u}}{mc_{-u}}$
matters for the sourcing pattern; $p_{-t}$ is irrelevant, as
the joint level of $p_{-u}$ and $mc_{-u}$. 

\begin{figure}
\begin{centering}
\includegraphics[width=3.7in]{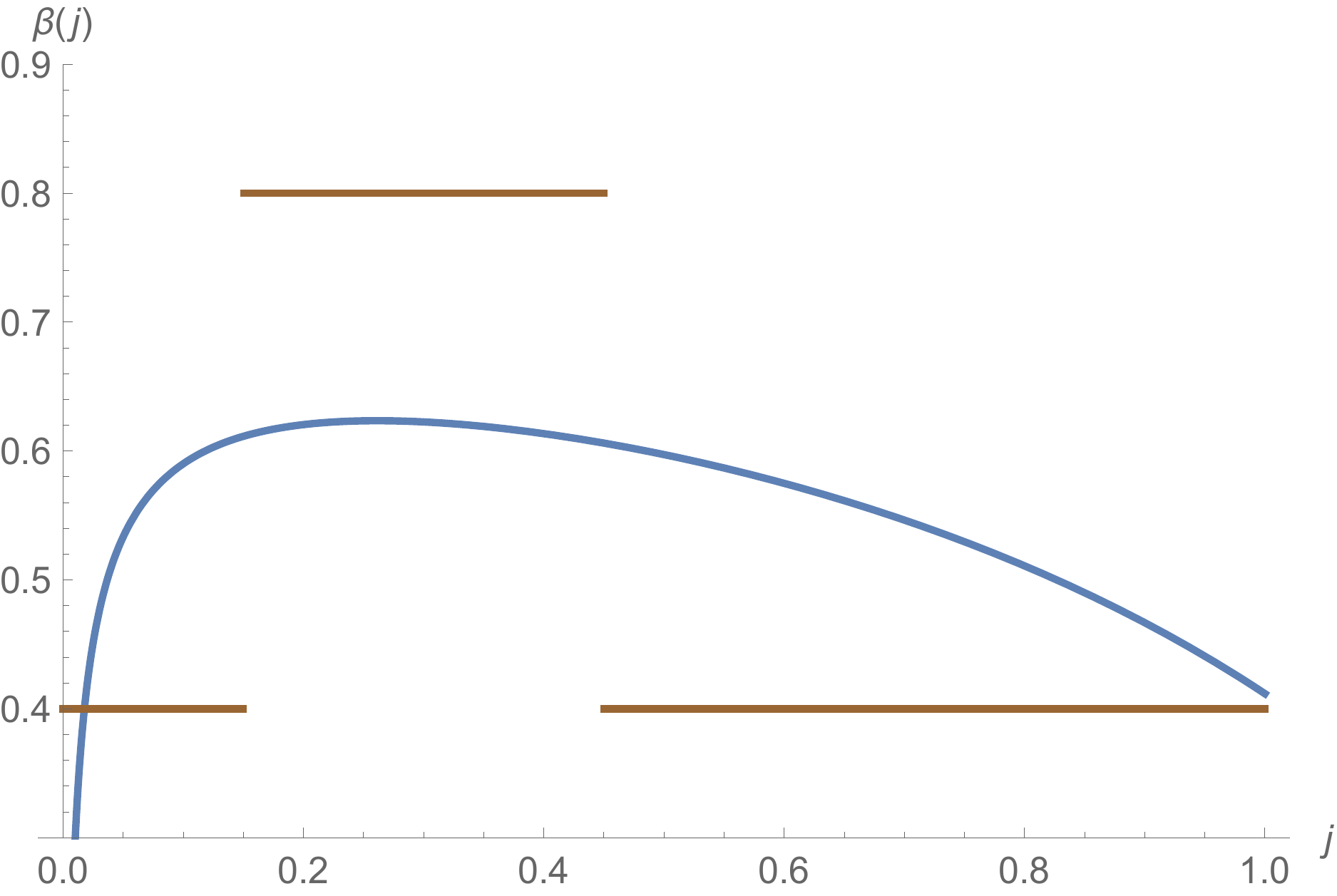}
\par\end{centering}

\caption{Optimal relaxed and restricted $\beta^{\star}$ in the AC model when
$t=0.35,u=0.7,\frac{p_{-u}}{mc_{-u}}=-4$.}

\label{antrasfigure}
\end{figure}

However, for many parameters an inverted U-shape emerges. For example,
Figure \ref{antrasfigure} shows the case when $t=0.35,u=0.7,p_{-t}=1.8,\frac{p_{-u}}{mc_{-u}}=-4$.
The curve corresponds to the shape of the relaxed solution. Depending
on precisely which values of $\beta$ we take insourcing and outsourcing
to correspond to, this can lead to insourcing in the middle of the
production and outsourcing at either end.
In Supplementary Material \ref{AppendixSupplyChainsWithHoldupDetails} we 
study in detail the constrained problem using largely closed-form methods for the 
case when outsourcing gives $\beta_{O}=0.8$ and
insourcing gives $\beta_{I}=0.4$. This is illustrated by the lines
in Figure \ref{antrasfigure}, which show the constrained optimum.
This gives the same qualitative answer as the relaxed problem, as expected.  

\subsection{Imperfectly competitive supply chains\label{sub:Imperfectly-competitive-supply}}

The models that founded the field of industrial organization were
\citet{cournot}'s of symmetric oligopoly and complementary monopoly.
Equilibrium in these models is characterized by 
\[
P+\theta P'q=MC.
\]
Under Cournot competition, $\theta=\nicefrac{1}{n}$, where $n$ is
the number of competing firms and $MC$ is interpreted as the common
marginal cost of all producers. Under Cournot complements (which does
not require symmetry) $\theta=m$, where $m$ is the number of complementary
producers and $MC$ is interpreted as the aggregated marginal cost
of all producers.\footnote{With constant marginal cost, and in some other special cases, the asymmetric Cournot competition model may also be solved if both demand  is specified in an appropriate form.  To maintain the generality of our analysis we do not discuss this solvable, asymmetric special case.} Note that $P+\theta P'q$ is just a linear combination
of $P$ and $P'q$ and thus has the same form as either of these components
in a form-preserving class of functional forms. Thus either problem
yields exactly the same characterization of tractability as the monopoly
problem. 

In the last half-century a variant on \citet{cournot}'s complementary
monopoly problem proposed by \citet{spengler} has been more commonly
used. In this model one firm sells an input to another who in turn
sells to a consumer. The difference from \citeauthor{cournot}'s model
is principally in the timing; namely the ``upstream'' firm is assumed
to set her price prior to the downstream firm. In this case the upstream
firm effectively sets part of the downstream firm's marginal cost.
Her first-order condition is 
\[
P+P'q=MC+\hat P,
\]
where $\hat P$ is the sales price set by the upstream firm. Thus the effective
inverse demand faced by the upstream firm is $\hat{P}(q)\equiv P(q)+P'(q)q-MC(q)$.
The upstream firm then solves a monopoly problem with this inverse
demand. This yields an upstream marginal revenue curve bearing the
same relationship to $\hat{P}$ that $MR$ bears to $P$. Because
the form-preserving feature may be applied an arbitrary number of
times, however, this transformation does not change our characterization
of tractability. Thus a form-preserving class has the same tractability
characterization in \citeauthor{spengler}'s model as in the standard
\citeauthor{cournot} model. 

We can go further and allow for many layers of production and arbitrary
imperfect competition (or complements) at each later as in \citet{salinger}.
The same characterization of tractability continues to apply. In Supplementary Material
\ref{sequential} we provide an explicit expression for the coefficients in
the polynomial equation for any tractable form. \citet{adachi2014cost, adachi2014double}
argue that flexible functional forms are particularly important in
such models because many important and policy-relevant properties
are imposed by standard tractable forms. For example, the markup
of the upstream firm in \citeauthor{spengler}'s model is identical
to that of the two firms if they merged under the BP demand class,
but the upstream firm will typically charge a lower markup than an
integrated firm under reasonable conditions 
(bell-shaped-distribution-generated demand and
U-shaped cost curves).

\newpage

$\ $

\bigskip

\begin{center}

 \LARGE \bf Supplementary Material
 
 \end{center}

\bigskip
\bigskip

\section{Laplace Inverse Demand Functions}

\label{sub:LaplaceInverseDemandFunctions}The following table contains
Laplace inverse demand functions corresponding to inverse demand functions
used in the literature. Although for most Laplace inverse demand functions
we include only a few terms, closed-form expressions for all terms
exist. Here $p_{a}$ refers to a mass-point of magnitude $p_{a}$ at location
$a$. In the alternative notation on the lower lines, $\delta(x-a)$
refers to a mass-point of magnitude 1 at location $a$, i.e. to a
Dirac delta function centered at $a$. We use standard notation for special functions: $\Gamma$ stands for the gamma function and $W$ for the Lambert W function.

\[
\begin{array}{l}
\begin{array}{ccc}
\text{Constant elasticity / Pareto:} & q(P)=\left(\frac{P}{\beta}\right)^{-\epsilon} & P(q)=\beta q^{-1/\epsilon}\end{array}\\
\begin{array}{ccc}
\text{} & p(t): & \begin{array}{c}
p_{\frac{1}{\epsilon}}=\beta\end{array}\end{array}\\
\begin{array}{ccc}
\text{} & p(t): & \beta\delta\left(t-\frac{1}{\epsilon}\right)\end{array}\\
\begin{array}{ccc}
\text{Constant pass-through / BP:} & q(P)=\left(\frac{P-\mu}{\beta}\right)^{-\epsilon} & P(q)=\mu+\beta q^{-1/\epsilon}\end{array}\\
\begin{array}{ccc}
\text{} & p(t): & \begin{array}{ccc}
p_{0}=\mu, & p_{\frac{1}{\epsilon}}=\beta\end{array}\end{array}\\
\begin{array}{ccc}
\text{} & p(t): & \beta\delta\left(t-\frac{1}{\epsilon}\right)+\mu\delta(t)\end{array}\\
\begin{array}{ccc}
\text{Gumbel distribution:} & q(P)=\exp\left(-\exp\left(\frac{P-\alpha}{\beta}\right)\right) & P(q)=\alpha+\beta\log(-\log(q))\end{array}\\
\begin{array}{ccc}
\text{} & p(t): & \begin{array}{ccccc}
p_{0}=\mu, & p(t)=-\frac{\beta}{t} & \text{for} & t<0\end{array}\end{array}\\
\begin{array}{ccc}
\text{} & p(t): & \alpha\delta(t)-\frac{\beta1_{t<0}}{t}\end{array}\\
\begin{array}{ccc}
\text{Weibull distribution:} & q(P)=e^{-\left(\frac{P}{\beta}\right)^{\alpha}} & P(q)=\beta(-\log(q))^{\frac{1}{\alpha}}\end{array}\\
\begin{array}{ccc}
\text{} & p(t): & \begin{array}{ccc}
\frac{(-1)^{\frac{1}{\alpha}}\beta t^{-\frac{1}{\alpha}-1}}{\Gamma\left(-\frac{1}{\alpha}\right)} & \text{for} & t<0\end{array}\end{array}\\
\begin{array}{ccc}
\text{} & p(t): & \frac{(-1)^{\frac{1}{\alpha}}\beta1_{t<0}t^{-\frac{1}{\alpha}-1}}{\Gamma\left(-\frac{1}{\alpha}\right)}\end{array}\\
\begin{array}{ccc}
\text{Fr{\' e}chet distribution:} & q(P)=1-e^{-\left(\frac{P-\mu}{\beta}\right)^{-\alpha}} & P(q)=\mu+\beta(-\log(1-q))^{-1/\alpha}\end{array}\\
\begin{array}{ccc}
\text{} & p(t): & \begin{array}{ccccccccc}
p_{0}=\mu, & p_{\frac{1}{\alpha}}=\beta, & p_{\frac{1}{\alpha}-1}=-\frac{\beta}{2\alpha}, & p_{\frac{1}{\alpha}-2}=\frac{\beta}{8\alpha^{2}}-\frac{5\beta}{24\alpha}, & \text{...}\end{array}\end{array}\\
\begin{array}{ccc}
\text{} & p(t): & \left(\frac{\beta}{8\alpha^{2}}-\frac{5\beta}{24\alpha}\right)\delta\left(t-\frac{1}{\alpha}+2\right)+\beta\delta\left(t-\frac{1}{\alpha}\right)-\frac{\beta\delta\left(t-\frac{1}{\alpha}+1\right)}{2\alpha}+\mu\delta(t)+\dots\end{array}\\
\begin{array}{ccc}
\text{Logistic distribution:} & q(P)=\left(\exp\left(\frac{P-\mu}{\beta}\right)+1\right)^{-1} & P(q)=\mu-\beta\log\left(\frac{1}{1-q}-1\right)\end{array}\\
\begin{array}{ccc}
\text{} & p(t): & \begin{array}{ccccccccccccc}
p_{0}=\mu, & p_{0}^{\text{(1)}}=-\beta, & p_{-1}=-\beta, & p_{-2}=-\frac{\beta}{2}, & p_{-3}=-\frac{\beta}{3}, & p_{-4}=-\frac{\beta}{4}, & \text{...}\end{array}\end{array}\\
\begin{array}{ccc}
\text{} & p(t): & -\beta\sum_{j=1}^{\infty}\frac{\delta(j+t)}{j}+\mu\delta(t)-\beta\delta'(t)\end{array}\\
\begin{array}{ccc}
\text{Log-logistic distribution:} & q(P)=\left(\left(\frac{P}{\sigma}\right)^{\gamma}+1\right)^{-1} & P(q)=\sigma\left(\frac{q}{1-q}\right)^{-1/\gamma}\end{array}\\
\begin{array}{ccc}
\text{} & p(t): & \begin{array}{ccccccccc}
p_{\frac{1}{\gamma}}=\sigma, & p_{\frac{1}{\gamma}-1}=-\frac{\sigma}{\gamma}, & p_{\frac{1}{\gamma}-2}=\frac{\sigma}{2\gamma^{2}}-\frac{\sigma}{2\gamma}, & p_{\frac{1}{\gamma}-3}=-\frac{\sigma}{6\gamma^{3}}+\frac{\sigma}{2\gamma^{2}}-\frac{\sigma}{3\gamma}, & \text{...}\end{array}\end{array}\\
\begin{array}{ccc}
\text{} & p(t): & \left(\frac{\sigma}{2\gamma^{2}}-\frac{\sigma}{2\gamma}\right)\delta\left(t-\frac{1}{\gamma}+2\right)+\sigma\delta\left(t-\frac{1}{\gamma}\right)-\frac{\sigma\delta\left(t-\frac{1}{\gamma}+1\right)}{\gamma}+\dots\end{array}\\
\begin{array}{ccc}
\text{Laplace distribution (}q<\frac{1}{2}\text{):} & q(P)=\frac{1}{2}\exp\left(\frac{\mu-P}{\beta}\right) & P(q)=\mu-\beta\log(2q)\end{array}\\
\begin{array}{ccc}
\text{} & p(t): & \begin{array}{ccc}
p_{0}=\mu-\beta\log(2), & p_{0}^{\text{(1)}}=-\beta\end{array}\end{array}\\
\begin{array}{ccc}
\text{} & p(t): & \delta(t)(\mu-\beta\log(2))-\beta\delta'(t)\end{array}\\
\begin{array}{ccc}
\text{Laplace distribution (}q>\frac{1}{2}\text{):} & q(P)=1-\frac{1}{2}\exp\left(\frac{P-\mu}{\beta}\right) & P(q)=\mu+\beta\log(2(1-q))\end{array}\\
\begin{array}{ccc}
\text{} & p(t): & \begin{array}{ccccccccccc}
p_{0}=\beta\log(2)+\mu, & p_{-1}=-\beta, & p_{-2}=-\frac{\beta}{2}, & p_{-3}=-\frac{\beta}{3}, & p_{-4}=-\frac{\beta}{4}, & \text{...}\end{array}\end{array}\\
\begin{array}{ccc}
\text{} & p(t): & \delta(t)(\beta\log(2)+\mu)-\beta\sum_{j=1}^{\infty}\frac{\delta(j+t)}{j}\end{array}\\
\begin{array}{ccc}
\text{Normal distribution:} & q(P)=\text{erfc}\left(\frac{P-\mu}{\sqrt{2}\sigma}\right) & P(q)=\mu-\sqrt{2}\sigma\text{erfc}^{-1}(2-q)\end{array}\\
\begin{array}{ccc}
\text{} & p(t): & \begin{array}{ccccccc}
p_{0}^{\text{(1)}}=-\sqrt{\frac{\pi}{2}}\sigma, & p_{0}^{\text{(2)}}=-\frac{1}{2}\sqrt{\frac{\pi}{2}}\sigma, & p_{0}^{\text{(3)}}=\frac{1}{24}\left(-\sqrt{2}\pi^{3/2}-2\sqrt{2\pi}\right)\sigma, & \text{...}\end{array}\end{array}\\
\begin{array}{ccc}
\text{} & p(t): & -\sqrt{\frac{\pi}{2}}\sigma\delta'(t)-\frac{1}{2}\sqrt{\frac{\pi}{2}}\sigma\delta''(t)+\frac{1}{24}\left(-\sqrt{2}\pi^{3/2}-2\sqrt{2\pi}\right)\sigma\delta^{(3)}(t)+\dots\end{array}\\
\begin{array}{ccc}
\text{Lognormal distribution:} & q(P)=\text{erfc}\left(\frac{\log(P)-\mu}{\sqrt{2}\sigma}\right) & P(q)=\exp\left(\mu-\sqrt{2}\sigma\text{erfc}^{-1}(2-q)\right)\end{array}\\
\begin{array}{ccc}
\text{} & p(t): & \begin{array}{ccccc}
p_{0}^{\text{(1)}}=\sqrt{\frac{\pi}{2}}\left(-e^{\mu}\right)\sigma\delta'(t), & p_{0}^{\text{(2)}}=\frac{1}{4}\pi e^{\mu}\sigma^{2}-\frac{1}{2}\sqrt{\frac{\pi}{2}}e^{\mu}\sigma, & \text{...}\end{array}\end{array}\\
\begin{array}{ccc}
\text{} & p(t): & \left(\frac{1}{4}\pi e^{\mu}\sigma^{2}-\frac{1}{2}\sqrt{\frac{\pi}{2}}e^{\mu}\sigma\right)\delta''(t)-\sqrt{\frac{\pi}{2}}e^{\mu}\sigma\delta'(t)+\dots\end{array}
\end{array}
\]
\[
\begin{array}{l}
\begin{array}{ccc}
\text{Almost Ideal Demand System:} & q(P)=\frac{\alpha+\beta\log(P)}{P} & P(q)=-\frac{\beta W\left(-\frac{qe^{-\frac{\alpha}{\beta}}}{\beta}\right)}{q}\end{array}\\
\begin{array}{ccc}
\text{} & p(t): & \begin{array}{ccccccccccc}
p_{0}=e^{-\frac{\alpha}{\beta}}, & p_{-1}=\frac{e^{-\frac{2\alpha}{\beta}}}{\beta}, & p_{-2}=\frac{3e^{-\frac{3\alpha}{\beta}}}{2\beta^{2}}, & p_{-3}=\frac{8e^{-\frac{4\alpha}{\beta}}}{3\beta^{3}}, & p_{-4}=\frac{125e^{-\frac{5\alpha}{\beta}}}{24\beta^{4}}, & \text{...}\end{array}\end{array}\\
\begin{array}{ccc}
\text{} & p(t): & \frac{125e^{-\frac{5\alpha}{\beta}}\delta(t+4)}{24\beta^{4}}+\frac{8e^{-\frac{4\alpha}{\beta}}\delta(t+3)}{3\beta^{3}}+\frac{3e^{-\frac{3\alpha}{\beta}}\delta(t+2)}{2\beta^{2}}+e^{-\frac{\alpha}{\beta}}\delta(t)+\frac{e^{-\frac{2\alpha}{\beta}}\delta(t+1)}{\beta}+\dots\end{array}\\
\begin{array}{ccc}
\text{Constant superelasticity:} & q(P)=\left(\epsilon\log\left(\frac{\theta-1}{\theta P}\right)+1\right)^{\frac{\theta}{\epsilon}} & P(q)=\frac{(\theta-1)e^{\frac{1}{\epsilon}-\frac{q^{\epsilon/\theta}}{\epsilon}}}{\theta}\end{array}\\
\begin{array}{ccc}
\text{} & p(t): & \begin{array}{ccccccccc}
p_{0}=e^{\frac{1}{\epsilon}}-\frac{e^{\frac{1}{\epsilon}}}{\theta}, & p_{-\frac{\epsilon}{\theta}}=\frac{e^{\frac{1}{\epsilon}}}{\theta\epsilon}-\frac{e^{\frac{1}{\epsilon}}}{\epsilon}, & p_{-\frac{2\epsilon}{\theta}}=\frac{e^{\frac{1}{\epsilon}}}{2\epsilon^{2}}-\frac{e^{\frac{1}{\epsilon}}}{2\theta\epsilon^{2}}, & p_{-\frac{3\epsilon}{\theta}}=\frac{e^{\frac{1}{\epsilon}}}{6\theta\epsilon^{3}}-\frac{e^{\frac{1}{\epsilon}}}{6\epsilon^{3}}, & \text{...}\end{array}\end{array}\\
\begin{array}{ccc}
\text{} & p(t): & \left(\frac{e^{\frac{1}{\epsilon}}}{2\epsilon^{2}}-\frac{e^{\frac{1}{\epsilon}}}{2\theta\epsilon^{2}}\right)\delta\left(t+\frac{2\epsilon}{\theta}\right)+\delta(t)\left(e^{\frac{1}{\epsilon}}-\frac{e^{\frac{1}{\epsilon}}}{\theta}\right)+\left(\frac{e^{\frac{1}{\epsilon}}}{\theta\epsilon}-\frac{e^{\frac{1}{\epsilon}}}{\epsilon}\right)\delta\left(t+\frac{\epsilon}{\theta}\right)+\dots\end{array}\\
\begin{array}{ccc}
\text{Cauchy distribution:} & q(P)=\frac{\tan^{-1}\left(\frac{a-P}{b}\right)}{\pi}+\frac{1}{2} & P(q)=a+b\tan\left(\pi\left(\frac{1}{2}-q\right)\right)\end{array}\\
\begin{array}{ccc}
\text{} & p(t): & \begin{array}{ccccccccccccc}
p_{1}=\frac{b}{\pi}, & p_{0}=a, & p_{-1}=-\frac{\pi b}{3}, & p_{-3}=-\frac{\pi^{3}b}{45}, & p_{-5}=-\frac{2\pi^{5}b}{945}, & p_{-7}=-\frac{\pi^{7}b}{4725}, & \text{...}\end{array}\end{array}\\
\begin{array}{ccc}
\text{} & p(t): & a\delta(t)+\frac{b\delta(t-1)}{\pi}-\frac{1}{3}\pi b\delta(t+1)-\frac{1}{45}\pi^{3}b\delta(t+3)-\frac{2}{945}\pi^{5}b\delta(t+5)-\frac{\pi^{7}b\delta(t+7)}{4725}+\dots\end{array}\\
\begin{array}{ccc}
\text{Singh Maddala distribution:} & q(P)=\left(\left(\frac{P}{b}\right)^{a}+1\right)^{-\tilde{q}} & P(q)=b\left(q^{-\frac{1}{\tilde{q}}}-1\right)^{\frac{1}{a}}\end{array}\\
\begin{array}{ccc}
\text{} & p(t): & \begin{array}{ccccccccc}
p_{\frac{1}{a\tilde{q}}}=b, & p_{-\frac{a-1}{a\tilde{q}}}=-\frac{b}{a}, & p_{-\frac{2a-1}{a\tilde{q}}}=\frac{b}{2a^{2}}-\frac{b}{2a}, & p_{-\frac{3a-1}{a\tilde{q}}}=-\frac{b}{6a^{3}}+\frac{b}{2a^{2}}-\frac{b}{3a}, & \text{...}\end{array}\end{array}\\
\begin{array}{ccc}
\text{} & p(t): & \left(\frac{b}{2a^{2}}-\frac{b}{2a}\right)\delta\left(\frac{2a-1}{a\tilde{q}}+t\right)+b\delta\left(t-\frac{1}{a\tilde{q}}\right)-\frac{b\delta\left(\frac{a-1}{a\tilde{q}}+t\right)}{a}\end{array}+\dots\\
\begin{array}{ccc}
\text{Tukey lambda distribution:} & q(P)=P^{(-1)}(P) & P(q)=\frac{(1-q)^{\lambda}-q^{\lambda}}{\lambda}\end{array}\\
\begin{array}{ccc}
\text{} & p(t): & \begin{array}{ccccccccccc}
p_{-\lambda}=-\frac{1}{\lambda}, & p_{0}=\frac{1}{\lambda}, & p_{-1}=-1, & p_{-2}=\frac{\lambda}{2}-\frac{1}{2}, & p_{-3}=-\frac{\lambda^{2}}{6}+\frac{\lambda}{2}-\frac{1}{3}, & \text{...}\end{array}\end{array}\\
\begin{array}{ccc}
\text{} & p(t): & \left(-\frac{\lambda^{2}}{6}+\frac{\lambda}{2}-\frac{1}{3}\right)\delta(t+3)+\frac{\delta(t)}{\lambda}+\left(\frac{\lambda}{2}-\frac{1}{2}\right)\delta(t+2)-\frac{\delta(t+\lambda)}{\lambda}-\delta(t+1)+\dots\end{array}\\
\begin{array}{ccc}
\text{Wakeby distribution:} & q(P)=P^{(-1)}(P) & P(q)=\mu-\frac{\gamma\left(1-q^{-\delta}\right)}{\delta}+\frac{\alpha\left(1-q^{\beta}\right)}{\beta}\end{array}\\
\begin{array}{ccc}
\text{} & p(t): & \begin{array}{ccccc}
p_{0}=\frac{\alpha}{\beta}-\frac{\gamma}{\delta}+\mu, & p_{-\beta}=-\frac{\alpha}{\beta}, & p_{\delta}=\frac{\gamma}{\delta}\end{array}\end{array}\\
\begin{array}{ccc}
\text{} & p(t): & \delta(t)\left(\frac{\alpha}{\beta}-\frac{\gamma}{\delta}+\mu\right)-\frac{\alpha\delta(t+\beta)}{\beta}+\frac{\gamma\delta(t-\delta)}{\delta}+\dots\end{array}
\end{array}\label{TableOfInverseLaplaceTransforms}
\]

Here we provide clarification of some of the expressions in the table above. In many cases the terms in the Laplace inverse demand were obtained
utilizing series expansions, as discussed in Subsection  \ref{UsingTaylorSeriesExpansion} of this supplementary material. More explicitly, we utilized the following series representations of the Laplace inverse demand.

\[\begin{array}{ll}
 \text{Fr{\' e}chet distribution:} & P(q)=\mu +\beta  q^{-1/\alpha } \sum _{k=0}^{\infty } \binom{-\frac{1}{\alpha }}{k} \left(\sum _{j=2}^{\infty
} j^{-1} q^{j-1}\right){}^k \\
 \text{Log-logistic distribution:} & P(q)=\sigma  q^{1/\gamma } \sum _{n=0}^{\infty } (-1)^n \binom{\frac{1}{\gamma }}{n} q^n \\
 \text{Almost Ideal Demand System:} & P(q)=\beta  \sum _{n=0}^{\infty } \frac{(-1-n)^n}{(1+n)!} \left(-\frac{e^{-\frac{\alpha }{\beta }}}{\beta }\right)^{1+n}
q^n \\
 \text{Constant superelasticity:} & P(q)=\sum _{n=0}^{\infty } \frac{e^{1/\epsilon } (-\epsilon )^{-n} (\theta -1)}{\theta  n!} q^{\frac{n \epsilon
}{\theta }} \\
 \text{Cauchy distribution:} & P(q)=a+\frac{b}{\pi  q}+b \sum _{k=1}^{\infty } \frac{(-1)^k 2^{2 k} \pi ^{-1+2 k} B_{2 k}}{(2 k)!} q^{-1+2 k} \\
 \text{Singh-Maddala distribution:   } & P(q)=b q^{-\frac{1}{a \tilde{q}}} \sum _{n=0}^{\infty } (-1)^n \binom{\frac{1}{a}}{n} q^{\frac{n}{\tilde{q}}}

\end{array}\]

\[\text{}\]
 We used the standard notation for generalized binomial coefficients and denoted Bernoulli numbers by \(B_{2 k}\). Constant
superelasticity refers to the inverse demand function introduced by  \citet{klenowwillis}. The Gumbel distribution is also known as the type I extreme value distribution, the Fr{\' e}chet distribution as type II
extreme value, and the Weibull distribution as type III extreme value. 

In the case of the Gumbel distribution, the Laplace inverse demand is not an ordinary function, but a distribution (generalized function) in the
sense of the distribution theory by Laurent Schwartz. For this reason, we use regularization to give a precise meaning to integrals involving the
Laplace inverse demand function that was schematically written in the previous table. We provide three regularization prescriptions and illustrate
them for the case of Laplace inverse demand itself. { }The first prescription is

\[P(q)=\lim_{a\to 0^+} \, \left(\int _{-\infty }^{\infty }(\alpha -\beta  (\gamma +\log (a))) \delta (t)dt+\int_{-\infty }^a \frac{\beta  q^{-t}}{t}
\, dt\right).\]
 Here we moved the upper bound in the second integral beyond zero and added the regularization term proportional to \(\log
(a)\). The second integral is to be interpreted in the sense of principal value.\footnote{In
Mathematica, principal value integrals may be computed by choosing the option \textit{ PrincipalValue }\textit{ \(\to\)}\textit{  True} for the \textit{
Integrate} function.} { }It is straightforward to verify that this expression leads to the correct expression
for \(P(q)\). First, we evaluate the integrals to get 

\[P(q)=\lim_{a\to 0^+} \, (\alpha -\gamma  \beta +\beta  \text{Ei}(-a \log (q))-\beta  \log (a)).\]
 Here $\gamma $ is the Euler gamma, \(\gamma \approx 0.577216\), and Ei stands for the special function called exponential
integral. Evaluating the limit then leads to the correct expression

\[P(q)=\alpha +\beta  \log (-\log (q)).\]

The second prescription is analogous and shifts the upper bound of the second integral to negative numbers:

\[P(q)=\lim_{a\to 0^+} \, \left(\int _{-\infty }^{\infty }(\alpha -\beta  (\gamma +\log (a))) \delta (t)dt+\int_{-\infty }^{-a} \frac{\beta  q^{-t}}{t}
\, dt\right).\]
 Evaluating the integral gives

\[P(q)=\lim_{a\to 0^+} \, (\alpha -\gamma  \beta -\beta  \Gamma (0,-a \log (q))-\beta  \log (a)),\]
 and taking the limit leads again to the correct expression. Here $\Gamma $ is the incomplete gamma function. 

The third prescription is computationally most convenient because it does not involve taking a limit. The regularizing term is expressed in the form
of an integral

\[P(q)=\int _{-\infty }^{\infty }(\alpha -\beta  \gamma ) \delta (t)dt+\beta  \int_{-\infty }^0 \frac{q^{-t}-1_{t>-1}}{t} \, dt.\]
 Here \(1_{t>-1}\) is an indicator function. The integral may then be computed directly, again leading to the correct
expression. The same methods may be used when interpreting other integrals involving generalized functions that behave as \(1/t\) close to \(t=0\).

For the normal and lognormal distributions the Laplace inverse demand functions are again not ordinary functions. They were obtained using Taylor
series expansions of the error function. Expressions involving these Laplace inverse demand functions may need to be summed using the Euler summation
method to ensure proper convergence.

The expressions above may be used to straightforwardly derive Theorems  \ref{TheoremCompleteMonotonicityOfDemandSpecification} and  \ref{TheoremAbsenceOfCompleteMonotonicityOfDemandSpecification} of the main paper. An alternative, but sometimes less straightforward way is to utilize Theorems 1--6 of 
\citet{miller2001completely}. If we are interested in the monotonicity properties of the pass-through
rate, we can use the corollary in Section \ref{SectionArbitraryDemandAndCostFunctions}. However, we also identified more direct ways to prove monotonicity properties of the pass-through rate for certain demand
functions. These proofs are included in Section  \ref{forms}
of this supplementary material.

\section{Evaluation of Inverse Laplace Transform}

In the main text of the paper we used the term Laplace-log transform to emphasize that this is Laplace transform in terms of the logarithm of an economic quantity. In this section, which focuses on mathematical issues, we use the term Laplace transform, keeping the economic interpretation implicit.

\subsection{Numerical Evaluation of Inverse Laplace Transform}

There exist many methods for numerical evaluation of inverse Laplace
transform, now usually integrated into mathematical and statistical
software. For a summary and important references, see, e.g., Chapter
6 of \citet{egonmwan2012numerical}. Note that just like other types
of non-parametric methods, numerical inversion of Laplace transform
requires some regularization, such as the \citet{tikhonov1963solution}
regularization. This is because Laplace transform inversion is a so-called
\emph{ill-posed} problem, which means that there exist large changes
in the inverse Laplace transform that lead to only small changes in
the original function in the domain of interest. For a classic discussion,
see \citet{bellman1966numerical}.

\subsection{Analytic Evaluation of Inverse Laplace Transform}

Mathematical software allows for symbolic inversion of Laplace transform.\footnote{ The corresponding functions are \emph{InverseLaplaceTransform} in
Mathematica, \emph{ilaplace} in MATLAB, or \emph{inverse\_laplace\_transform}
in Python (SymPy).} However, it is often more convenient to evaluate the inverse Laplace
transform using more direct methods.

\subsubsection{Using Taylor series expansion}
\label{UsingTaylorSeriesExpansion}

We would like to emphasize that finding analytic expressions for Laplace
inverse demand is often much simpler than it seems since it many
cases it only requires finding a Taylor series expansion of a definite
function. Consider, for example, the case of log-logistic distribution
of valuations included in the Supplementary Material \ref{sub:LaplaceInverseDemandFunctions},
which corresponds to inverse demand $P\left(q\right)=\sigma(\frac{q}{1-q})^{-1/\gamma}$.
This may be written as $P\left(q\right)=\sigma q^{-1/\gamma}\times\left(1-q\right)^{1/\gamma}$,
i.e. a product of a power function and a function that has a well-defined
Taylor expansion at $q=0$: 
\[
\left(1-q\right)^{1/\gamma}=1-\frac{1}{\gamma}\ q+\frac{1-\gamma}{2\gamma^{2}}\ q^{2}+...=\sum_{n=0}^{\infty}(-1)^{n}\binom{\frac{1}{\gamma}}{n}q^{n},
\]
where the $n$th term contains a generalized binomial coefficient.
This immediately translates into
\[
P\left(q\right)=\sigma q^{-\frac{1}{\gamma}}-\frac{\sigma}{\gamma}\ q^{1-\frac{1}{\gamma}}+\frac{\left(1-\gamma\right)\sigma}{2\gamma^{2}}\ q^{2-\frac{1}{\gamma}}+...=\sum_{n=0}^{\infty}(-1)^{n}\binom{\frac{1}{\gamma}}{n}q^{n-\frac{1}{\gamma}},
\]
and from here we can read off the masses at points $t=\frac{1}{\gamma},\frac{1}{\gamma}-1,\frac{1}{\gamma}-2,$...
that together constitute the Laplace inverse demand included in Supplementary Material
\ref{sub:LaplaceInverseDemandFunctions}.

\subsubsection{Using the traditional inverse Laplace transform formula}

The readers may be familiar with the traditional inverse Laplace transform
formula based on the Bromwich integral in the complex plane:\footnote{Here $i$ is the imaginary unit and $\gamma$ is a real number large
enough to ensure that $F\left(s\right)$ is holomorphic in the half-plane
$\mbox{Re\ }s>\gamma$ (or has a holomorphic analytic continuation
to this half-plane). }
\begin{equation}
f\left(t\right)=\frac{1}{2\pi i}\lim_{T\rightarrow\infty}\int_{\gamma-iT}^{\gamma+iT}e^{st}f_{\mathcal{L}}\left(s\right)ds,\mbox{\  where \ensuremath{f_{\mathcal{L}}\left(s\right)\equiv\int_{0}^{\infty}e^{-st}f\left(t\right)dt.}}\label{eq:BromwichIntegral}
\end{equation}
For the purposes of this paper we did not actually need it. We obtained
the Laplace inverse demand function listed in Supplementary Material \ref{sub:LaplaceInverseDemandFunctions}
by simpler methods.

\subsubsection{Piecewise inverse Laplace transform}

Readers familiar with Fourier transform but not
with Laplace transform may potentially be concerned about applicability of our
approach to the case of linear demand. Our prescription is simple:
If the inverse demand takes the form $P\left(q\right)=a-bq$, we restrict
our attention to $q\in\left(0,\frac{a}{b}\right)$, without affecting
the form of Laplace inverse demand $p\left(t\right)$. The reason
why this is possible is that to evaluate $p\left(t\right)$ using,
say, Equation \ref{eq:BromwichIntegral}, we do not need the values
of $P\left(q\right)\equiv P\left(e^{s}\right)$ for $s\in\left(-\infty,\infty\right)$,
as a superficial analogy with Fourier transform might suggest. Instead,
the integral in Equation \ref{eq:BromwichIntegral} is in the imaginary
direction. Writing the inverse demand as $P\left(e^{s}\right)=a-be^{s}$
for $\mbox{Re }s<\log\frac{a}{b}$ and $P\left(e^{s}\right)=0$ for
$\mbox{Re}\ s>\frac{a}{b}$ and working with each piece separately
will not make the Laplace inverse demand complicated. We will just
have two different Laplace inverse demand functions, each valid for
a range of $q$.

\section{Solving Cubic and Quartic Equations Simply}

The readers may have seen general formulas for solutions to cubic
and quartic equations that looked very complicated. It turns out that
the intimidating look is caused just by shifts and rescalings of variables.
Solving these equations is actually very straightforward:

\paragraph{Cubic equations.}

To solve the equation $x^{3}+3ax+2=0$, we substitute $x\equiv y^{1/3}-ay^{-1/3}$,
which leads to the quadratic equation $y^{2}+2y-a^{3}=0$ with solutions
$y=\pm\sqrt{1+a^{3}}-1$. Given this result, the solutions to any
other cubic equation may be obtained by rescaling and shifting of
$x$.\footnote{What we described here is a version of Vieta's substitution that we
customized to avoid cluttering of various rational factors. }

\paragraph{Quartic equations.}

A quartic equation of the form $x^{4}+ax^{2}+bx+1=0$ is equivalent
to $(x^{2}+\sqrt{\alpha}x+\beta)(x^{2}-\sqrt{\alpha}x+\beta^{-1})=0$
with $\alpha\ge0$ and $\beta$ chosen to satisfy $\beta+\beta^{-1}=a+\alpha$
and $\beta^{-1}-\beta=\frac{b}{\sqrt{\alpha}}$ so that the coefficients
of different powers of $x$ match. If we substitute the right-hand-side
expressions into the trivial identity $(\beta+\beta^{-1})^{2}-(\beta-\beta^{-1})^{2}=4$,
we get a cubic equation for $\alpha$, which we know how to solve.
With the help of the quadratic formula, a solution for $\alpha$ then
translates into a solution for $\beta$, and consequently for $x$.
Given these results, the solutions to any other quartic equation may
be obtained simply by rescaling and shifting of $x$.

 \section{Discussion of the Characterization Theorem}\label{AppendixDiscussionOfTheCharacterizationTheorem} 

Here we present a discussion of the logic behind Theorem  \ref{formpreserve} that only requires knowledge of very elementary calculus, i.e. simple differentiation and integration, instead of assuming
knowledge of complex analysis and functional transforms.\footnote{We are grateful to an anonymous
referee for the excellent suggestion that this is possible and worth doing. In the original version of the paper we only included the proof of Theorem
\ref{formpreserve} based on functional transforms. That proof is quick and easy, but less pedagogical because it requires knowledge
of more advanced mathematics.} 

 \subsection{One-dimensional functional form classes}

Consider a real function \(P(q)\) (on an open interval of positive \(q\)) that belongs to a one-dimensional functional form class invariant under
(i.e., preserved by) average-marginal transformations. It must be the case that { }\(a P(q)+b q P'(q)\) also belongs to this class. For this reason,
{ }if \(q P'(q)\) were not a multiple of \(P(q)\), the class would not be one-dimensional, which would be a contradiction. Denoting the coefficient
of proportionality as \(A\), we get the differential equation\[q P'(q)=A P(q),\]
which implies\[\frac{dP}{P}=A \frac{dq}{q}\, \, \, \, \, \Longrightarrow \, \, \, \, \, \log \left| P\right| =A \log  q+\text{const.}\, \, \, \, \, \Longrightarrow
\, \, \, \, \, P(q)=c_1 q^A,\]
with some real constant \(c_1\). We conclude that the one-dimensional functional form classes invariant under the average-marginal
transformations are the classes of power functions with a fixed exponent.

In the next subsection we work in terms of the variable \(s=\log  q\). For clarity, let us repeat the computation above using \(s\), with the identification
\(H(s)\equiv P(q)\) for \(s=\log  q\). The differential equation becomes \(H'(s)=A H(s)\) and leads to the same result as above:\[\frac{dH}{H}=A\, ds\, \, \, \, \, \Longrightarrow \, \, \, \, \, \log \left| H\right| =A s+\text{const.}\, \, \, \, \, \Longrightarrow \, \, \,
\, \, H(s)=c_1 e^{A s}.\]

 \subsection{Two-dimensional functional form classes}

Any member \(H(\log (q))=P(q)\) of a two-dimensional class invariant under average-marginal transformations must satisfy the following differential
equation:\[H''(s)=A H(s)+B H'(s).\]
The logic is analogous to the one-dimensional case in the previous subsection. If \(H'(s)\) and \(H(s)\) are proportional
to each other, the equation is automatically satisfied. Consider the case where \(H'(s)\) and \(H(s)\) are not proportional to each other. By the
invariance property, any linear combination of \(H(s)\) and \(H'(s)\) must belong to the class. Consequently, any linear combination of \(H(q)\),
\(H'(s)\), and \(H''(s)\) must belong to the class. These linear combinations span a three-dimensional space, { }unless \(H''(s)\) is a linear combination
of { }\(H(q)\) and \(H'(s)\), i.e. unless the differential equation is satisfied for some \(A\) and \(B\). This establishes the validity of the differential
equation.

Denote by \(r_1\) and \(r_2\) the two roots of the quadratic equation\[x^2=A +B\text{  }x,\]
namely\[r_1=\frac{1}{2} \left(B+\sqrt{4 A+B^2}\right),\, \, \, \, \, r_2=\frac{1}{2} \left(B-\sqrt{4 A+B^2}\right).\]
It is straightforward to check that the differential equation \(H''(s)=A H(s)+B H'(s)\) may be alternatively written as\[\left(1-r_1\frac{d}{ds}\right)\left(\left(1-r_2\frac{d}{ds}\right)H(s)\right)=0.\]
The function \(\left(1-r_2\frac{d}{ds}\right)H(s)\), which we denote \(f(s)\), therefore satisfies the differential equation\[\left(1-r_1\frac{d}{ds}\right)f(s)=0.\]
This differential equation gives the solution\[f(s)-r_1f'(s)=0\Longrightarrow ds=r_1\frac{df}{f}\Longrightarrow \log  |f|=\frac{s}{r_1}+\text{const}.\Longrightarrow f(s)=\tilde{c}_1 e^{\frac{s}{r_1}},\]
which means\[\left(1-r_2\frac{d}{ds}\right)H(s)= \tilde{c}_1 e^{\frac{s}{r_1}}.\]
To get the final expression for \(H(q)\), we will solve this differential equation in two alternative cases.

 \subsubsection{The case of two distinct roots} 

Let us consider the case where the two roots, \(r_1\) and \(r_2\), are not equal. To solve this last differential equation, let us perform the substitution
\(H(s)=e^{\frac{s}{r_1}}g(s)+\frac{1 }{r_1-r_2}\tilde{c}_1r_1e^{\frac{s}{r_1}}\). The differential equation then becomes\[\left(1-r_2\frac{d}{ds}\right)\left(e^{\frac{s}{r_1}}g(s)+\frac{e^{\frac{s}{r_1}} \tilde{c}_1r_1}{-r_1+r_2}\right)= \tilde{c}_1 e^{\frac{s}{r_1}},\]
or after canceling terms proportional to \(\tilde{c}_1\), \[\left(1-r_2\frac{d}{ds}\right)\left(e^{\frac{s}{r_1}}g(s)\right)=0.\]
This differential equation has the following solution:\[\left(1-\frac{r_2}{r_1}-r_2\frac{d}{ds}\right)g(s)=0\, \, \, \, \, \Longrightarrow \, \, \, \, \, g(s)=c_2 e^{\left(1-\frac{r_2}{r_1}\right)\frac{s}{r_2}}\,
\, \, \, \, \Longrightarrow \, \, \, \, \, g(s)=c_2 e^{s \left(\frac{1}{r_2}-\frac{1}{r_1}\right)}.\]
For the function \(H(s)\) this implies\[H(s)=c_2 e^{ \frac{s}{r_2}}+\frac{ \tilde{c}_1r_1}{r_1-r_2}e^{\frac{s}{r_1}}.\]
After introducting the notation \(c_1=\frac{ 1}{r_1-r_2}\tilde{c}_1r_1\), the expression for \(H(s)\) becomes\[H(s)=c_1e^{\frac{s}{r_1}}+c_2 e^{ \frac{s}{r_2}}.\]
If both roots \(r_1\) and \(r_2\) are real, then this is the desired final form. If they have a non-zero imaginary part,
we would like to manipulate the expression further. In this case \(\frac{1}{r_1}\) and \(\frac{1}{r_2}\) are complex conjugates, which means we may
write them as \[\frac{1}{r_1}=a_R+a_I i,\text{$\, \, \, \, \, $ }\frac{1}{r_2}=a_R-a_I i.\]
For \(H(s)\) to be real, we also need \(c_1\) and \(c_2\) to be complex conjugates:\[c_1=c_R-c_I i,\text{$\, \, \, \, \, $ }c_2=c_R+c_I i.\]
\(H(s)\) then becomes\[H(s)=\left(c_R-c_I i\right)e^{a_Rs}\left(\cos \left(a_Is\right)+i \sin \left(a_Is\right)\right)+\left(c_R+c_I i\right) \left(\cos \left(a_Is\right)-i
\sin \left(a_Is\right)\right),\]
which after canceling terms gives the desired final form\[H(s)=c_Re^{a_Rs}\cos \left(a_Is\right)+c_I e^{a_Rs} \sin \left(a_Is\right).\]

 \subsubsection{The case of a double root} 

Now let us consider the case where the two roots are equal: \(r_1=r_2\). In this case, we need to solve the equation\[\left(1-r_1\frac{d}{ds}\right)H(s)= \tilde{c}_1 e^{\frac{s}{r_1}}\]
We perform the substitution { }\(H(s)=e^{\frac{s}{r_1}}g(s)\), which leads to\[\left(1-r_1\frac{d}{ds}\right)\left(e^{\frac{s}{r_1}}g(s)\right)= \tilde{c}_1 e^{\frac{s}{r_1}}\, \, \, \, \Longrightarrow \, \, \, \, \, -r_1\frac{d}{ds}g(s)=
\tilde{c}_1\text{ $\, \, \, \, $}\Longrightarrow \, \, \, \, \, g(s)=- \frac{\tilde{c}_1}{r_1} s + \text{const}.\]
The result for \(H(s)\) is \(H(s)=e^{\frac{s}{r_1}}(- \frac{\tilde{c}_1}{r_1} s + \text{const}.)\), which, after renaming
the constants, gives the desired final form\[H(s)=e^{\frac{s}{r_1}}\left(c_2+c_1s\right).\]

We see that overall, the resulting characterization of two-dimensional classes of functional forms invariant under average-marginal transformations
is consistent with the statement of Theorem  \ref{formpreserve}.

 \subsection{Higher-dimensional functional form classes}

The same method may be used to derive higher-dimensional form-preserving classes of functional forms. The differential equations one needs to solve
are standard and have known solutions that utilize the properties of the \textit{ characteristic equations} of the differential equations, which
are analogous to the equation \(x^2=A +B\text{  }x\) we used above. The proof of Theorem  \ref{formpreserve} described in Appendix  \ref{AppendixProofsOfTheorems} may be thought of as using such differential equations, just represented in a different, transformed way as the vanishing
of the expression in Equation  \ref{eq:T0GS} inside
the interval \(S\). Solving the differential equations using transforms is much quicker and more convenient. \\

\section{ Applications} \label{+apps}

\subsection{Supply chains with hold-up: the restricted problem}\label{AppendixSupplyChainsWithHoldupDetails}

Extending our analysis of the Antr\`as and Chor model in Appendix \ref{AppendixSupplyChainsWithHoldup}, 
we consider the solution of the restricted AC model in the case where the firm is 
restricted to two discrete levels of bargaining power corresponding to ``out-sourcing'' and ``in-sourcing''.
As in the relaxed solution, consider the optimal choice of a path
for $\beta$ subject to producing a total quality $\hat{q}$. Note
that $q(j;\beta)$ is a strictly increasing function of $j$ for any
path of $\beta$ achieving $\hat{q}$ by definition. Thus it is equivalent,
instead of solving for the optimal restricted $\beta$ for each $j$,
to solve for the optimal $\beta^{\star\star}$ for each $q\left(j;\beta\right)\in\left[0,\hat{q}\right]$
and then invert the resulting $q\left(j;\beta^{\star\star}\right)$
function to recover the value optimal $\beta$ at each $j$. This
method preserves the separability we used in the relaxed problem and
thus greatly simplifies the restricted problem. Wherever it does not
create confusion we suppress as many arguments as possible, especially
the dependence on $\beta$, to preserve notational economy.

By the same arguments as in the restricted case, the cost of production
$\hat{q}$ is $C\left(\hat{q};\beta\right)$ where
\[
C\left(\hat{q};\beta\right)=\int_{0}^{\hat{q}}\left[1-\beta(q)\right]MR(q)dq,
\]
where $\beta(q)$ is a notationally-abusive contraction of $\beta\left(j\left(q;\beta\right)\right)$.
However, to actually produce $\hat{q}$, we need
\[
\int_{0}^{1}S\left(\left[1-\beta\left(q(j)\right)\right]MR\left(q(j)\right)\right)dj=\hat{q},
\]
where $S=MC^{-1}$, the supply curve, exists because of our assumption
that $MC$ is strictly monotone increasing. Changing variables so
that both integrals are taken over $j$:
\[
C\left(\beta\right)=\int_{0}^{1}\left[1-\beta\left(q(j)\right)\right]MR\left(q(j)\right)S\left(\left[1-\beta\left(q(j)\right)\right]MR\left(q(j)\right)\right)dj.
\]

Thus the firm solves a Lagrangian version of this problem that is
separable in each $j$, or equivalently $q$:
\[
\max_{\beta}\int_{0}^{1}\lambda S\left(\left[1-\beta\left(q(j)\right)\right]MR\left(q(j)\right)\right)-\left(\left[1-\beta(q)\right]MR(q)S\left(\left[1-\beta\left(q(j)\right)\right]MR\left(q(j)\right)\right)\right)dj-\lambda\hat{q}.
\]
At each $q$ this is a simple maximization problem. The firm chooses
the value of $\beta$ maximizing
\[
\lambda S\left(\left[1-\beta\left(q\right)\right]MR\left(q\right)\right)-\left[1-\beta(q)\right]MR(q)S\left(\left[1-\beta\left(q\right)\right]MR\left(q\right)\right),
\]
the difference between the total value of the production by that firm
and the total cost of that production. Clearly, both terms are decreasing
in $\beta$ given that $MR>0$ in any range where the firm would consider
producing, so given that the firm chooses between only two values
of $\beta$, $\beta_{I}>\beta_{O}$, the firm will strictly choose
in-sourcing if and only if
\begin{equation}
MR(q)>\frac{\lambda\left[S\left(\left[1-\beta_{O}\right]MR\left(q\right)\right)-S\left(\left[1-\beta_{I}\right]MR\left(q\right)\right)\right]}{\left[1-\beta_{O}\right]S\left(\left[1-\beta_{O}\right]MR\left(q\right)\right)-\left[1-\beta_{I}\right]S\left(\left[1-\beta_{I}\right]MR\left(q\right)\right)}.\label{restrictedconditionbeta}
\end{equation}
If the sign here is equality (which generically occurs on a set of
measure $0$ so long as the functions are nowhere constant relative
to one another) then the firm is indifferent and if the inequality
is reversed the firm strictly chooses in-sourcing. As $\lambda$ rises,
the firm will in-source less and produce more; thus varying $\lambda$
over all positive numbers traces out all potentially optimal solutions.
Note that this could easily be extended to a situation where the firm
has any simple restricted choice of $\beta$, not just two values.

Furthermore, once $\beta(q)$ is set, we can easily recover the optimal
$\beta^{\star\star}$ for each $j$ by noting that the optimal value
of $\beta^{\star\star}$ at $\tilde{j}$ is the optimal value at $\tilde{q}$
satisfying the production equation
\[
\int_{0}^{\tilde{j}}S\left(\left[1-\beta^{\star\star}\left(q(j)\right)\right]MR\left(q^{\star\star}(j)\right)\right)dj=\tilde{q}.
\]
This implies the differential equation $q'(j)=S\left(\left[1-\beta^{\star\star}\left(q(j)\right)\right]MR\left(q^{\star\star}(j)\right)\right)$
and thus the inverse differential equation $j'(q)=\frac{1}{S\left(\left[1-\beta^{\star\star}\left(q\right)\right]MR\left(q^{\star\star}\right)\right)}$
which together with the boundary condition $j(0)=0$ yields $j(q)$
and thus $\beta^{\star\star}$ at each $j$.

It remains only to pin down the optimal value of $\lambda$. To do
this, denote the set of $q$ on which Inequality \ref{restrictedconditionbeta}
is satisfied $B_{I}(\lambda)$ and on which it is reversed $B_{O}(\lambda)$.%
\footnote{We ignore the generically $0$-measure set on which it is an equality.%
} Total production is
\[
q_{\lambda}=\int_{j\in(0,1):q(j)\in B_{I}(\lambda)}S\left(\left(1-\beta_{I}\right)MR\left(q(j)\right)\right)dj+\int_{j\in(0,1):q(j)\in B_{O}(\lambda)}S\left(\left(1-\beta_{O}\right)MR\left(q(j)\right)\right)dj,
\]
while total cost $C_{\lambda}=$
\[
\int_{B_{I}(\lambda)\cap\left(0,q_{\lambda}\right)}\left[1-\beta_{I}\right]MR\left(q\right)dq+\int_{B_{O}(\lambda)\cap\left(0,q_{\lambda}\right)}\left[1-\beta_{O}\right]MR\left(q\right).
\]
Profit is
\[
R\left(q_{\lambda}\right)-C_{\lambda}
\]
and the first-order condition for its maximization is
\[
MR\left(q_{\lambda}\right)\frac{\partial q_{\lambda}}{\partial\lambda}-\frac{\partial C_{\lambda}}{\partial\lambda}=0\implies MR\left(q_{\lambda}\right)=\frac{\frac{\partial C_{\lambda}}{\partial\lambda}}{\frac{\partial q_{\lambda}}{\partial\lambda}}=\lambda,
\]
because $\lambda$ is defined as the shadow cost of relaxing the constraint
on production.

Now we consider obtaining as close as possible to an explicit solution.
Note that, to do so, we must be able to characterize $S,B_{O}$ and
$B_{I}$ explicitly. $S$ is the inverse of $MC$ and thus $MC$ must
admit an explicit inverse. To characterize $B_{O}$ and $B_{I}$ explicitly
requires solving Inequality \ref{restrictedconditionbeta} with equality
to determine the relevant thresholds, which, as we will see, requires
marginal revenue to have an explicit inverse.

One of the simplest forms satisfying these conditions and yet yielding
our desired non-monotonicity is $P(q)=p_{0}+p_{-t}q^{t}+p_{-2t}q^{2t}$
and $MC(q)=mc_{-t}q^{t}$, where $t,p_{0},p_{-t},mc_{-t}>0>p_{-2t}$.
In this case $S(p)=\left(\frac{p}{mc_{-t}}\right)^{\frac{1}{t}}$.
Thus the equality version of Inequality \ref{restrictedconditionbeta}
becomes
\[
MR(q)=\frac{\lambda\left(\left[\frac{\left(1-\beta_{O}\right)MR(q)}{mc_{-t}}\right]^{\frac{1}{t}}-\left[\frac{\left(1-\beta_{I}\right)MR(q)}{mc_{-t}}\right]^{\frac{1}{t}}\right)}{\left(1-\beta_{O}\right)\left[\frac{\left(1-\beta_{O}\right)MR(q)}{mc_{-t}}\right]^{\frac{1}{t}}-\left(1-\beta_{I}\right)\left[\frac{\left(1-\beta_{I}\right)MR(q)}{mc_{-t}}\right]^{\frac{1}{t}}}\implies
\]
\[
\implies MR(q)=\frac{\lambda\left[\left(1-\beta_{O}\right)^{\frac{1}{t}}-\left(1-\beta_{I}\right)^{\frac{1}{t}}\right]}{\left(1-\beta_{O}\right)^{\frac{1+t}{t}}-\left(1-\beta_{I}\right)^{\frac{1+t}{t}}}\equiv\lambda k,
\]
where $k$ is the relevant collection of constants. Note that this
is an extremely simple threshold rule in terms of marginal revenue.
Given that we have chosen a form of marginal revenue that admits an
inverse, it is simple to solve out for the threshold rule in terms
of quantities; this is why we needed marginal revenue to have an inverse
solution.
\[
p_{0}+(1+t)p_{-t}q^{t}+(1+2t)p_{-2t}q^{2t}=\lambda k\implies
\]
\[
q=\left(\frac{-p_{-t}(1+t)\pm\sqrt{p_{-t}(1+t)^{2}+4\left(p_{0}-k\lambda\right)p_{-2t}(1+2t)}}{2p_{-2t}(1+2t)}\right)^{\frac{1}{t}}.
\]
Between these two roots, in-sourcing is optimal; outside them, outsourcing
is optimal.%
\footnote{Actually if $\lambda k<p_{0}$ then the lower root should be interpreted
as $0$. %
}

This provides closed-form solutions as a function of $\lambda$, but
$\lambda$ remains to be determined. This is, unfortunately, where
things start to get a bit messier. The integral determining $q_{\lambda}$
can be explicitly taken, but only in terms of the less-standard Appell
Hypergeometric function. The equation for $MR\left(q_{\lambda}\right)=\lambda$
therefore cannot be solved explicitly for $\lambda$. However, it
is a single explicit equation. Once $\lambda$ has been determined,
optimal sourcing is determined in closed-form as described above.
We plot this and the relaxed optimal $\beta$, in Figure \ref{closedformAC},
in the same format as in the paper for the case when $p_{0}=.2,p_{-t}=2,p_{-2t}=-4,mc_{-t}=.5,t=.5,\beta_{I}=.8,\beta_{O}=.3$. Clearly, we obtain similar, non-monotone results, but now these require
only a single call of Newton's method to solve an otherwise explicit
equation, as opposed to the two-dimension search we required to solve
the case presented in the paper.

\begin{figure}
\begin{centering}
\includegraphics[width=4in]{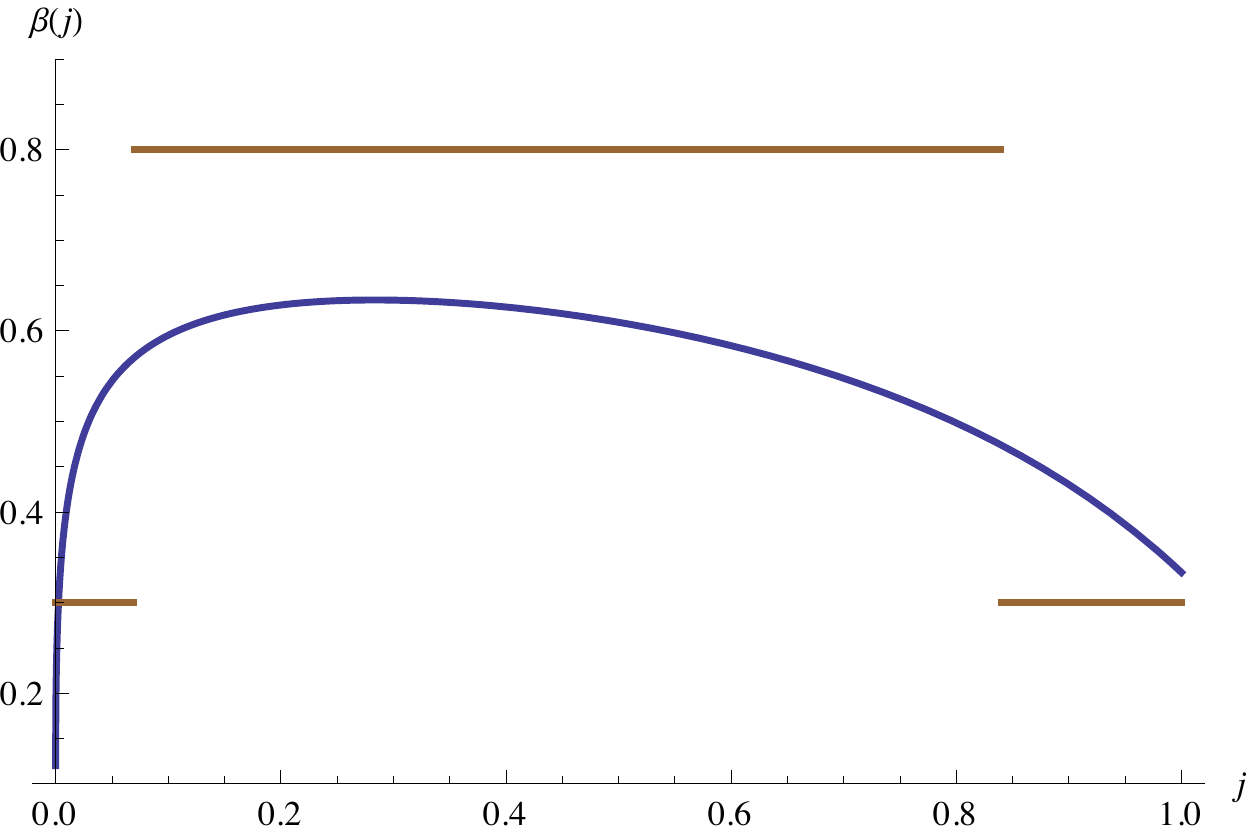}
\par\end{centering}

\caption{Relaxed and restricted solutions to the AC model when $P(q)=0.2+2q^{\frac{1}{2}}-4q$,
$MC(q)=\frac{q^{\frac{1}{2}}}{2}$, $\beta_{O}=0.3$ and $\beta_{I}=0.8$.}

\label{closedformAC}
\end{figure}

We do not discuss second-order conditions here, but they can easily
be derived and checked to hold for this example as well as for the
example in the paper. A grossly sufficient condition is that marginal
revenue is declining over the solution range, as is the case in both
of these examples.

\subsection{Labor bargaining without commitment \citep{stole,stole2}}\label{AppendixLaborBargainingWithoutCommitment}

\citet[henceforth SZ]{stole,stole2} consider a model of labor market
bargaining where contracts cannot commit workers. Each worker
is, therefore, able to extract a share of the surplus the firm gains
from a marginal worker. However, that surplus is determined by the
profits the firm would earn if that worker were to leave, in which
case the firm would bargain with other workers for a share of the
remaining surplus. This causes (a) wages to depend on infra-marginal
profits and (b) firms to over-employ workers relative to a standard
labor market since having reserve workers decreases the marginal value
of any given employee, lowering equilibrium wages and raising profits.

The setup of the SZ model is as follows. At the beginning of a period, a firm hires workers, each of
whom supplies one unit of labor if employed.%
\footnote{The model is formally dynamic but is usually studied in its steady
state as described here.%
} When this process has been completed but before production takes
place, the workers are free to bargain over their wages for this period.
At that time the firm cannot hire any additional workers, so if any
bargaining is not successful and any worker leaves the firm, fewer
workers will be available for production in this period. Moreover,
after the worker's departure, the remaining employees are free to
renegotiate their wages, and in principle the process may continue
until the firm loses all its employees. Assuming its revenues are
concave in labor employed, this gives the firm an incentive to ``over-employ''
or {\em hoard} workers as hiring more workers makes holding a marginal
worker less valuable to the firm and thus reduces workers' bargaining
power.

If the bargaining weight of the worker relative to that of the firm's
owner is $\lambda$, then the relationship surplus splitting condition
is $S_{w}$=$\lambda S_{f}$. The worker's surplus is simply the equilibrium
wage corresponding to the current employment level minus the outside
option: $S_{w}=W\left(l\right)-W_{0}$, where $W$ is the wage as
a function of $l$, the labor supplied. For expositional simplicity,
we assume the firm transforms labor into output one-for-one, though
analytic solutions also exist for any power law production technology
when $\lambda=1$ and in other cases. Thus we assume $q=l$ and henceforth
use $q$ as our primary variable analysis for consistency with previous
sections.

The firm faces inverse demand $P(q)$ and thus its profits are $\Pi(q)=\left[P(q)-W(q)\right]q$.
The firm's surplus from hiring an additional worker is then $\Pi'(q)$.
This gives the differential equation
\begin{equation}
W\left(q\right)-W_{0}=\lambda MR\left(q\right)+\lambda\left(W\left(q\right)q\right)'\Rightarrow\lambda(W\left(q\right)q^{1+\frac{1}{\lambda}})'=q^{\frac{1}{\lambda}}\left(\lambda MR\left(q\right)+W_{0}\right),\label{eq:SZwage}
\end{equation}
where $MR(q)\equiv P(q)+P'(q)q$ and the implication can be verified by simple
algebra and is a standard transformation for an ordinary differential
equation of this class. Integrating both of the sides of the equation,
imposing the boundary condition that the wage bill shrinks to $0$
at $q=0$, and solving out yields wages and profits
\begin{center}
$W\left(q\right)=q^{-\left(1+\frac{1}{\lambda}\right)}\intop_{0}^{q}x^{\frac{1}{\lambda}}MR(x)dx+\frac{W_{0}}{1+\lambda},\quad \Pi\left(q\right)=P(q)q-q^{-\frac{1}{\lambda}}\intop_{0}^{q}x^{\frac{1}{\lambda}}MR(x)dx-\frac{W_{0}}{1+\lambda}$.
\end{center}

While the wage equation is intractable in general, the operation on the
right-hand side does not change the functional form of any element of the
 average-marginal form-preserving class. To gain intuition for this, note
that as $\lambda\rightarrow0$, the model converges to the neoclassical
model because the worker has no bargaining power; thus the equation becomes $MR(q)=W_{0}$.
On the other hand as $\lambda\rightarrow\infty$, the equation converges
to $P(q)=W_{0}$ as workers capture all revenue and divide it equally.
Thus for intermediate $\lambda$ the marginal-average transformation
is effectively applied ``partially'' to $P(q)$. To see this mathematically,
suppose $MR(q)=aq^{-b}$. Then the integral term in Equation \ref{eq:SZwage}
becomes

\[
\frac{(1+\lambda)a\int_{0}^{q^{\star}}x^{\frac{1}{\lambda}}x^{-b}dx}{\lambda (q^{\star})^{1+\frac{1}{\lambda}}}
=\frac{(1+\lambda)a}{(q^{\star})^{1+\frac{1}{\lambda}}}\ \frac{(q^{\star})^{\frac{1+\lambda-b\lambda}{\lambda}}}{1+\lambda-b\lambda}
=\frac{1+\lambda}{1+\lambda-b\lambda}\ a(q^{\star})^{-b}.
\]
More generally, for $MR(q)$ a linear combination of power terms,
each term of becomes multiplied by ${(1+\lambda)}/{(1+\lambda+t\lambda)}$,
where $t$ is the power on the term. This tractability under form-preserving
classes, but general intractability, has led researchers to study
the SZ model almost exclusively under linear and constant elasticity
demand.

While this class can yield important insights, it also has significant
limitations. In particular, in the rest of this subsection we show that under
this class the percentage over-employment relative to the neoclassical
benchmark is constant as a function of the prevailing wage and multiplicative
demand shifters. Thus proportional over-employment does not vary,
for example, over the business cycle as consumers become richer and
employment grows overall. By contrast in a calibrated model with equal
bargaining weights ($\lambda=1$), using demand derived from the US
income distribution as in Section \ref{SectionReplacingConstantElasticityDemand},
Equation \ref{EquationQuadraticallyTractableFormForIncomeDistribution},
we find that over a reasonable business cycle range over-employment
should shift by roughly 0.4\% of total employment. While quite small in absolute terms, this could account for a non-trivial fraction of cyclic variation
in employment and is ruled out by the standard model. Furthermore,
this model is quadratically tractable, nearly as tractable as the standard constant
elasticity or linear specifications that are linearly tractable. It thus seems a natural alternative
to make future analysis of labor bargaining more realistic without
losing significant tractability.

To carry out this calculation, we note that the firm's optimal $q$ solves its first-order condition, $\Pi'\left(q\right)=0$,
which, after some algebraic manipulations, is
\begin{equation}
\frac{(1+\lambda)\intop_{0}^{q}x^{\frac{1}{\lambda}}MR(x)dx}{\lambda q^{1+\frac{1}{\lambda}}}=W_{0}.\label{eq:StoleZwiebelFOC}
\end{equation}

Let us define (relative) labor hoarding as $h\equiv {(q^{\star}-q^{\star\star})}/{q^{\star\star}}$,
where $q^{\star}$ is SZ employment and $q^{\star\star}$ is the employment
level that a neoclassical firm with identical technology would choose:
$MR\left(q^{\star\star}\right)=W_{0}$. Combining these definitions
with Equation \ref{eq:StoleZwiebelFOC} gives a useful condition for $h$
in terms of the equilibrium employment level $q^{\star}$:
\begin{equation}
MR\left(\frac{q^{\star}}{1+h}\right)=\frac{\left(1+\lambda\right)\intop_{0}^{q^{\star}}x^{\frac{1}{\lambda}}MR(x)dx}{\lambda\left(q^{\star}\right)^{1+\frac{1}{\lambda}}}.\label{eq:StoleZwiebelOveremploymentCondition}
\end{equation}
Note that this equation, and Equation \ref{eq:StoleZwiebelFOC}, involves
only (a) marginal revenue and (b) integrals of it multiplied by a power
of $q$ and then divided by a power of $q$ higher by 1. It can easily
be shown that the support of the Laplace marginal revenue is preserved
by this transformation using essentially the same argument we used
in the paper to show this support was shifted by exactly
one unit when consumer surplus is calculated. This implies that
Equations \ref{eq:StoleZwiebelFOC} and \ref{eq:StoleZwiebelOveremploymentCondition}
have precisely the same tractability characterization as does the
basic monopoly model we studied in Section \ref{SectionReplacingConstantElasticityDemand} of the paper.%
\footnote{Note that Equation \ref{eq:StoleZwiebelFOC} also involves a constant
and thus only our tractable forms with a constant term will maintain
their tractability in this model. This is why we focus on this class
below.%
}

 Given the complexity of
Equation \ref{eq:StoleZwiebelOveremploymentCondition} from any perspective
other than our tractable forms, we investigate it using these forms,
following \citet{inequality} who study the model under constant elasticity
demand. First consider the BP class, $P(q)=p_{0}+p_{t}q^{-t}$,
which nests the constant elasticity case when $p_{0}=0$. Solving
Equation \ref{eq:StoleZwiebelOveremploymentCondition} for $h$ yields
\begin{equation}
h=\left(\frac{1+\lambda}{1+\lambda-t\lambda}\right)^{\frac{1}{t}}-1.\label{SZBP}
\end{equation}
Therefore hoarding is constant in $q^{\star}$ and consequently in $W_{0}$.
Thus under the BP
class of demand, including constant elasticity, the economic cycle (the nominal outside option) has
{\em no effect} on relative hoarding. It can easily be shown that
$h$ monotonically increases in $t$, so that the less concave demand
(and thus profits) are, the more hoarding occurs.  We found this  counterintuitive, as we believed, building off
the intuition supplied by SZ about the relationship between the ``front-loading''
that drives hoarding and concavity, that labor hoarding was driven
by concavity in the firm's profit function.\footnote{However, it is worth noting that another source of profit convexity, fixed costs, has an opposite effect and is a natural element to include in the model.  This can be done in a straightforward way using our technology given our previous discussion, but we omit it here for brevity. } Instead it appears that
the reverse is the case. This shows one advantage of considering
explicit functional forms: they help correct false intuitions.
In particular, because $t$ clearly parameterizes concavity the comparative
static has a natural interpretation.

This new intuition suggests that the hoarding may not be constant
over the economic cycle if, during that cycle, the curvature of firm
profits change. For example, if during booms broad parts of the population
are served and during recessions only wealthier individuals are served,
then labor hoarding should be counter-cyclical as the distribution
of income among the wealthy is more convex than among the middle-class
and poor.

To analyze this we used our proposed functional form from Equation \ref{EquationQuadraticallyTractableFormForIncomeDistribution} in Section \ref{SectionReplacingConstantElasticityDemand}, in the version where it actually represents the income distribution as this is appropriately normalized to our assumption of $q=l$ (willingness-to-pay for a unit of labor):
\[
P(q)= 50000\left(\frac{1}{2} q^{-\frac{2}{5}}+2- \frac{5}{2}q^{\frac{2}{5}}\right).
\]
(in US dollars). Plugging this into Equation \ref{eq:StoleZwiebelOveremploymentCondition}
and $MR\left(q^{\star\star}\right)=W_{0}$, assuming to match the convention in the literature that $\lambda=1$ (though this plays no role in the simple form of our solution) and rounding to the second significant digit yields:
\begin{equation}
h=1.6\left(\frac{1+\sqrt{1+\frac{1.2\cdot10^9}{\left(100000-W_0\right)^2}}}{1+\sqrt{1+\frac{1.1\cdot10^8}{\left(100000-W_0\right)^2}}}.\right)^{\nicefrac{5}{2}}-1,\label{hoardingresult}
\end{equation}

\begin{figure}[t]
\begin{centering}
\includegraphics[width=3.2in]{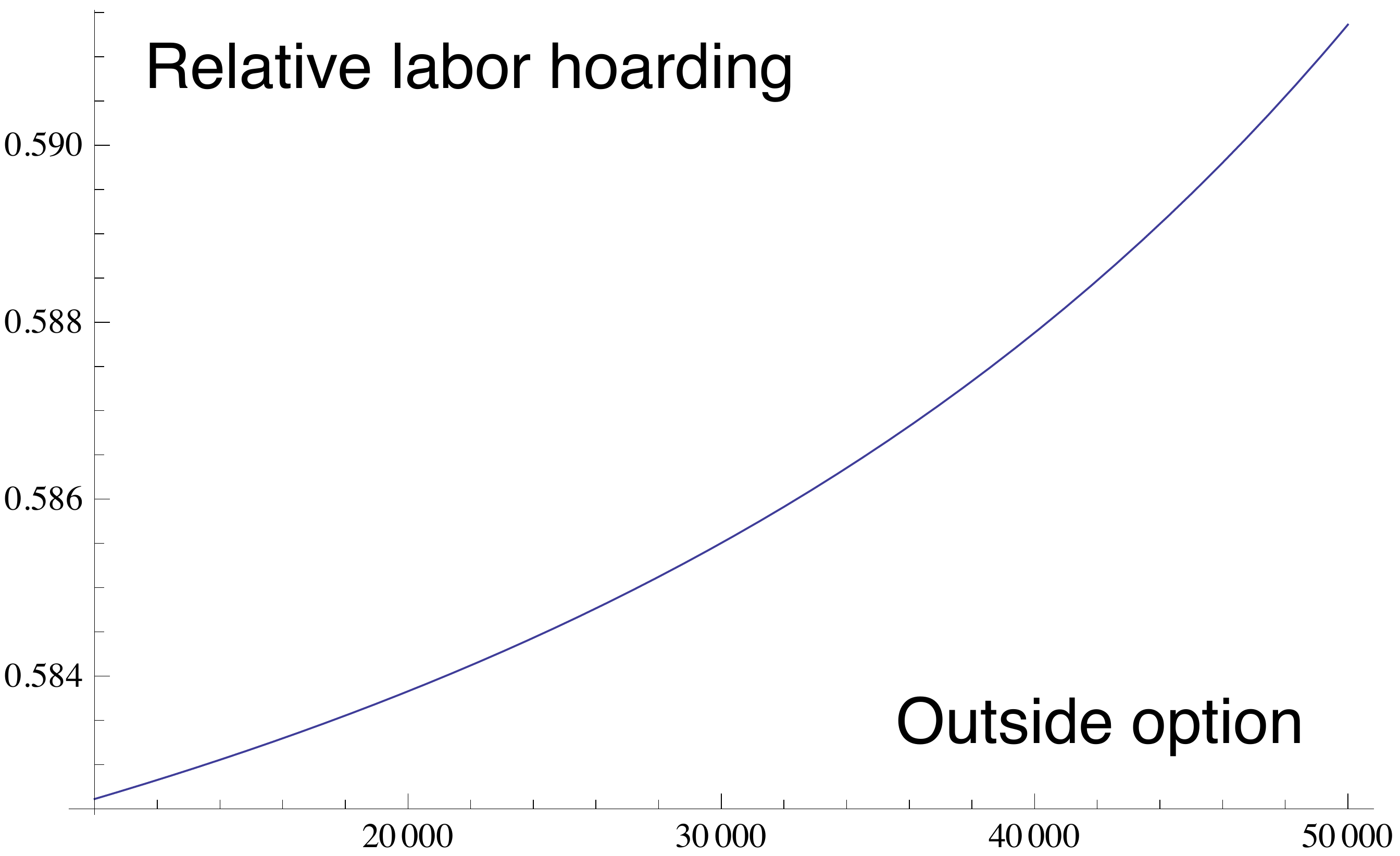}
\par\end{centering}
\caption{Relative labor hoarding in the Stole-Zwiebel model with $\lambda=1$ and
demand given by the approximation for $W_{0}\in$ [10000,50000] (in USD).}

\label{SZcalibration}
\end{figure}

 We interpret a reduction in $W_{0}$, or equivalently a multiplicative
scaling up of $P$, to be a boom (as it leads to higher production)
and a rise in $W_{0}$ to be a recession. The expression on the right-hand side of Equation \ref{hoardingresult} can easily be shown to be increasing in $W_0$ (in fact, this is true quite broadly beyond this particular calibration).  Thus a recession
hoarding rises, contrasting with the standard intuition that unions exacerbate recessions by creating
nominal wage rigidity and suggesting the effects of individual workers'
bargaining may have qualitatively different comparative statics than collective bargaining does.  Figure \ref{SZcalibration} shows the results quantitatively.   Hoarding is large ($\approx$ 59\%), but its comparative statics are less pronounced. It rises by a bit less than one percentage point when the outside option
rises from \$30k to \$50k, a reasonable range of variation
over the economic cycle. Thus, while the BP approximation of constancy appears
not to be very far off these effects may of a similar magnitude to cyclic shifts in employment and are thus worth considering.

\subsection{Imperfectly competitive supply chains}
\label{sequential}

Consider the model of imperfectly competitive supply chains where
each stage of production strategically anticipates the reactions of
the subsequent stage proposed by \citet{salinger}. There are $m$
stages of production interacting via linear pricing. Producers at
each stage act simultaneously and the stages act in sequence. We solve
by backwards induction.

Producers at stage $m$ take an input from producers at stage $m-1$
and sell it to final consumers, facing inverse demand $P_{m}$. The
$n_{m}$ firms at stage $m$ are symmetric Cournot competitors with
average cost $AC_{m}$. The linear price clearing the market between
stage $m-1$ and $m$ is $\hat{P}_{m-1}$. Using the standard first-order
condition for Cournot competition and dropping arguments, the first-order
equilibrium conditions are 
\[
P_{m}+\frac{1}{n_{m}}P_{m}'q=\hat{P}_{m-1}+AC_{m}+\frac{1}{n_{m}}AC_{m}'q\iff
\]
\[
\hat{P}_{m-1}=P_{m}+\frac{1}{n_{m}}P_{m}'q-AC_{m}-\frac{1}{n_{m}}AC_{m}'q.
\]
Thus the effective inverse demand facing the firms at stage $m-1$
is 
\[
P_{m-1}\equiv P_{m}+\frac{1}{n_{m}}P_{m}'q-AC_{m}-\frac{1}{n_{m}}AC_{m}'q,
\]
as all output produced at stage $m-1$ is used as an input at stage
$m$. Effectively the inverse demand at stage $m-1$ is the (competition-adjusted)
marginal profit (competition-adjusted marginal revenue less marginal
cost) at stage $m$.

This analysis may be back-propagated up the supply chain to obtain
a first-order condition at the first stage determining the quantity
in the industry. However, at each stage one higher derivative of $P_{m}$,
at least and also of some of the cost curves, enters the first-order
conditions. Thus the implicit equation for the first-order conditions
characterizing the supply chain is usually quite elaborate and is
both difficult to analyze in general and highly intractable, even
computationally, for many functional forms. For example, \citet{cable}
use this computational tractability concern to justify their focus
on simultaneous decisions upstream and downstream in a related vertical
contracting model.

However, we now derive a simple explicit transformation of the Laplace
inverse demand and average cost characterizing the supply chain and
discuss how this can be used to overcome these difficulties. Note
that 
\[
P_{m}+\frac{1}{n_{m}}P_{m}'q=\left(1-\frac{1}{n_{m}}\right)P_{m}+\frac{1}{n_{m}}MR_{m},
\]
where $MR_{m}=P_{m}+P_{m}'q$. Let $p_{m}$ be the Laplace inverse
demand. From Section \ref{SectionArbitraryDemandAndCostFunctions} we have that the Laplace marginal
revenue is $(1-t)p_{m}$ and thus that the inverse Laplace-log transform of $\left(1-\frac{1}{n_{m}}\right)P_{m}+\frac{1}{n_{m}}MR_{m}$
is just $\left(1-\frac{t}{n_{m}}\right)p_{m}$. By the same logic,
if we denote the Laplace average cost by $ac_{m}$ the inverse Laplace-log transform of $AC_{m}+\frac{1}{n_{m}}AC_{m}'q$
is $\left(1-\frac{t}{n_{m}}\right)ac_{m}$.

Iterating this process, one obtains that the Laplace first-order condition
at the initial stage, which we denote $f_{1}$, is 
\[
p_{m}\prod_{i=1}^{m}\left(1-\frac{t}{n_{i}}\right)-\sum_{i=1}^{m}\left[ac_{i}\prod_{j=1}^{i}\left(1-\frac{t}{n_{j}}\right)\right].
\]
This obviously differs only in its (trivially computed) coefficients
and not in its support from the $ac_{i}$'s and $p_{m}$ that make
it up. Thus if all $ac_{i}$'s and $p_{m}$ are chosen to have the
same tractable support (with the desired number of evenly
spaced mass points
to achieve desired tractability) then the full will be equally tractable.
Beyond this, even if $p_{m}$ and the $ac_{i}$'s are specified in
an arbitrary manner, the resulting Laplace first-order condition can
be trivially computed from the inverse Laplace-log transforms of each of these inputs and then
either solved directly by applying the Laplace transform or approximated
using a small number of evenly spaced mass points for tractability. In either
case, this approach significantly reduces the complexity of computing
and representing the system.

\subsection{Two-sided platforms \`a la \citet{rt2003}\label{AppendixTwoSidedPlatforms}}

\citet{rt2003} propose a model of a two-sided platform motivated
by the credit card industry. Sellers and buyers are randomly matched
and independently decide whether they want to accept credit cards
and whether they want to use them conditional on cards being accepted.
These decisions are driven by the price charged (or subsidy paid)
to each side. In particular, in order for a fraction of sellers $q_{S}$
to wish to accept cards, the price that must be charged to sellers
is $P_{S}\left(q_{S}\right)$, and similarly for buyers.

Let $U_{I}\left(q_{I}\right)\equiv\int_{0}^{q_{I}}P_{I}(x)dx$ be
the gross utility on side $I$. Because $U_{I}'\left(q_I\right)=P_{I}\left(q_I\right)$, the average
gross utility $\overline{U}_{I}\left(q_{I}\right)\equiv\nicefrac{U_{I}\left(q_{I}\right)}{q_{I}}$ has
the average-marginal relationship to inverse demand. Thus average
consumer surplus $\overline{V}_{I}\left(q_{I}\right)=\overline{U}_{I}\left(q_{I}\right)-P_{I}\left(q_{I}\right)$
has the same functional form as $P_I'q_I$ for a form-preserving functional
form class.

\citeauthor{rt2003} show that, when there is a constant and symmetric
marginal cost of clearing transactions $c$, imperfectly competitive
equilibrium between symmetric firms is characterized by 
\[
P_{S}\left(q_{S}\right)+P_{B}\left(q_{B}\right)-c=-\theta P_{S}'\left(q_{S}\right)q_{S}=-\theta P_{B}'\left(q_{B}\right)q_{B}
\]
for some constant $\theta<1$.\footnote{\citet{doublemarg2sms} extends this characterization to the case
of complements when $\theta>1$. For analogous reasons to the previous
applications all results here may be extended to arbitrary imperfectly
competitive supply chains.} On the other hand, they show that Ramsey pricing (which nests the unconstrained social planner's problem as a special case)
is characterized by 
\[
P_{S}\left(q_{S}\right)+P_{B}\left(q_{B}\right)-c=-\theta\overline{V}_{S}\left(q_{S}\right)=-\theta\overline{V}_{B}\left(q_{B}\right)
\]
for some constant $\theta$, equal to unity in the case of the unconstrained
social optimum and approaching $0$ as Ramsey pricing is required
to break even. Thus if inverse demand on both sides of the market
is specified within the same form-preserving class (\citeauthor{rt2003}
assume linear demand in their example) then our characterization of
tractability applies here as well.

Again the added flexibility of our forms is important in this context.
For example, \citet{ramslind} considers how platforms would choose an ``interchange fee'' between two sides of the market, holding fixed the overall level of prices.  He demonstrates that if both sides have BP demand, then users on both sides of the market and profit maximization {\em all} agree on the same optimal interchange fee.  However, this is generally false
and thus assuming BP demand trivializes the wide-ranging regulatory
debate over interchange fees.  In fact under plausible  (bell-shaped) demand forms, perhaps surprisingly, both sides in aggregate prefer to face higher prices (consumers prefer lower interchange, merchants prefer higher) to subsidize use on the other side of the market.  From a social perspective, the more heterogeneous side and/or the side which has more complete adoption should be taxed to subsidize the other side more than will be in the interest of a profit-maximizing platform, even for fixed aggregate prices.

\subsection{Auction theory}\label{AppendixAuctionTheory}

\subsubsection{Symmetric independent private values first-price auctions}\label{AppendixSymmetricIndependentPrivateValuesFirstPriceAuctions}

Consider $N$ symmetric bidders with privately-known values $v_{i}$ for a
single object drawn independently and identically from a distribution with differentiable CDF
$F$. Let $V(q)\equiv F^{-1}(q)$ be the quantile function
of $F$. Let $b_{\star}$ be a symmetric-equilibrium bid function
mapping values to bids in a first-price auction in which the highest
bidder wins and pays her bid value; any such equilibrium bid function can be shown
to be strictly monotone increasing under weak conditions. The probability
that the bid of any individual bidder is below $x$ is then $G_{\star}(x)\equiv F\left(b_{\star}^{-1}(x)\right)$.
Thus the probability that bidder $i$ wins if she submits a bid of
$x$ is, by symmetry, $\left[G_{\star}(x)\right]^{N-1}$.

The expected utility a bidder with value $v$ thus earns from a bid
of $x$ is 
\[
\left(v-x\right)\left[G_{\star}(x)\right]^{N-1}=\left(v-B_{\star}(q)\right)q^{N-1},
\]
where $q$ is the fraction of other bidders with (weakly) lower bids and $B_{\star}(q)\equiv G_{\star}^{-1}(q)$ is the 
quantile function of the equilibrium bid distribution. A necessary
condition for her optimization is therefore
\[
\left(v-B_{\star}(q)\right)(N-1)q^{N-2}+B_{\star}'(q)q^{N-1}=0\iff v=B_{\star}(q)+\frac{1}{N-1}qB_{\star}'(q).
\]
For this to be a symmetric, monotone equilibrium for the posited bid
distribution, it must be that a bidder with value at reversed quantile
$q$ of the value distribution chooses to bid (weakly) higher than precisely
a fraction $q$ of her rivals. Thus a necessary condition for a symmetric
equilibrium is 
\begin{center}
$V(q)=B_{\star}(q)+\frac{1}{N-1}qB_{\star}'(q)$.
\end{center}
Sufficient conditions, which we omit here, are well-known in the literature.
Note that the right-hand side of this expression involves the marginal
and average forms of $B_{\star}$. Thus, by simple coefficient matching,
if $V$ is chosen to be from a form-preserving class then there is
always an equilibrium $B_{\star}$ from the same class. This
may be used directly to analytically relate the values and bids at
various quantiles, which is all that is necessary for many analytic
problems. 

However if one wishes to obtain a closed form for $b_{\star}$ itself,
then one must choose the class to be tractable at the level of complexity
of the desired closed form and include a constant (a power of $0$)
in the class. By definition, $G_{\star}=F\circ b_{\star}^{-1}$, so $b_{\star}^{-1}=V\circ G_{\star}$
and consequently $b_{\star}=B_{\star}\circ F$. Thus if $F$  and $V$ have 
forms tractable at level $k$, then so does $b_{\star}$.  Evidently
uniform and exponential distributions, which have linear and logarithmic
$V$ respectively, are linearly tractable, explaining why they are
ubiquitously used for examples in symmetric first-price auction models. 

However these forms are quite restrictive in that they cannot, for
example, have the bell shape usually found in empirical studies of
valuation distributions in auctions \citep{hailetamer,hortacsu}.
Our forms can easily generate such shapes and thus allow tractable
examples with realistic value distributions.

\subsubsection{Auctions v. posted prices \citep*{onlineauctions}}\label{AppendixAuctionsVPostedPrices}

\citet{onlineauctions} consider the trade-off a seller faces between
using an auction and setting a posted price in an online retail market.  They assume sellers of goods know the common (positive) hassle cost $\lambda$ for buyers to participate in an auction, but may still use an auction because they do not know their common value $v$ of the good. The seller has an opportunity cost
of selling $c$, and $v$ is drawn from a distribution
$F$ that the seller knows. Assuming, as the authors do, that at least two bidders participate,
the auction guarantees that the seller gets value $v-\lambda$ as
long as $v-\lambda\geq r$, where $r$ is the reserve price the seller
sets. Alternatively, the seller may set a posted price $p$, in which
case she will sell the good if $v\geq p$.

Let $P(q)\equiv F^{-1}(1-q).$ If a seller sells the good with probability
$q$, then in an auction with the reserve price set to $P(q)-\lambda$
she will receive 
an average price $\overline{U}(q)-\lambda$, where $\overline{U}(q)\equiv \nicefrac{\int_{0}^{q}P(x)dx}{q}$
by the same logic as in Supplementary Material \ref{AppendixTwoSidedPlatforms}. If the seller sells the good with probability $q$ with a posted
price by setting price $P(q)$, she will receive price $P(q)$ with
certainty. Thus the region in which she wishes to use an auction rather
than a posted price is when she wishes to sell with probability $q$
such that $\overline{U}(q)>P(q)+\lambda$. As noted in Supplementary Material
\ref{AppendixTwoSidedPlatforms}, $\overline{U}$ has the average-marginal
relationship to $P$. For this reason, if $P$ is specified according to an average-marginal
form-preserving class including a constant term (power $0$ term), then
the resulting optimal cut-off rules for using a definite mechanism are
tractable at the level of tractability
of the class (in terms of both the cost and the desired probability of sale, which
is more directly observed in the authors' data). 

\citeauthor{onlineauctions} present such an example,
by assuming a uniform distribution and thus a linear form for $P$. 
In this case $\overline{U}-P$ uniformly grows in $q$. This implies
that sellers with a low cost (low opportunity cost of sale),
 such as impatient private individuals clearing old property out of the house, who wish to achieve sale with high probability
(quickly) will use auctions.  On the other hand, those who have a high cost, such as professional vendors,  who 
want to achieve a sale with low probability (slowly) to wait will
set a high posted price. However this is not generally true. If $P$
takes a constant elasticity form, for example, the reverse pattern
holds: low-cost sellers set a low posted price and sell quickly while
high cost (patient) sellers run an auction. 

For the bell-shaped demands that appear to fit \citeauthor{onlineauctions}'s
data best, the gap between $\overline{U}$ and $P$ is actually non-monotone,
first declining an then rising. This suggests auctions should be polarized
into goods that sell with very low and very high probability; that is among those clearing out their houses and among the most professional sellers. This
is in fact what the authors find; they cannot even measure the posted-price
demand curve at very low sale probabilities as they do not observe
sufficiently many items selling that infrequently with posted prices, while
the same is true at very high probabilities. This suggests richer
classes of tractable, form-preserving demand may be more useful in
modeling this trade-off than is the uniform distribution.

\subsection{Selection markets}\label{AppendixSelectionMarkets}

\citet{akerlof1970market} analyzed markets where the cost of providing
a service differs by the identity of the consumer to whom it is provided.
He studies a case that he labels ``adverse selection'' in which
consumers differ in only a single characteristic and in which raising this
one dimension increases both consumers' willingness-to-pay for the
product and the cost of serving them. \citet{einav2010estimating}
and \citet{einav2011selection} maintain \citeauthor{akerlof1970market}'s
assumption of a single product but allow consumers to differ along
multiple dimensions that may impact their willingness to pay and cost
in potentially rich ways.

\citet{einav2010estimating} define an inverse demand curve $P(q)$
for $q\in(0,1)$ as the willingness to pay of the individual in the
$(1-q)$th  quantile of the willingness-to-pay distribution. They define
average cost $AC(q)$ as the average cost of individuals who are in
the quantiles above $1-q$ of the willingness-to-pay distribution. They
argue that perfectly competitive equilibrium requires $AC(q)=P(q)$
while social optimization requires $MC(q)=P(q)$, where $MC$ has
the average-marginal relationship to $AC$. \citet{imperfectcomp}
extend this framework to nest a variety of models of imperfect competition
using a conduct parameter $\theta$ as in Subsection \ref{sub:Imperfectly-competitive-supply}
above and show that equilibrium is characterized by $\theta MC(q)+(1-\theta)AC(q)=(1-\theta)P(q)+\theta MR(q)$.

As is clear by now, both sides of this equation are tractable for
any value of $\theta$ at whatever level the cost and demand side
are specified if these are chosen to be part of a form-preserving
class. Many analyses have assumed linear forms on both the cost and
demand side \citep{cutler,einav2010estimating,einav2011selection},
partly for tractability. As \citet{scheuersmetters} highlight, this
assumption rules out many interesting phenomena, such as selection
that is ``advantageous'' (higher willingness to pay correlating
with lower cost) over some range but adverse over other ranges or
multiple local competitive equilibria that \citeauthor{scheuersmetters}
argue may have challenged the introduction of the Affordable Care
and Patient Protection Act in the United States. Broader tractable
form-preserving classes, especially those with bell-shaped demand
and cost curves, allow these possibilities and appear to fit existing
empirical evidence more closely.

\subsection{Monopolistic competition}\label{AppendixMonopolisticCompetition}

\subsubsection{Tractable generalizations of the Dixit-Stiglitz framework with separable utility}

In the simplest monopolistic competition model, consumers derive their
utility from a continuum of varieties $\omega\in\Omega$ of a single
heterogeneous good in a separable way:
$%
U_{\Omega}=\int_{\Omega}u_{\omega}\left(q_{\omega}\right)\, d\omega.\ \label{eq:SeparableUtility}
$%
 In the original Dixit-Stiglitz model specialized to the case of constant elasticity of substitution
$\sigma$, $u_{\omega}(q_{\omega})$ is a power of the consumed quantities
$q_{\omega}$: $u_{\omega}(q_{\omega})\propto q_{\omega}^{1-{1}/{\sigma}}$.
In our generalization we wish to be able to apply Theorem \ref{TheoremClosedFormSolutions}, 
so we let  $u\left(q_{\omega}\right)$ be
a linear  combination of different powers of $q_{\omega}$. More
explicitly, consumer optimization requires that marginal utility of
extra spending is equalized across varieties: $u_{\omega}'\left(q_{\omega}\right)=\lambda P_{\omega}$,
where $P_{\omega}$ is the price of variety $\omega$ and $\lambda$
is a Lagrange multiplier related to consumers' wealth. To ensure tractability,
we let the residual inverse demand $P_{\omega}\left(q_{\omega}\right)={u_{\omega}'\left(q_{\omega}\right)}/{\lambda}$
and the corresponding revenue $R_{\omega}\left(q_{\omega}\right)$
be linear combinations of equally-spaced powers of $q_{\omega}$:
$%
P_{\omega}\left(q_{\omega}\right)=\sum_{t\in T}p_{\omega,t}q_{\omega}^{-t},\quad R_{\omega}\left(q_{\omega}\right)=\sum_{t\in T}p_{\omega,t}q_{\omega}^{1-t}\
$%
for some finite and evenly-spaced set $T$, with the number of elements
of $T$ determining the precise degree of tractability. For convenience
of notation, we choose a num\' eraire in a way that keeps $P_{\omega}\left(q_{\omega}\right)$
for a given $q_{\omega}$ independent of macroeconomic circumstances.

Each variety of the differentiated good is produced by a single firm.
We assume that the marginal cost and average cost of production can
be written as
$%
MC_{\omega}(q)=\sum_{t\in T}mc_{\omega,t}q_{\omega}^{-t},\quad AC_{\omega}(q)=\sum_{t\in T\cup\left\{ 1\right\} }ac_{\omega,t}q_{\omega}^{-t},\
$%
where $mc_{\omega,t}=\left(1-t\right)ac_{\omega,t}$. A constant component
of average cost (and marginal cost) would correspond to $ac_{\omega,0}$
and a fixed cost would correspond to $ac_{\omega,1}$. However, given
the generality possible here we do not necessarily have to assume
that these components are present in all models under consideration.

With this specification, Theorem \ref{TheoremClosedFormSolutions} applies and the firm's problem has closed-form solutions unless $T$ has six elements or more. Moreover, if firms are heterogeneous in their productivity, then
\ref{TheoremAggregation} leads to closed-form aggregation integrals for suitable choices of the productivity distribution, as in the case of a generalized Melitz model discussed below.

\subsubsection{Tractable generalizations of the D-S framework with non-separable utility}

Here we briefly discuss tractable monopolistic competition in the case of non-separable utility.\footnote{
In the case of heterogeneous firms, this generalization contains as special cases both the \citeauthor{melitz} model
and the \citeauthor{melitzottaviano} model. 
To be more precise, let us note that in addition to the heterogeneous-good varieties explicitly considered
here, the \citeauthor{melitzottaviano} model includes a homogeneous
good. In our discussion, the homogeneous good is absent, but adding
it to the model is straightforward.%
} The utility has the very general form
\[
U_{\Omega}\equiv F\left(U_{\Omega}^{\left(1\right)},U_{\Omega}^{\left(2\right)},...,U_{\Omega}^{\left(m\right)}\right),\quad U_{\Omega}^{\left(i\right)}\equiv\int_{\Omega}U^{\left(i,\omega\right)}\left(q_{\omega}\right)\, d\omega.
\]
In order to preserve tractability, we assume that $U^{\left(i,\omega\right)}\left(q_{\omega}\right)$
are linear combinations\footnote{Of course, without loss of generality we could assume that $U^{\left(i,\omega\right)}\left(q_{\omega}\right)$
are power functions and let the function $F$ combine them into any
desired linear combinations. However, for clarity of notation it is
preferable to keep the number $m$ of different expressions $U_{\Omega}^{\left(i\right)}$
small.%
} of equally-spaced powers of $q_{\omega}$ and that the set of exponents
does not depend on $i$ or $\omega$. For example, we could specify
$U_{\Omega}\equiv U_{\Omega}^{\left(1\right)}+\kappa_{1}(U_{\Omega}^{\left(1\right)})^{\xi_{1}}+\kappa_{2}(U_{\Omega}^{\left(2\right)})^{\xi_{2}}$,
$U_{\Omega}^{\left(1\right)}\equiv\int_{\Omega}q_{\omega}^{\gamma_{1}}\, d\omega$,
and $U_{\Omega}^{\left(2\right)}\equiv\int_{\Omega}q_{\omega}^{\gamma_{2}}\, d\omega$,
with $(\gamma_{1}+1)/(\gamma_{2}+1)$ equal to the ratio of two small
integers. In the language of heterogeneous-firm models, the choice $\kappa_{1}=\kappa_{2}=0$ corresponds to the
\citeauthor{melitz} model, while the choice $\xi_{1}=2$, $\xi_{2}=1$,
$\gamma_{1}=1$, and $\gamma_{2}=2$ gives the \citeauthor{melitzottaviano}
model, which is based on a non-homothetic quadratic utility. Our general
specification allows also for homothetic non-separable utility functions
that feature market toughness effects analogous to those in the \citeauthor{melitzottaviano}
model.

It is straightforward to verify that as in the separable-utility case, Theorems \ref{TheoremClosedFormSolutions} and \ref{TheoremAggregation} still apply and 
lead to closed-form solutions to the firm's problem and closed-form aggregation. This is because the structure of the firm's problem is unchanged. Non-separability only makes the resulting system of equations for macroeconomic aggregates more complex.  The system itself may still be written in closed form due to Theorem \ref{TheoremAggregation}, under appropriate assumptions on the productivity distribution.

\subsubsection{Tractable generalizations of the Dixit-Stiglitz framework}

In the baseline monopolistic competition model consumers derive their
utility from a continuum of varieties $\omega\in\Omega$ of a single
heterogeneous good:
\begin{equation}
U_{\Omega}=\int_{\Omega}u_{\omega}\left(q_{\omega}\right)\, d\omega.\label{eq:SeparableUtility}
\end{equation}
In the original Dixit-Stiglitz model with constant elasticity of substitution
$\sigma$, $u_{\omega}(q_{\omega})$ is a power of the consumed quantities
$q_{\omega}$: $u_{\omega}(q_{\omega})\propto q_{\omega}^{1-\nicefrac{1}{\sigma}}$.
In our generalization $u\left(q_{\omega}\right)$ is assumed to be
a function of a combination of different powers of $q_{\omega}$ . More
explicitly, consumer optimization requires that marginal utility of
extra spending is equalized across varieties: $u_{\omega}'\left(q_{\omega}\right)=\lambda P_{\omega}$,
where $P_{\omega}$ is the price of variety $\omega$ and $\lambda$
is a Lagrange multiplier related to consumers' wealth. To ensure tractability,
we let the residual inverse demand $P_{\omega}\left(q_{\omega}\right)=\nicefrac{u_{\omega}'\left(q_{\omega}\right)}{\lambda}$
and the corresponding revenue $R_{\omega}\left(q_{\omega}\right)$
be linear combinations of equally-spaced powers of $q_{\omega}$:
\[
P_{\omega}\left(q_{\omega}\right)=\sum_{t\in T}p_{\omega,t}q_{\omega}^{-t},\quad R_{\omega}\left(q_{\omega}\right)=\sum_{t\in T}p_{\omega,t}q_{\omega}^{1-t}
\]
for some finite and evenly-spaced set $T$, with the number of elements
of $T$ determining the precise degree of tractability. For the convenience
of notation, we choose a num\'eraire in a way that keeps $P_{\omega}\left(q_{\omega}\right)$
for a given $q_{\omega}$ independent of macroeconomic circumstances.

Each variety of the differentiated good is produced by a single firm.
We assume that the marginal cost and average cost of production can
be written as
\[
MC_{\omega}(q)=\sum_{t\in T}mc_{\omega,t}q_{\omega}^{-t},\quad AC_{\omega}(q)=\sum_{t\in T\cup\left\{ 1\right\} }ac_{\omega,t}q_{\omega}^{-t},
\]
where $mc_{\omega,t}=\left(1-t\right)ac_{\omega,t}$. A constant component
of average cost (and marginal cost) would correspond to $ac_{\omega,0}$
and a fixed cost would correspond to $ac_{\omega,1}$. However given
the generality possible here we do not necessarily have to assume
that these components are present in all models under consideration.

\subsubsection{Flexible \citeauthor{krugman} model}

The \citet{krugman} model of trade, featuring monopolistic competition
and free entry of identical single-product firms, may be solved explicitly
for the tractable demand and cost functions mentioned above, not just
constant-elasticity demand and constant marginal cost specified in
the original paper. Here we consider these solutions in the case of
two symmetric countries, which leads to a symmetric equilibrium.

There is a continuum of identical consumers with preferences as in
Equation \ref{eq:SeparableUtility} who earn labor income. The amount
of labor a firm needs to hire in order to produce quantity $q$ may
be split into a fixed part $f$ and a variable part $L\left(q\right)$
that vanishes at zero quantity. Both $L\left(q\right)$ and the revenue
function $R\left(q\right)$ are assumed to allow for a linear term.
The firm only uses labor for production, so its total cost is $w\left(L\left(q\right)+f\right)$,
where $w$ is the competitive wage rate. Having produced quantity
$q$, the firm splits it into $q_{d}$ to be sold domestically, and
$\tau q_{x}$ to be shipped abroad. Due to\emph{ }iceberg-type\emph{
}trade costs ($\tau\ge1$), a fixed fraction of the shipped good is
lost during transport, and only quantity $q_{x}$ is received in the
other country. (Non-iceberg trade costs are considered in the appendix.)
Let us denote the equilibrium level of marginal cost, measure of firms,
international trade flows, and welfare by $MC$$^{\star}$, $N^{\star}$,
$X^{\star}$, and $W^{\star}$, respectively, and similarly for other
variables. The total labor endowment of one of the two symmetric economies
is $L_{E}$. \medskip{}

\begin{observation} There exists an explicit map $MC^{\star}\rightarrow\left(f,q_{d}^{\star},q_{x}^{\star},w^{\star}\right)$
and an explicit map $\left(MC^{\star},L_{E}\right)\rightarrow\left(N^{\star},X^{\star},W^{\star}\right)$.
These relationships represent a closed-form solution to the model
in terms of $MC^{\star}$ and exogenous parameters. \end{observation}

To see briefly why this is the case, it is convenient to express the model's equations in terms of
the equilibrium level of marginal cost $MC^{\star}$.%
\footnote{The case of a single country corresponds to the Dixit-Stiglitz model.
It may be obtained from our two-country discussion by setting $\tau\rightarrow\infty$
and $q_{x}=0$. In this case one does not have to express the model's
equations in terms of the equilibrium level of marginal cost $MC^{\star}$
as we do below. Instead, for tractable functions $R\left(q\right)$
and $L\left(q\right)$ one can solve for equilibrium quantity $q^{\star}$
in closed form (in terms of the fixed cost of production $f$) from
an equation that combines profit maximization and free entry: $(L(q)+f)R'(q)=R(q)L'(q)$.%
} Output optimally designated for the domestic market and the export
market will satisfy $R'(q_{d})=MC^{\star}$ and $R'(q_{x})=\tau\, MC^{\star}$,
respectively, and therefore may be solved for in closed form in terms
of $MC^{\star}$ for tractable specifications of the revenue function
(or consumer preferences).%
\footnote{As mentioned in the paper, a convenient choice of num\'eraire allows
us to keep the revenue function $R\left(.\right)$ independent of
economic circumstances.%
} The same is true for wages since $w=MC^{\star}/L'(q_{d}+\tau q_{x})$.

For a chosen $MC^{\star}$ we may compute the level of fixed cost
$f$ consistent with it using the free-entry condition: $R(q_{d})+R(q_{x})=wL(q_{d}+\tau q_{x})+wf$.
The equilibrium number (measure) $N^{\star}$ of firms in each economy
then satisfies $N^{\star}=L_{E}/(L\left(q_{d}+q_{x}\right)+f)$, where
$L_{E}$ is the labor labor endowment one of the two economies.%
\footnote{$L_{E}$ may be exogenous, as in the original Krugman model, but even
for endogenous labor supply, it is possible to obtain fully explicit
solutions to the model in terms of the parameter $MC^{\star}$.%
} Other variables of interest, e.g. trade flows or welfare, are then
simply functions of the ones discussed above.

\paragraph*{Krugman model with non-iceberg and iceberg international trade costs.}

Although the Krugman model with non-iceberg trade costs is not our
main focus here, we mention it for completeness. Let us assume the
presence of non-iceberg international trade costs that require hiring
labor $L_{T}\left(q_{x}\right)$ in order for $q_{x}$ to reach its
destination in the other country.%
\footnote{In a symmetric equilibrium it does not matter how this labor is split
between the countries, as long as the symmetry of the model is maintained.
For asymmetric countries, we could assume that the transport requires
labor from both countries. The model may be solved in terms of marginal
costs of serving each market.%
} The export FOC is now $R'(q_{x})-wL'_{T}(q_{x})=\tau MC^{\star}$,
while the free entry condition becomes $R(q_{d})+R(q_{x})=wL(q_{d}+\tau q_{x})+wL_{T}(q_{x})+wf$.
The resulting number (measure) of firms is $N^{\star}=L_{E}/((L\left(q_{d}+q_{x}\right)+f)+L_{T}(q_{x}))$.
The model may be solved explicitly along the same lines in terms of
chosen $MC^{\star}$ and $w$, with $f$ and $\tau$ treated as derived
quantities.

\subsubsection{Flexible \citeauthor{melitz} model }\label{AppendixFlexibleMelitzModel}

The \citet{melitz} model is again based on monopolistic competition
and assumes a constant elasticity of substitution between heterogeneous-good
varieties. Relative to the \citet{krugman} model, it introduced a
novel channel for welfare gains from trade, namely increased average
firm productivity resulting from trade liberalization or analogous
decreases in trade costs. Here we generalize the model to allow for
more flexible demand functions, non-constant marginal costs of production,
and trade costs that may have components that are neither iceberg-type
nor constant per unit.

\paragraph*{Single country. }

For clarity of exposition, we first describe the flexible and tractable
version of the \citeauthor{melitz} model in the case of a single
country and later discuss its generalization. Just like the \citeauthor{krugman}
model, it involves two types of agents: monopolistic single-product
firms and identical consumers, who supply their labor in a competitive
labor market and consume the firms' products.%
\footnote{For simplicity, consumers do not discount the future, although it would
be easy to incorporate an explicit discount factor. Formally, the
model includes an infinite number of periods, but it may be thought
of as a static model because the equilibrium is independent of time.%
}

Labor is the only factor of production: all costs have the interpretation
of labor costs and are proportional to a competitive wage rate $w$.
Each heterogeneous-good variety is produced by a unique single-product
firm, which uses its monopolistic market power to set marginal revenue
equal to marginal cost. Demand and costs are specified tractably as
discussed above; this time we do not need to assume that variable
cost and revenue functions allow for a linear term.

If the firm is not able to make positive profits, it is free to exit
the industry. In situations of main interest, this endogenous channel
of exit is active: there exist firms that are indifferent between
production and exit. There is also an exogenous channel of exit: in
every period with probability $\delta_{e}$ the firm is forced to
permanently shut down.

Entry into the industry is unrestricted but comes at a fixed one-time
cost $wf_{e}$. Only after paying this fixed cost, the entering firm
observes a characteristic $a$, drawn from a distribution with cumulative
distribution function $G\left(a\right),$ that influences the firm's
cost function. In the original \citeauthor{melitz} model the constant
marginal cost of production is equal to $wa$. Here we leave the specification
more general while maintaining the convention that increasing $a$
increases the firm's cost at any positive quantity $q$. In expectation,
the stream of the firm's profits must exactly compensate the (risk-neutral)
owner for the entry cost, which implies the \emph{unrestricted entry
condition} $wf_{e}=\mathbb{E}\Pi\left(q;a\right)/\delta_{e}$, with
the profit $\Pi\left(q;a\right)$ evaluated at the optimal quantity.%
\footnote{In the case of a single country, the profit is simply $\Pi\left(q;a\right)=q\left[P\left(q;a\right)-AC\left(q;a\right)\right]$.
Also, note that the unrestricted entry condition is often referred
to as the \emph{free entry condition}, but here we avoid this term
since there is a positive entry cost.%
}

The amount of labor needed to produce quantity $q$ is $L\left(q;a\right)+f$,
where $L\left(q;a\right)$ corresponds to variable cost ($L\left(0;a\right)=0$)
and $f$ to a fixed cost. $L\left(q;a\right)$ is assumed to be tractable
with respect to $q$, but also with respect to $a$.%
\footnote{For example, the function $L\left(q;a\right)$ could be linear in
$a$, as would be the case in the original \citeauthor{melitz} model.
A simple example of a tractable choice of functional forms is $L\left(q\right)=\tilde{L}\left(q\right)+a\hat{L}\left(q\right)$,
$\hat{L}\left(q\right)\equiv q^{t}$, $\tilde{L}\left(q\right)\equiv\tilde{\ell}_{t}q^{-t}+\tilde{\ell}_{u}q^{-u}$,
and $R\left(q\right)=r_{t}q^{-t}+r_{u}q^{-u}$.%
} In terms of the labor requirement function $L\left(q;a\right)$,
the firm profit maximization condition and the z\emph{ero cutoff profit
condition} are $R'\left(q\right)=wL'\left(q;a\right)$ and $R(q_{c})=wL(q_{c};a_{c})+wf$,
where $q_{c}$ and $a_{c}$ correspond to a \emph{cutoff firm}, i.e.
a firm that is in equilibrium indifferent between exiting and staying
in the industry. We denote by $L_{E}$, $M^{\star}$,$M_{e}^{\star}$,
$W^{\star}$ the labor endowment, and the equilibrium measure of firms,
measure of entering firms, and level of welfare, respectively.\medskip{}

The firm profit maximization condition and the free entry condition
are\vspace{-0.7cm}

\begin{gather}
R'\left(q\right)=wL'\left(q;a\right),\label{eq:MelitzOneCountryA}\\
R(q_{c})=wL(q_{c};a_{c})+wf.\label{eq:MelitzOneCountryB}
\end{gather}
A convenient solution strategy is to choose $q_{c}$ and then calculate
$f_{e}$ as a derived quantity. For a chosen $q_{c}$ we can find
$a_{c}$ explicitly by combining (\ref{eq:MelitzOneCountryA}) and
(\ref{eq:MelitzOneCountryB}) into $R'\left(q_{c}\right)\left(L(q_{c};a_{c})+f\right)$=
$R(q_{c})L'\left(q_{c};a_{c}\right)$, since $L\left(q;a\right)$
is assumed to be tractable also with respect to $a$. Wages are then
given recovered from (\ref{eq:MelitzOneCountryB}): $w=R(q_{c})/(L(q_{c};a_{c})+f)$.

Now we need to show how to calculate the fixed cost of entry $f_{e}$
and the measure of firms. The fixed cost of entry consistent with
the chosen cutoff quantity is given simply by the unrestricted entry
condition:
\[
w\delta_{e}f_{e}=\bar{\Pi}=\intop_{q\geq q_{c}}\left(R\left(q\right)-wL\left(q;a\right)-wf\right)dG\left(a\left(q\right)\right).
\]
Here $a\left(q\right)$ is the firm's productivity parameter as an
explicit function of the optimally chosen quantity $q$ that results
from using (\ref{eq:MelitzOneCountryA}). For Pareto $G$, and $L$
and $R$ tractable from the point of view of $q$ (but not necessarily
having a linear term) and $L\left(q;a\right)$ linear in $a$, there
exist closed-form expressions for this integral in terms of special
functions, which are straightforward to derive, especially if one
uses symbolic manipulation software such as Mathematica. If the shape
parameter of the Pareto distribution is a negative integer, the integrals
actually reduce to simple power functions.

If $M_{e}$ denotes the measure of firms that enter each period (in
one country), then the measure of operating firms is $M=G\left(a_{c}\right)M_{e}/\delta_{e}$.
The total labor used in the economy is given by $L_{E}=M_{e}f_{e}+Mf+M\bar{L}$,
where $\bar{L}=G\left(a_{c}\right)^{-1}\intop_{q\geq q_{c}}L\left(q;a\right)dG\left(a\left(q\right)\right)$
is the labor on average hired for the variable cost of production.
Under the same assumptions, the integral again has an explicit form
in terms of special functions. We see that in these cases we can get
fully explicit expressions for $f_{e}$ and $M$ in terms of chosen
$q_{c}$ and $L_{E}$.

Other quantities of interest, such as trade flows or welfare, may
be found in an analogous, straightforward fashion.

\paragraph{Two countries with non-iceberg and iceberg international trade costs.}

Just like in the case of the flexible Krugman model, it is convenient
to write the model in terms of equilibrium marginal cost, which this
time is firm-specific and also depends on the firm's chosen export
status. For tractability we will need the revenue function $R\left(q\right)$
and the production labor requirement function $L\left(q;a\right)$
to allow for a linear term. The same is true for labor corresponding
to the non-iceberg trade costs, here denoted by $L_{T}\left(q_{x}\right)$.
As in the original \citet{melitz} paper, we consider equilibria characterized
by two cutoffs, here denoted $a_{1}$ and $a_{2}$, such that least
productive firms with $a>a_{1}$ exit, more productive firms with
$a\in(a_{2},a_{1}]$ serve only their domestic market, and most productive
firms with $a\leq a_{2}$ serve both countries. In general, we denote
the equilibrium marginal cost of a non-exporting firm as $MC_{n}^{\star}$
and that of an exporting firm as $MC_{x}^{\star}$. Variables corresponding
to the two cutoffs are distinguished by subscripts $1$ and $2$,
so for example $MC_{1n}^{\star}$ is the optimal marginal cost of
a firm with $a=a_{1}$, and $MC_{2x}^{\star}$ and $MC_{2n}^{\star}$
are optimal marginal costs of a firm with $a=a_{2}$ that decides
to export or not to export, respectively. We denote by $M_{x}^{\star}$
and $X^{\star}$ the equilibrium measure of exporting firms and international
trade flows.

Our solution strategy is to treat $MC_{1n}$ and $MC_{2x}$ as given
and to express other variables of the model in terms to these two
chosen parameters. In particular, we will show how to derive explicit
expressions for the fixed cost of exporting $f_{x}$ and cost of entry
$f$$_{e}$. The (variable-cost) labor requirement $L\left(q;a\right)$
is assumed to be a tractable combination of equidistant powers of
$a$, with coefficients that in general depend on $q$. Firms' profit
maximization leads to the set of equations: \begin{subequations}
\begin{gather}
MC_{n}=R'(q_{n})\label{eq:MelitzTwoCountriesA}\\
MC_{n}=wL'\left(q_{n};a\right)\label{eq:MelitzTwoCountriesB}\\
MC_{x}=R'(q_{d})\label{eq:MelitzTwoCountriesC}\\
MC_{x}=\begin{array}{c}
\frac{1}{\tau}\end{array}\!\! R'(q_{f})-\begin{array}{c}
\frac{1}{\tau}\end{array}\!\! wL_{T}'(q_{f})\label{eq:MelitzTwoCountriesD}\\
MC_{x}=wL'(q_{d}\!+\!\tau q_{f};a)\label{eq:MelitzTwoCountriesE}\\
R(q_{1n})-wL(q_{1n};a_{1})=f\label{eq:MelitzTwoCountriesF}\\
R(q_{2d})+R(q_{2f})-wL(q_{2d}+\tau q_{2f};a_{2})-wL_{T}(q_{2f})=f+f_{x}\label{eq:MelitzTwoCountriesG}
\end{gather}
\end{subequations}Here $q_{n}$ is the quantity sold by a non-exporting
firm, while $q_{d}$ and $q_{f}$ represent quantities that reach
domestic and foreign customers of an exporting firm, respectively.
In addition to exporting cost $wL_{T}(q_{f})$, we allow for an iceberg
trade cost factor $\tau\ge1$.

For a chosen $MC_{1n}$, we can calculate $q_{1n}$ from (\ref{eq:MelitzTwoCountriesA}).
The corresponding $a_{1}$ may be found by solving a linear equation
that results from combining (\ref{eq:MelitzTwoCountriesB}) and (\ref{eq:MelitzTwoCountriesF})
in a way that eliminates wages. Wages then may be recovered by substituting
back to (\ref{eq:MelitzTwoCountriesB}).

For a chosen $MC_{2x}$, we can derive $q_{2d}$ from (\ref{eq:MelitzTwoCountriesC})
and $q_{2f}$ from (\ref{eq:MelitzTwoCountriesD}). The value of $a_{2}$
is then determined by (\ref{eq:MelitzTwoCountriesE}). We find $q_{2n}$
by solving (\ref{eq:MelitzTwoCountriesA}) and (\ref{eq:MelitzTwoCountriesB})
with $MC_{2n}$ eliminated, and then in turn use one of these to find
$MC_{2n}$. This means that we know the marginal cost at the cutoffs.
The fixed cost of exporting $f_{e}$ is then identified from (\ref{eq:MelitzTwoCountriesG}).

For a given marginal cost, we can find the corresponding quantities
and productivity parameters $a$ by a similar method from (\ref{eq:MelitzTwoCountriesA}-\ref{eq:MelitzTwoCountriesE}),
this time treating $w$ as known. We denote the resulting functions
$q_{n}\left(MC_{n}\right)$, $q_{d}\left(MC_{x}\right)$, $q_{x}\left(MC_{x}\right)$,
$a_{n}\left(MC_{n}\right)$, and $a_{x}\left(MC_{x}\right)$. Using
these functions we can now determine the entry labor requirement $f_{e}$
from the unrestricted entry condition:
\[
w\delta_{e}f_{e}=\bar{\Pi}=\intop_{y\in S_{n}}\Pi\left(q_{n}\left(y\right);a_{n}\left(y\right)\right)dG\left(a_{n}\left(y\right)\right)+\intop_{y\in S_{x}}\Pi\left(q_{x}\left(y\right);a_{x}\left(y\right)\right)dG\left(a_{x}\left(y\right)\right),
\]
where $\Pi$ is the profit function (revenue minus cost), $G\left(a\right)$
is the cumulative distribution function of $a$, and the integration
ranges are $S_{n}\equiv\left(MC_{2n},MC_{1n}\right)$ and $S_{x}\equiv\left(0,MC_{1n}\right)$.
Under various assumptions these integrals may be evaluated in closed
form, often involving special functions. If a measure $M_{e}$ of
firms enters each period (in one of the countries), then the equilibrium
measure of operating firms is $M=M_{e}G\left(a_{1}\right)/\delta_{e}$
and that of exporting firms is $M_{x}=M_{e}G\left(a_{2}\right)/\delta_{e}$.
These measures may be calculated from the labor market clearing condition
$M_{e}f_{e}+Mf+M_{x}f_{x}+(M-M_{x})\bar{L}_{n}+M_{x}\bar{L}_{x}=L_{E},$
where

\[
\bar{L}_{n}\equiv\!\!\begin{array}{c}
\frac{1}{G\left(a_{1}\right)-G\left(a_{2}\right)}\end{array}\!\!\!\intop_{y\in S_{n}}\!\! L\left(q_{n}\left(y\right);a_{n}\left(y\right)\right)dG\left(a_{n}\left(y\right)\right)\!,\,\bar{L}_{x}\equiv\!\!\begin{array}{c}
\frac{1}{G\left(a_{2}\right)}\end{array}\!\!\!\intop_{y\in S_{x}}\!\! L\left(q_{x}\left(y\right);a_{x}\left(y\right)\right)dG\left(a_{x}\left(y\right)\right)\!.
\]
Under the same assumptions as before, these integrals may be evaluated
in closed form. Again, other variables of interest, such as trade
flows or welfare, may be obtained in a similar way.

\subsubsection{Flexible \citeauthor{melitz}/Melitz-Ottaviano model with non-separable
utility}

While a significant part of the international trade literature relies
on separable utility functions, there exist realistic economic phenomena
what are more easily modeled with non-separable utility. An instantly
classic alternative to the \citeauthor{melitz} model that uses non-separable
utility is the model of \citeauthor{melitzottaviano}, which assumes
that with a greater selection of heterogeneous-good varieties available
to consumers, the marginal gain from an additional variety decreases
relative to the gains from increased quantity. Trade liberalization
leads to tougher competition, which results not only in higher productivity
but also in the decrease of markups charged by a given firm.

Here we briefly discuss a generalization of the flexible Melitz model
where the utility function is allowed to be non-separable. This generalized
model contains as special cases both the \citeauthor{melitz} model
and the \citeauthor{melitzottaviano} model.%
\footnote{In addition to the heterogeneous-good varieties explicitly considered
here, the \citeauthor{melitzottaviano} model includes a homogeneous
good. In our discussion, the homogeneous good is absent but adding
it to the model is straightforward.%
} The utility is of the form
\[
U_{\Omega}\equiv F\left(U_{\Omega}^{\left(1\right)},U_{\Omega}^{\left(2\right)},...,U_{\Omega}^{\left(m\right)}\right),\quad U_{\Omega}^{\left(i\right)}\equiv\int_{\Omega}U^{\left(i,\omega\right)}\left(q_{\omega}\right)\, d\omega.
\]
In order to preserve tractability, we assume that $U^{\left(i,\omega\right)}\left(q_{\omega}\right)$
are linear combinations%
\footnote{Of course, without loss of generality we could assume that $U^{\left(i,\omega\right)}\left(q_{\omega}\right)$
are power functions and let the function $F$ combine them into any
desired linear combinations. However, for clarity of notation it is
preferable to keep the number $m$ of different expressions $U_{\Omega}^{\left(i\right)}$
small.%
} of equally-spaced powers of $q_{\omega}$ and that the set of exponents
does not depend on $i$ or $\omega$. For example, we could specify
$U_{\Omega}\equiv U_{\Omega}^{\left(1\right)}+\kappa_{1}(U_{\Omega}^{\left(1\right)})^{\xi_{1}}+\kappa_{2}(U_{\Omega}^{\left(2\right)})^{\xi_{2}}$,
$U_{\Omega}^{\left(1\right)}\equiv\int_{\Omega}q_{\omega}^{\gamma_{1}}\, d\omega$,
and $U_{\Omega}^{\left(2\right)}\equiv\int_{\Omega}q_{\omega}^{\gamma_{2}}\, d\omega$,
with $(\gamma_{1}+1)/(\gamma_{2}+1)$ equal to the ratio of two small
integers. The choice $\kappa_{1}=\kappa_{2}=0$ corresponds to the
\citeauthor{melitz} model, while the choice $\xi_{1}=2$, $\xi_{2}=1$,
$\gamma_{1}=1$, and $\gamma_{2}=2$ gives the \citeauthor{melitzottaviano}
model, which is based on a non-homothetic quadratic utility. Our general
specification allows also for homothetic non-separable utility functions
that feature market toughness effects analogous to those in the \citeauthor{melitzottaviano}
model.

It is straightforward to verify that just like the flexible \citeauthor{melitz}
model with separable utility, this more general version leads to tractable
optimization by individual firms, as well as for tractable aggregation
under the same conditions. The reason for the tractability of the
firm's problem is simple: the firm's first-order condition will have
the same structure as previously, a linear combination of equidistant
powers (with an additional dependence of the coefficients of the linear
combination on \emph{aggregate} variables of the type $\int_{\Omega}q_{\omega}^{\gamma}\, d\omega$
for some constants $\gamma$). Given that the nature of the firm's
problem is unchanged, it follows that being able to explicitly aggregate
over heterogeneous firms does not require any additional functional
form assumptions relative to the separable utility case.

\section{Demand Forms}

\label{forms}

\subsection{Curvature properties}

\label{curvature}

Table \ref{taxonomy} provides a taxonomy of the curvature properties
of demand functions generated by common statistical distributions
and the single-product version of the Almost Ideal Demand System.
Following \citet{socialnalebuff,imperfectnalebuff}, we define the
curvature of demand as 
\[
\kappa(p)\equiv\frac{Q''(p)Q(p)}{\left[Q'(p)\right]^{2}}.
\]
\citet{cournot} showed that the pass-through rate of a constant marginal
cost monopolist is 
\[
\frac{1}{2-\kappa}
\]
and thus that a) that the comparison of $\kappa$ to unity determines
the comparison of pass-through to unity in this case and b) that if
$\kappa'(p)>0$ that pass-through rises with price (falls with quantity),
and conversely if $\kappa$ declines with price (rises with quantity).
The comparison of $\kappa$ to unity also determines whether a demand
is log-convex and its sign whether demand is convex. The comparison
of $\kappa$ to $2$ determines whether demand has declining marginal
revenue, a condition also known as \citet{myerson}'s regularity condition.

For probability distribution $F$, the corresponding demand function
$Q(p)=s\left(1-F\left(\frac{p-\mu}{m}\right)\right)$ where $s$ and
$m$ are stretch parameters \citep{wt} and $\mu$ is a position parameter.
Note that in this case 
\[
\kappa(p)=-\frac{\frac{s^{2}}{m^{2}}F''\left(\frac{p-\mu}{m}\right)\left(1-F\left(\frac{p-\mu}{m}\right)\right)}{\frac{s^{2}}{m^{2}}\left[F'\left(\frac{p-\mu}{m}\right)\right]^{2}}=-\frac{F''\left(\frac{p-\mu}{m}\right)\left(1-F\left(\frac{p-\mu}{m}\right)\right)}{\left[F'\left(\frac{p-\mu}{m}\right)\right]^{2}}.
\]
Note, thus, that neither global level nor slope properties of $\kappa$
are affected by $s,m$ or $\mu$. We can thus analyze the properties
of relevant distributions independently of their values, as represented
in the table and the following proposition.

The most prominent conclusion emerging from this taxonomy is that
the vast majority of forms used in practice in computational, statistical
models such as \citet{blp} have monotonically increasing curvature
and most have curvature below unity. This suggests two conclusions.
The first, highlighted in the paper, is that, to the extent we believe
these forms are more realistic than tractable forms, they have properties
systematically differing from the BP class and thus it is important
to derive tractable forms capable of matching their central property
of monotonically increasing in price/decreasing in quantity curvature.

A second possible conclusion is that, to the extent that in some cases
these properties are {\em not} empirically relevant, such as in
the data of \citet{levinisgod}, standard forms rule out observed
behavior and thus analysts may wish to consider more flexible forms
along these dimensions, such as those we derive in the paper. To the
extent there are not strong theoretical reasons to believe in the
restrictions imposed by standard statistically based forms (which,
in many cases, there are) allowing such relaxation is important because
in many contexts the properties of firm demand and equilibrium are
inherited directly from the demand function, at least with constant
marginal cost \citep{weylfabinger,gabaixetal,quint}. Which conclusion
is most appropriate will obviously depend on the empirical context
and the views of the analyst.

\begin{table}[t]
{\scriptsize{}{}}%
\begin{tabular}{|c|c|c|c|c|}
\hline 
 & {\scriptsize{}{$\kappa<1$ }}  & {\scriptsize{}{$\kappa>1$ }}  & {\scriptsize{}{$\left.\begin{array}{c}
\text{Price-}\\
\text{dependent}
\end{array}\right.$ }}  & {\scriptsize{}{$\!\!\begin{array}{c}
\text{Parameter-}\\
\text{dependent}
\end{array}\!\!$ }}\tabularnewline
\hline 
{\scriptsize{}{$\begin{array}{c}
\kappa^{\prime}<0\end{array}$ }}  &  &  & {\scriptsize{}{AIDS with $b<0$ }}  & \tabularnewline
\hline 
{\scriptsize{}{$\begin{array}{c}
\kappa^{\prime}>0\end{array}$ }}  & {\scriptsize{}{$\left.\begin{array}{c}
\text{ Normal (Gaussian) }\\
\text{Logistic}\\
\text{Type I Extreme Value}\\
\text{(Gumbel)}\\
\text{Laplace}\\
\text{Type III Extreme Value}\\
\text{ (Reverse Weibull)}\\
\text{Weibull with shape}~\alpha>1\\
\text{Gamma with shape}~\alpha>1
\end{array}\right.$ }}  &  & {\scriptsize{}{$\left.\begin{array}{c}
\text{Type II}\\
\text{ Extreme Value}\\
\text{ (Fr\'echet) with}\\
\text{ shape}~\alpha>1
\end{array}\right.$ }}  & \tabularnewline
\hline 
{\scriptsize{}{$\begin{array}{c}
\text{Price-}\\
\text{dependent}
\end{array}$ }}  &  &  &  & \tabularnewline
\hline 
{\scriptsize{}{$\begin{array}{c}
\text{Parameter-}\\
\text{dependent}
\end{array}$ }}  &  &  &  & \tabularnewline
\hline 
{\scriptsize{}{$\begin{array}{c}
\text{Does not}\\
\text{globally }\\
\text{satisfy}\\
\kappa<2
\end{array}$ }}  &  & {\scriptsize{}{$\left.\begin{array}{c}
\text{Type II Extreme Value}\\
\text{ (Fr\'echet) with shape}~\alpha<1\\
\text{Weibull with shape}~\alpha<1\\
\text{Gamma with shape}~\alpha<1
\end{array}\right.$ }}  &  & \tabularnewline
\hline 
\end{tabular}\caption{A taxonomy of some common demand functions}

\label{taxonomy} 
\end{table}

\newtheorem{proposition}{Proposition}

\begin{proposition}\label{categorization} Table \ref{taxonomy}
summarizes global properties of the listed statistical distributions
generating demand functions. $\alpha$ is the standard shape parameter
in distributions that call for it. \end{proposition}

\begin{proof} Characterization of the curvature level (comparisons
of $\kappa$ to unity) follow from classic classifications of distributions
as log-concave or log-convex as in \citet{logconcave}, except in
the case of AIDS in which the results are novel.\footnote{We do not classify the slope of pass-through for demand functions
violating declining marginal revenue as this is such a common assumption
that we think such forms would be unlikely to be widely used and because
it is hard to classify the slope of pass-through when it is infinite
over some ranges.} Note that our discussion of stretch parameters in the paper implies
we can ignore the scale parameter of distributions, normalizing this
to $1$ for any distributions which has one. A similar argument applies
to position parameter: because this only shifts the values where properties
apply by a constant, it cannot affect global curvature or higher-order
properties. This is useful because many of the probability distributions
we consider below have scale and position parameters that this fact
allows us to neglect. We will denote this normalization by {\em
Up to Scale and Position} (USP).

We begin by considering the first part of the proof, that for any
shape parameter $\alpha<1$ the Fr\'echet, Weibull and Gamma distributions
with shape $\alpha$ violate DMR at some price. We show this for each
distribution in turn: 
\begin{enumerate}
\item Type II Extreme Value (Fr\'echet) distribution: USP, this distribution
is $F(x)=e^{-x^{-\alpha}}$ with domain $x>0$. Simple algebra shows
that 
\[
\kappa(x)=\frac{(e^{x^{-\alpha}}-1)x^{\alpha}(1+\alpha)+\left(1-e^{x^{-\alpha}}\right)\alpha}{\alpha}.
\]
As $x\rightarrow\infty$ and therefore $x^{-\alpha}\rightarrow0$
(as shape is always positive), $e^{x^{-\alpha}}$ is well-approximated
by its first-order approximation about $0$, $1+x^{-\alpha}$. Therefore
the limit of the above expression is the same as that of 
\[
\frac{x^{-\alpha}x^{\alpha}(1+\alpha)-x^{-\alpha}\alpha}{\alpha}=\frac{1+\alpha+x^{-\alpha}\alpha}{\alpha}\rightarrow\frac{1}{\alpha}+1
\]
as $x\rightarrow\infty$. Clearly, this is greater than $2$ for $0<\alpha<1$
so that for sufficiently large $x$, $\kappa>2$. 
\item Weibull distribution: USP, this distribution is $F(x)=1-e^{-x^{\alpha}}$.
Again algebra yields: 
\[
\kappa(x)=\frac{1-\alpha}{\alpha x^{\alpha}}+1.
\]
Clearly, for any $\alpha<1$ as $x\rightarrow0$ this expression goes
to infinity, so that for sufficiently small $x$, $\kappa>2$. 
\item Gamma distribution: USP, this distribution is $F(x)=\frac{\gamma(\alpha,x)}{\Gamma(\alpha)}$
where $\gamma(\cdot,\cdot)$ is the lower incomplete Gamma function,
$\Gamma(\cdot,\cdot)$ is the upper incomplete Gamma function and
$\Gamma(\cdot)$ is the (complete) Gamma function: 
\begin{equation}
\kappa(x)=\frac{e^{x}(1-\alpha+x)\Gamma(\alpha,x)}{x^{\alpha}}.\label{gammak}
\end{equation}
By definition, $\lim_{x\rightarrow0}\Gamma(\alpha,x)=\Gamma(\alpha)>0$
so 
\[
\lim_{x\rightarrow0}\kappa(x)=+\infty
\]
as $1-\alpha>0$ for $\alpha<1$. Thus clearly for small enough $x$,
the Gamma distribution with shape $\alpha<1$ has $\kappa>2$. 
\end{enumerate}
We now turn to the categorization of demand functions as having increasing
or decreasing pass-through. As price always increases in cost, this
can be viewed as either pass-through as a function of price or pass-through
as a function of cost. 
\begin{enumerate}
\item Normal (Gaussian) distribution: USP, this distribution is given by
$F(x)=\Phi(x)$, where $\Phi$ is the cumulative normal distribution
function; we let $\phi$ denote the corresponding density. It is well-known
that $\Phi''(x)=-x\phi(x)$. Thus 
\[
\kappa(x)=\frac{x\left[1-\Phi(x)\right]}{\phi(x)}.
\]
Taking the derivative and simplifying yields 
\[
\kappa'(x)=\frac{\left[1-\Phi(x)\right]\left(1+x^{2}\right)-x\phi(x)}{\phi(x)},
\]
which clearly has the same sign as its numerator, as $\phi$ is a
density and thus everywhere positive. But a classic strict lower bound
for $\Phi(x)$ is $\frac{x}{1+x^{2}}\phi(x)$, implying $\kappa'>0$. 
\item Logistic distribution: USP, this distribution is $F(x)=\frac{e^{x}}{1+e^{x}}$.
Again algebra yields 
\[
\kappa'(x)=e^{-x}>0.
\]
Thus the logistic distribution has $\kappa'>0$. 
\item Type I Extreme Value (Gumbel) distribution: USP, this distribution
has two forms. For the minimum version it is $F(x)=1-e^{-e^{x}}$.
Algebra shows that for this distribution 
\[
\kappa'(x)=e^{-x}.
\]
Note that this is the same as for the logistic distribution; in fact
$\kappa$ for the Gumbel minimum distribution is identical to the
logistic distribution. This is not surprising given the close connection
between these distributions \citep{mcfadden}.

For the maximum version it is $F(x)=e^{-e^{-x}}$. Again algebra yields
\[
\kappa'(x)=e^{-x}\big(e^{2x}[e^{e^{-x}}-1]-e^{e^{-x}}[e^{x}-1]\big).
\]
For $x<0$ this is clearly positive as both terms are strictly positive:
$1>e^{x}$ and because $e^{-x}>0$, $e^{e^{-x}}>1$. For $x>0$ we
can rewrite $\kappa'$ as 
\[
e^{e^{-x}}\left(e^{x}-1\right)+e^{-x}\left(e^{e^{-x}}-1\right),
\]
which again is positive as $e^{x}>1$ for $x>0$ and $e^{e^{-x}}>1$
by our argument above.

\item Laplace distribution: USP, this distribution is 
\[
F(x)=\left\{ \begin{array}{cc}
1-\frac{e^{-x}}{2} & x\geq0,\\
\frac{e^{x}}{2} & x<0.
\end{array}\right.
\]
For $x>0$, $\rho=1$ (so in this range pass-through is not strictly
increasing). For $x<0$ 
\[
\kappa'(x)=2e^{-x}>0.
\]
So the Laplace distribution exhibits globally weakly increasing pass-through,
strictly increasing for prices below the mode. The curvature for this
distribution is $1-2e^{-x}$ as opposed to $1-e^{-x}$ for Gumbel
and Logistic. However, these are very similar, again pointing out the
similarities among curvature properties of common demand forms. 
\item Type II Extreme Value (Fr\'echet) distribution with shape $\alpha>1$:
From the formula above it is easy to show that the derivative of the
pass-through rate is 
\[
\kappa'(x)=x^{-(1+\alpha)}\Big([1+\alpha]\big[x^{2\alpha}(e^{x^{-\alpha}}-1)-e^{x^{-\alpha}}x^{\alpha}\big]+\alpha e^{x^{-\alpha}}\Big)>0,
\]
which can easily be shown to be positive as follows. Let us multiply
the inequality by the positive factor $\frac{e^{-x^{-\alpha}}}{\alpha+1}$.
Denoting $X\equiv x^{-\alpha}$, the inequality becomes 
\[
\left(\frac{\alpha}{\alpha+1}-\frac{1}{2}\right)+\left(\frac{1}{X^{2}}-\frac{e^{-X}}{X^{2}}-\frac{1}{X}+\frac{1}{2}\right)>0.
\]
The first term is positive because $\alpha>1$. The second term is
positive because $e^{-X}<1-X+\frac{1}{2}X^{2}$ for any $X>0$. Thus
this distribution, as well, has $\kappa'>0$. 
\item Type III Extreme Value (Reverse Weibull) distribution: USP, this distribution
is $F(x)=e^{-(-x)^{\alpha}}$ and has support $x<0$. Algebra shows
\[
\kappa'(x)=(-x)^{\alpha-1}\alpha^{2}\left[1-\alpha+e^{(-x)^{\alpha}}\Big([1-\alpha]\big[(-x)^{\alpha}-1\big]+[-x]^{2\alpha}\alpha\Big)\right],
\]
which has the same sign as 
\begin{equation}
1-\alpha+e^{(-x)^{\alpha}}\Big([1-\alpha]\big[(-x)^{\alpha}-1\big]+[-x]^{2\alpha}\alpha\Big).\label{type3firstderiv}
\end{equation}
Note that the limit of this expression as $x\rightarrow0$ is 
\[
1-\alpha-(1-\alpha)=0
\]
and its derivative is 
\[
\frac{e^{(-x)^{\alpha}}(-x)^{2\alpha}\alpha\big(1+\alpha+[-x]^{\alpha}\alpha\big)}{x},
\]
which is clearly strictly negative for $x<0$. Thus Expression \ref{type3firstderiv}
is strictly decreasing and approaches $0$ as $x$ approaches $0$.
It is therefore positive for all negative $x$, showing that again
in this case $\kappa'>0$. 
\item Weibull distribution with shape $\alpha>1$: As with the Fr\'echet distribution
algebra from the earlier formula shows 
\[
\kappa'(x)=x^{\alpha-1}(\alpha-1)\alpha^{2},
\]
which is clearly positive for $\alpha>1$ as the range of this distribution
is positive $x$. Thus the Weibull distribution with $\alpha>1$ has
$\kappa'>0$. 
\item Gamma distribution with shape $\alpha>1$: Taking the derivative of
Expression \ref{gammak} yields: 
\[
\kappa'(x)=\frac{\alpha-1-x+\frac{e^{x}}{x^{\alpha}}\big(x^{2}-2x[\alpha-1]+[\alpha-1]\alpha\big)\Gamma(\alpha,x)}{x},
\]
which has the same sign as 
\begin{equation}
\alpha-1-x+\frac{e^{x}}{x^{\alpha}}\big(x^{2}-2x[\alpha-1]+[\alpha-1]\alpha\big)\Gamma(\alpha,x),\label{gammaexpression}
\end{equation}
given that $x>0$. Note that as long as $\alpha>1$ 
\[
x^{2}+(\alpha-2x)(\alpha-1)=x^{2}-2(\alpha-1)x+\alpha(\alpha-1)>x^{2}-2(\alpha-1)x+(\alpha-1)^{2}=\left(x+1-\alpha\right)^{2}>0.
\]
Therefore so long as $x\leq\alpha-1$ this is clearly positive. On
the other hand when $x>\alpha-1$ the proof depends on the following
result of \citet{gamma}:

\newtheorem*{g}{Theorem \citep{gamma}}

\begin{g} Let a be a positive parameter, and let q(x) be a function,
differentiable on $(0,\infty)$ , such that $lim_{x\rightarrow\infty}x^{\alpha}e^{-x}q(x,\alpha)=0$.
Let

\[
T(x,\alpha)=1+(\alpha-x)q(x,\alpha)+x\frac{\partial q}{\partial x}(x,\alpha).
\]

If $T(x,\alpha)>0$ for all $x>0$ then $\Gamma(\alpha,x)>x^{\alpha}e^{-x}q(x,\alpha)$.
\end{g}

Letting

\[
q(x,\alpha)\equiv\frac{x-(\alpha-1)}{x^{2}+(\alpha-2x)(\alpha-1)},
\]

\[
T(x,\alpha)=\frac{2(\alpha-1)x}{\big(\alpha^{2}+x[2+x]-\alpha[1+2x]\big)^{2}}>0
\]

for $\alpha>1,x>0$. So $\Gamma(\alpha,x)>x^{\alpha}e^{-x}q(x,\alpha)$.
Thus Expression \ref{gammaexpression} is strictly greater than

\[
\alpha-1-x+x-(\alpha-1)=0
\]

as, again, $x^{2}+(\alpha-2x)(\alpha-1)>0$. Thus again $\kappa'>0$.

\end{enumerate}
This establishes the second part of the proposition. Turning to our
final two claims, algebra shows that the curvature for the Fr\'echet
distribution is

\[
\kappa(x)=\frac{\alpha-e^{x^{-\alpha}}\big(\alpha-x^{\alpha}[1+\alpha]\big)-x^{\alpha}(1+\alpha)}{\alpha}=\frac{\left(1-e^{x^{-\alpha}}\right)\left[\alpha-x^{\alpha}\left(1+\alpha\right)\right]}{\alpha}.
\]

Note for any $\alpha>1$ this is clearly continuous in $x>0$. Now
consider the first version of the expression. Clearly, as $x\rightarrow0$,
$x^{\alpha}\rightarrow0$ and $e^{x^{-\alpha}}\rightarrow\infty$
so the expression goes to $-\infty$. So for sufficiently small $x>0$,
$\kappa(x)<1$. On the other hand, consider the second version of the expression.
Its numerator is

\[
\left(1-e^{x^{-\alpha}}\right)\left[\alpha-x^{\alpha}\left(1+\alpha\right)\right].
\]
By the same argument as above with the Fr\'echet distribution the limit
of the above expression as $x\rightarrow\infty$ is the same as that
of 
\[
\left(-x^{-\alpha}\right)\left(-x^{\alpha}\left(1+\alpha\right)\right)
\]
as $x\rightarrow\infty$. Thus 
\[
\lim_{x\rightarrow\infty}\kappa(x)=\frac{1+\alpha}{\alpha}>1
\]
and thus for sufficiently large $x$ and any $\alpha>1$, this distribution
has $\kappa>1$.

Finally, consider our claim about AIDS. First note that for this demand
function

\[
\kappa(p)=2+\frac{b\big(a-2b+b\log p\big)}{\big(a-b+b\log p\big)^{2}}<1
\]

as $b<0$ and $p\leq e^{-\frac{a}{b}}<e^{2-\frac{a}{b}}$. This is
less than $1$ if and only if 
\[
a^{2}+2ab\big(\log p-2\big)+b^{2}\Big(1+\big[\log(p)-2\big]\log p\Big)<b^{2}\big(2-\log p\big)-ab
\]
or 
\[
\big(a+b\log p\big)^{2}-b^{2}\big(\log p+1\big)<0.
\]
Clearly, as $p\rightarrow0$ the second term is positive; therefore
there is always a price at which $\kappa(p)>1$. On the other hand
as $p\rightarrow e^{-\frac{a}{b}}$ this expression goes to 
\[
0-b^{2}\bigg(1-\frac{a}{b}\bigg)=b(a-b)<0.
\]
Thus there is always a price at which $\kappa(p)<1$. 
\[
\kappa'(p)=b^{2}-\big(a-2b+b\log p\big)^{2},
\]
which has the same sign as 
\[
b^{2}-\big(a-2b+b\log p\big)^{2}<b^{2}-(2b)^{2}=-3b^{2}<0.
\]
Thus $\kappa'<0$.

\end{proof}

We now turn to two important distributions, which are typically used
to model the income distribution, whose behavior is more complex and
which, to our knowledge, have not been analyzed for their curvature
properties. We focus only on the two that we believe to be most common
(the first), best theoretically founded (both) and to provide the
most accurate match to the income distribution (the second). Namely,
we analyze the lognormal and double Pareto-lognormal (dPln) distributions,
the latter of which was proposed by \citet{reed} and \citet{reedjorgensen}.
Other common, accurate models of income distributions which we have
analyzed in less detail, appear to behave in a similar fashion.

We begin with the lognormal distribution, which is much more commonly
used, and for which we have detailed, analytic results. However, while
most of the arguments for the below proposition are proven analytically,
some simple points are made by computational inspection.

\begin{proposition} For every value $\sigma$, there exist finite
thresholds $\overline{y}(\sigma)>\underline{y}(\sigma)$ such that 
\begin{enumerate}
\item If $y\geq\overline{y}(\sigma)$ then $\kappa'\leq0$, and similarly
with strict inequalities or if the directions of the inequalities
both reverse. 
\item If $y\geq\underline{y}(\sigma)$ then $\kappa\geq1$, and similarly
with strict inequalities or if the directions of the inequalities
both reverse. 
\end{enumerate}
Both $\overline{y}$ and $\underline{y}$ are strictly decreasing
in $\sigma$. \end{proposition}

\begin{figure}
\begin{centering}
\includegraphics[width=5in]{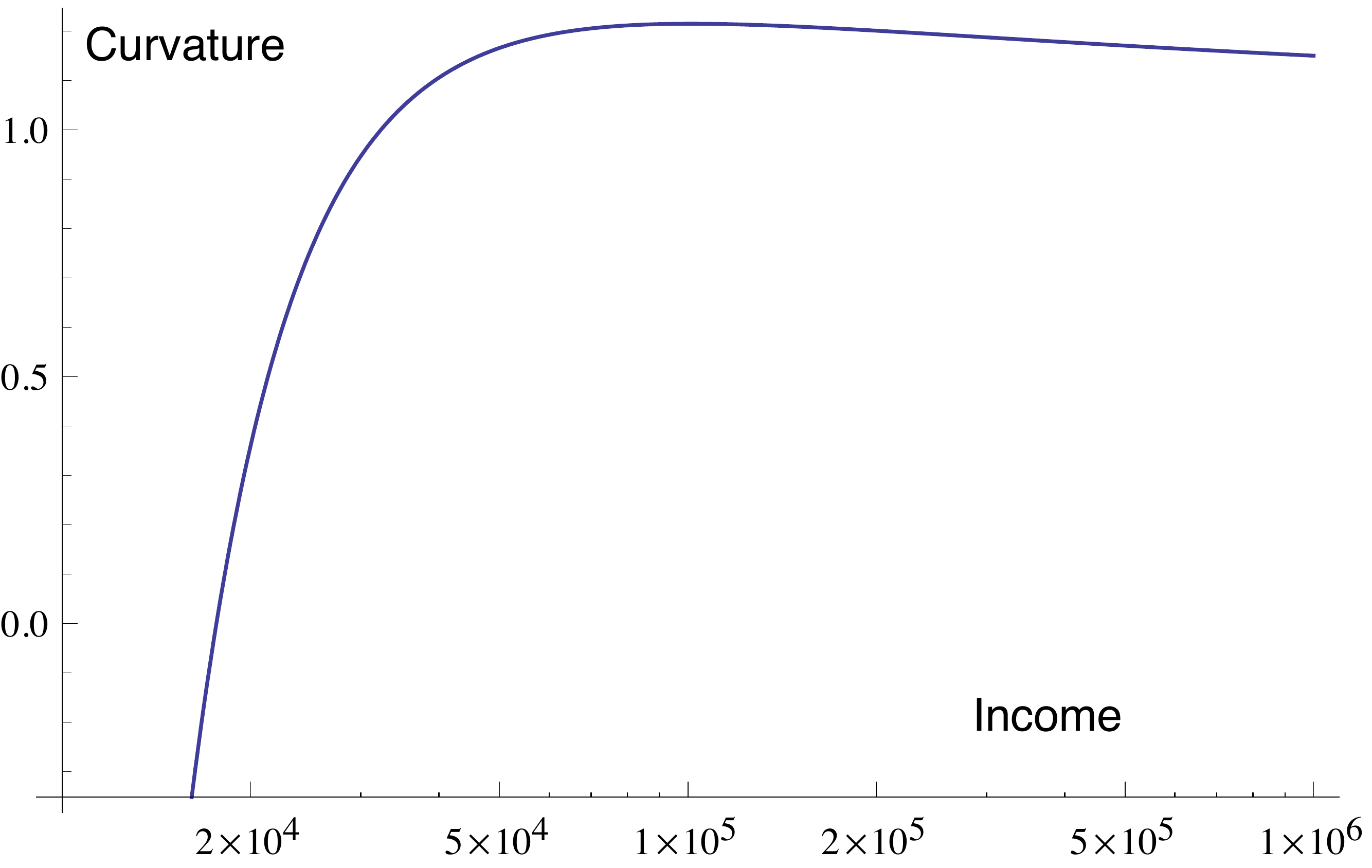} 
\par\end{centering}

\caption{Curvature of a lognormal distribution calibrated to the US income
distribution: parameters are $\mu=10.5$ and $\sigma=0.85$.}

\label{lnincome} 
\end{figure}

Under the lognormal distribution, the behavior depends critically on the
amount of inequality or equivalently the standard deviation of the
logarithm of the distribution: there is famously a one-to-one relationship
between the Gini coefficient associated with a lognormal distribution
and its logarithmic standard deviation. If inequality is not high,
the behavior of curvature like a normal distribution occurs except
at fairly high incomes levels; for a Gini of $.34$, for example,
monotonicity of $\kappa$ is preserved until the top 1\% of the income
distribution and log-concavity outside of the top 30\%. However, if
inequality is sufficiently high, in particular if the Gini coefficient
is above about $.72$, then the lognormal distribution has $\kappa>2$
over some range and then $\kappa$ converges back to $1$ for very
large incomes. This result is not discussed in the proposition but
can easily be seen by inspecting a graph of the expression for $\kappa$
given in the proof of the proposition for various values of $\sigma$
yielding Gini coefficients of various magnitudes around $.72$.

For intermediate levels of inequality between these, like that seen
in nearly every country, the lognormal distribution has curvature
that rises from $-\infty$ to above unity before gradually returning
towards unity. For an example calibrated to the US income distribution
(Figure \ref{lnincome}), the crossing to above unity occurs at an
income of about \$33k, between the mode and the median and the downward
slope begins at about \$100k. Despite this, curvature never falls
below unity again and in fact is at each quantile increasing in $\sigma$
(again, not discussed in the proposition). Again taking the example
of the US income-calibrated distribution, curvature peaks at about
$1.21$ and only falls to $1.20$ by \$200k, eventually leveling out
to about $1.1$ for the extremely wealthy.\footnote{Note, however, that in the true limit as $y\rightarrow\infty$, $\kappa\rightarrow1$.
However, in practice this occurs at such high income levels that the
asymptote to a bit above $1$ is a more realistic representation.} Thus, in practice, curvature is closer to flat at the top than significantly
declining.

\begin{proof} For a lognormal distribution with parameters $(\mu,\sigma)$,
$F(x)=\Phi\left(\frac{\log(x)-\mu}{\sigma}\right)$, so that 
\[
Q(p)=1-\Phi\left(\frac{\log(p)-\mu}{\sigma}\right),Q'(p)=-\frac{\phi\left(\frac{\log(p)-\mu}{\sigma}\right)}{\sigma p}
\]
and 
\[
Q''(p)=-\frac{\phi'\left(\frac{\log(p)-\mu}{\sigma}\right)}{\sigma^{2}p^{2}}+\frac{\phi\left(\frac{\log(p)-\mu}{\sigma}\right)}{\sigma p^{2}}=-\frac{\phi\left(\frac{\log(x)-\mu}{\sigma}\right)}{\sigma^{2}p^{2}}\left(\sigma+\frac{\log(x)-\mu}{\sigma}\right).
\]
where the second equality follows from the identities regarding the
normal distribution from the previous proof and $y\equiv\frac{\log(p)-\mu}{\sigma}$.
Thus 
\begin{equation}
\kappa\left(p(y)\right)=\frac{\left(y+\sigma\right)\left[1-\Phi(y)\right]}{\phi(y)}.\label{lognormalcurve}
\end{equation}
Note that we immediately see, as discussed above, that $\kappa$ increases
in $\sigma$ at each quantile as the inverse hazard rate $\frac{1-\Phi}{\phi}>0$;
similarly, for any quantile associated with $y$, $\kappa\rightarrow\infty$
as $\sigma\rightarrow\infty$ so it must be that the set of $y$ for
which $\kappa>1$ a) exists for sufficiently large $\sigma$ and b)
expands monotonically in $\sigma$. This implies that, if point 2)
of the proposition is true, $\underline{y}$ must strictly decrease
in $\sigma$. This also implies that for sufficiently large $\sigma$,
$\kappa>2$ for some $y$.

Now note that $\lim_{y\rightarrow\infty}\frac{y\left[1-\Phi(y)\right]}{\phi(y)}=1$.
To see this, note that both the numerator and denominator converge
to $0$ as $1-\Phi$ dies super-exponentially in $y$. Applying l'Hospital's
rule: 
\[
\lim_{y\rightarrow\infty}\frac{y\left[1-\Phi(y)\right]}{\phi(y)}=\lim_{y\rightarrow\infty}\frac{1-\Phi(y)-\phi(y)y}{\phi'(y)}=\frac{y\phi(y)-\left[1-\Phi(y)\right]}{y\phi(y)}=\frac{0}{0}.
\]
where the first equality follows from the identity for $\phi'$ we
have repeatedly been using, and from here on we no longer note the
use of. Again applying l'Hospital's rule: 
\[
\lim_{y\rightarrow\infty}\frac{y\left[1-\Phi(y)\right]}{\phi(y)}=\lim_{y\rightarrow\infty}\frac{\phi(y)+y\phi'(y)+\phi(y)}{\phi(y)+y\phi'(y)}=\lim_{y\rightarrow\infty}\frac{2\phi(y)-y^{2}\phi(y)}{\phi(y)-y^{2}\phi(y)}=\lim_{y\rightarrow\infty}\frac{2-y^{2}}{1-y^{2}}=1.
\]
The same argument, but one step less deep, shows that $\lim_{y\rightarrow\infty}\frac{\sigma\left[1-\Phi(y)\right]}{\phi(y)}=0$.
Together these imply that $\lim_{y\rightarrow\infty}\kappa\left(p(y)\right)=1$
and thus that, if $\kappa>1$ at some point, it must eventually decrease
to reach $1$.

Similar methods may be used to show, as discussed in the paper, that
$\kappa\rightarrow-\infty$ as $y\rightarrow-\infty$. Furthermore,
we know from the proof for the normal distribution above that $\frac{y\left[1-\Phi(y)\right]}{\phi(y)}$
is monotone increasing and that $\frac{\sigma\left[1-\Phi(y)\right]}{\phi(y)}$
is monotone decreasing. The latter point implies that the set of $y$
for which $\kappa$ is decreasing must be strictly increasing in $\sigma$
and thus that, if point 1) of the proposition is true, then $\overline{y}$
must strictly decrease in $\sigma$.

All that remains to be shown is that $\kappa$'s comparison to unity
and the sign of $\kappa'$ obey the threshold structure posited. Note
that we only need to show the cut-off structure for $\kappa'$ and
that this immediately implies the structure for $\kappa$, given the
smoothness of all functions involved, because if $\kappa$ increases
up to some threshold and then decreases monotonically while reaching
an asymptote of unity, it must lie above unity above some threshold.
Otherwise, if it ever crossed below unity, it would have to be increasing
in some region to asymptote to unity at very large $p$, violating
the threshold structure for $\kappa'$. Furthermore, the same logic
implies that the region where $\kappa>1$ must be strictly larger
than the region where $\kappa'<0$ (that $\overline{y}>\underline{y}$) as $\kappa$ must rise strictly above unity before sloping strictly
down towards it.

\begin{figure}
\begin{centering}
\includegraphics[width=3in]{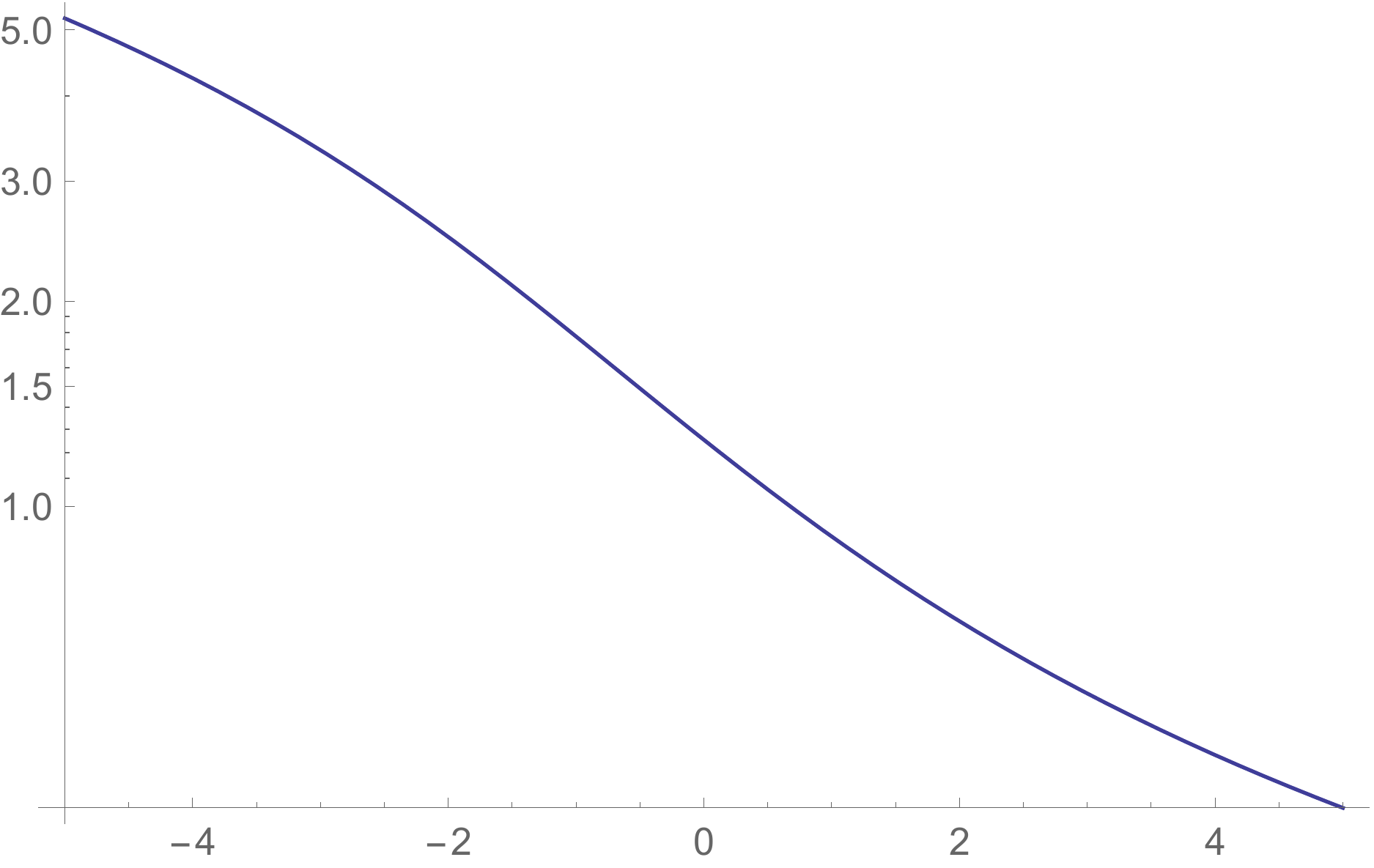} 
\par\end{centering}

\caption{The figure shows the value, in logarithmic scale, of the left-hand
side of Inequality \ref{lognormalcalculations}.}

\label{prooffigure} 
\end{figure}

We drop arguments wherever possible in what follows to ease readability.
We use the symbol $\propto$ to denote expressions having the same
sign, not proportionality as is typical. 
\[
\kappa'=\frac{\left(1-\Phi\right)\phi-(y+\sigma)\phi^{2}-(y+\sigma)(1-\Phi)\phi'}{\phi^{2}}=\frac{1-\Phi-(y+\sigma)\left[\phi-y(1-\Phi)\right]}{\phi}\propto
\]
\[
1-\Phi-(y+\sigma)\left[\phi-y(1-\Phi)\right]\propto\frac{1-\Phi}{\phi-y(1-\Phi)}-y-\sigma.
\]
where the last sign relationship follows by the common inequality
that $\phi(y)>y\left[1-\Phi(y)\right]$. Thus $\kappa'>0$ if and
only if 
\begin{equation}
\frac{1-\Phi}{\phi-y(1-\Phi)}-y>\sigma.\label{lognormalcalculations}
\end{equation}
Figure \ref{prooffigure} shows that the left-hand side of this inequality
is strictly decreasing. We have not found a simple means to prove
this formally, but it is clearly true by inspection of the figure.
Thus the left-hand side of Inequality \ref{lognormalcalculations}
must cross $\sigma$ at most once, and this must be from above to below.

It only remains to show that this expression does, in fact, make such
as single crossing for all values of $\sigma$. It suffices to show
that the small $y$ limit of the left-hand side of inequality \ref{lognormalcalculations}
is $\infty$ and that its large $y$ limit is $0$. We show these
in turn.

The first claim is easy: clearly $-y\left(1-\Phi\right)\rightarrow\infty$,
while $1-\Phi$ is finite, as $y\rightarrow-\infty$. Thus the first
term approaches $0$ and the second $\infty$ as $y\rightarrow-\infty$.

The second claim is more delicate. The expression is the same as 
\[
\frac{\left(1-\Phi\right)\left(1+y^{2}\right)-y\phi}{\phi-y\left(1-\Phi\right)}.
\]
This asymptotes to the indefinite expression $\frac{0}{0}$ as $y\rightarrow\infty$
as it is well-known that $\lim_{y\rightarrow\infty}\frac{\phi}{y\left(1-\Phi\right)}=1$.
Applying l'Hospital's rule yields 
\[
\lim_{y\rightarrow\infty}\frac{\left(1-\Phi\right)\left(1+y^{2}\right)-y\phi}{\phi-y\left(1-\Phi\right)}=\lim_{y\rightarrow\infty}\frac{-\phi\left(1+y^{2}\right)+2y\left(1-\Phi\right)-\phi-y\phi'}{\phi'-\left(1-\Phi\right)+y\phi}=
\]
(applying now-familiar tricks) 
\[
\lim_{y\rightarrow\infty}2\frac{\phi-y(1-\Phi)}{1-\Phi}=\frac{0}{0}.
\]
Again, we apply l'Hospital's rule: 
\[
\lim_{y\rightarrow\infty}2\frac{\phi-y(1-\Phi)}{1-\Phi}=\lim_{y\rightarrow\infty}2\frac{\phi'-(1-\Phi)+y\phi}{\phi}=\lim_{y\rightarrow\infty}-\frac{1-\Phi}{\phi}=0.
\]
\end{proof}

Even the slight decline in the lognormal distribution's curvature
at very high incomes is an artifact of its poor fit to incomes distributions
at very high incomes. It is well-known that at very high incomes the
lognormal distribution fits poorly; much better fit is achieved by
distributions with fatter (Pareto) tails, especially in countries
with high top-income shares like the contemporary United States \citep{topincomes}.
A much better fit is achieved by the dPln distribution \citep{reed}.
Figure \ref{income}'s left panel shows curvature as a function of
income for the parameters \citeauthor{reed} estimates (for the 1997
US income distribution). Curvature monotonically increases up the
income distribution.

\begin{figure}
\begin{centering}
\includegraphics[width=3in]{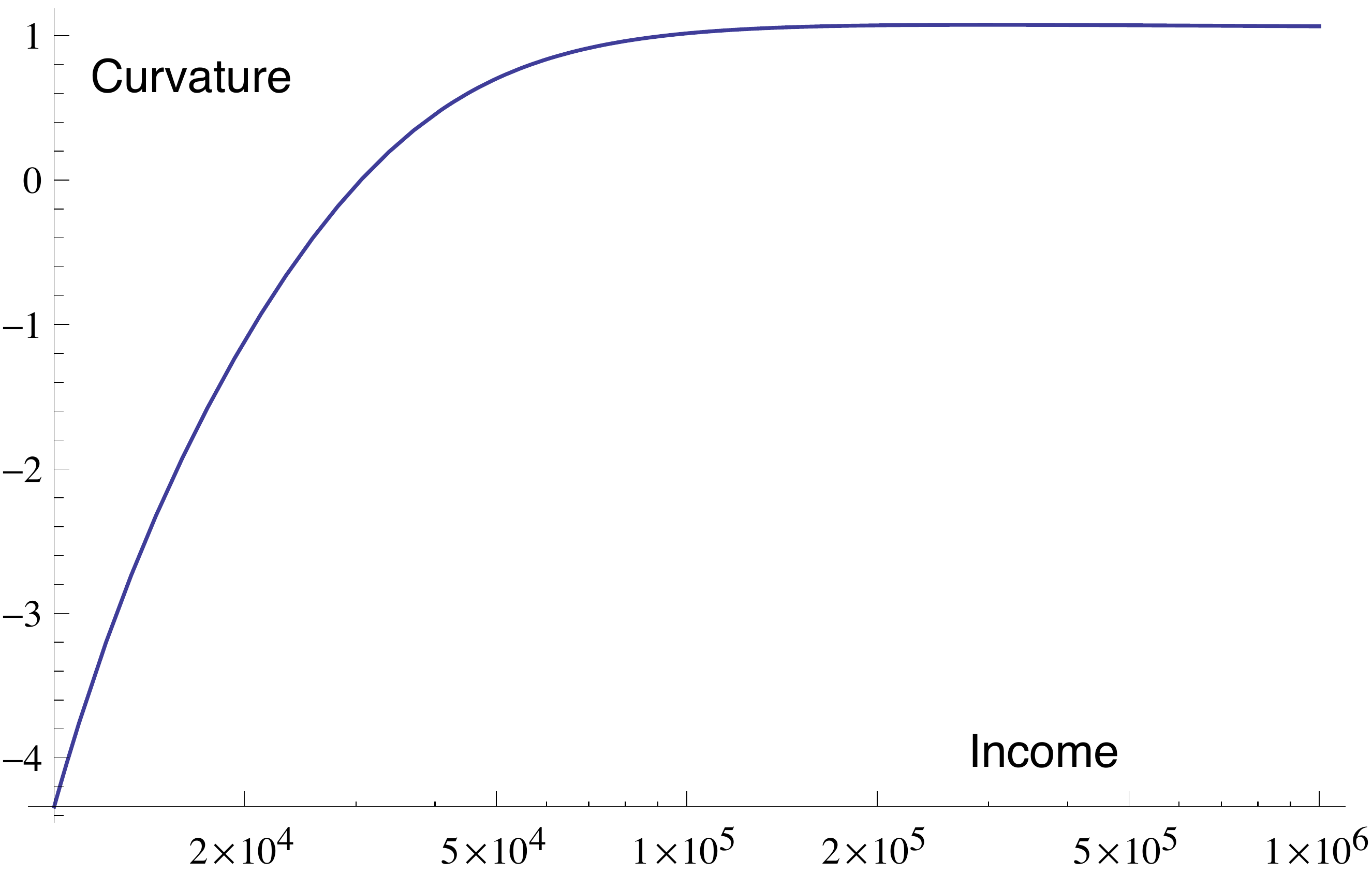} \includegraphics[width=3in]{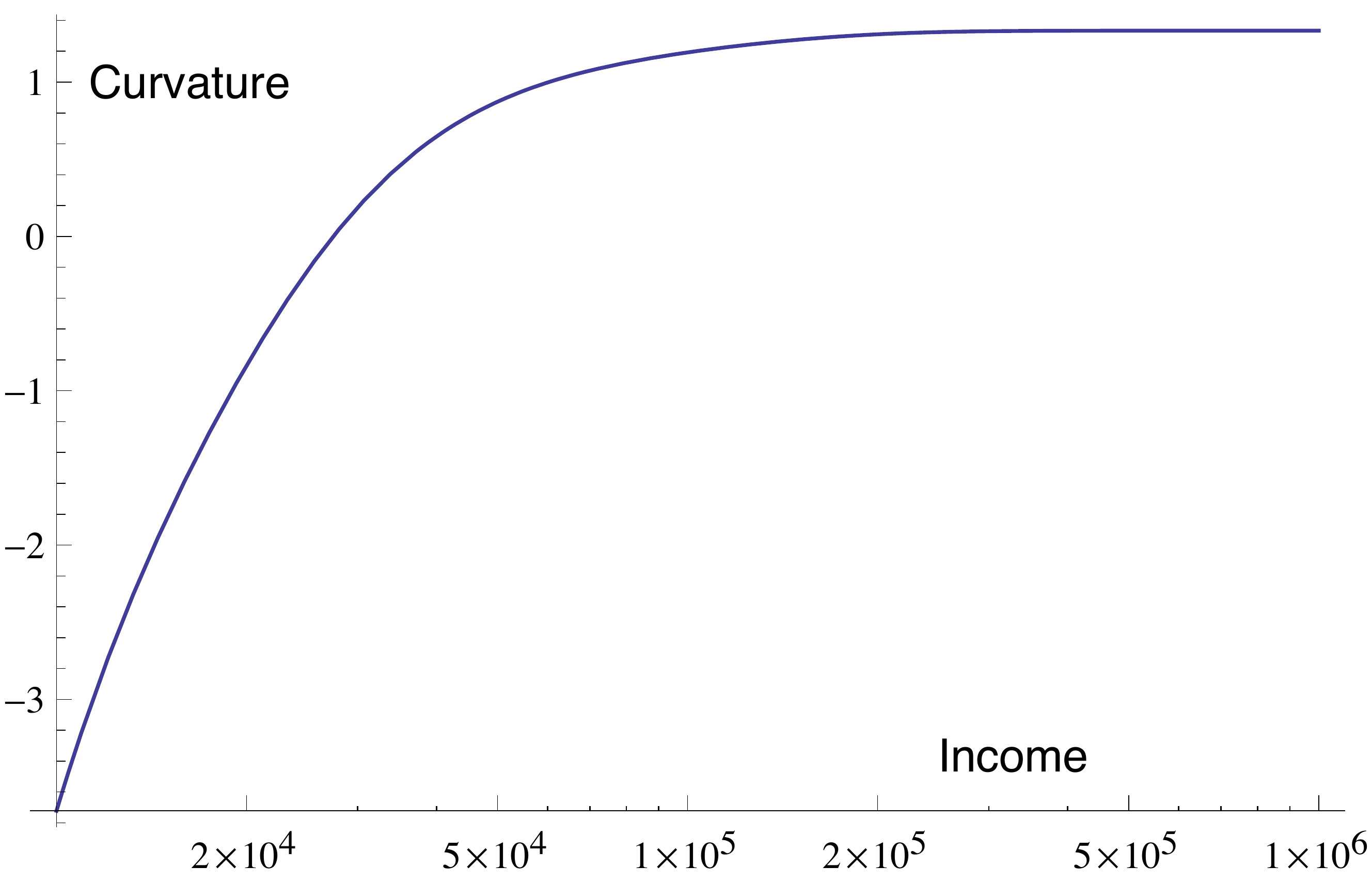} 
\par\end{centering}

\caption{Curvature of the double-Pareto lognormal distribution lognormal under
parameters estimated by \citep{reed} (left) and by updated by us
(right); parameters in the former case are $\alpha=22.43,\beta=1.43,\mu=10.9,\sigma=0.45$
ad in the latter case are $\alpha=3,\beta=1.43,\mu=10.9,\sigma=0.5$.
The x-axis has a logarithmic scale in income.}

\label{income} 
\end{figure}

However, it levels off at quite moderate income (it is essentially
flat beyond \$100k) and at a lower level ($\approx1.04$) than under
the lognormal calibration, except at exorbitant incomes, where the
lognormal distribution has thin tails. Thus it actually has a {\em
thinner} tail, except at the very extreme tail, than the lognormal
calibration, paradoxically. This is because \citeauthor{reed} calibrated
only to the mid-section of the US income distribution, given that
the survey he used is notoriously thin and inaccurate at higher incomes;
this led him to estimate a very high (thin-tailed) Pareto coefficient
in the upper tail of $22.43$. Consensus economic estimates, for example
\citet{progressive}, suggest that 1.5-3 is the correct range
for the Pareto coefficient of the upper tail of the income distribution
in the 2000's.

We therefore construct our own calibration consistent with that finding.
To be conservative we set the upper tail Pareto coefficient to $3$,
maintain $\beta=1.43$ to be consistent with \citeauthor{reed} and
because the lower-tail is both well-measured in his data and has not
changed dramatically in the last decade and a half \citep{striking}.
We then adjust $\mu$ and $\sigma$ in the unique way, given these
coefficients, to match the latest US post-tax Gini estimates ($.42$),
using a formula derived by \citet{dPlnglobal}, and average income
(\$53k). This yields the plot in the right panel of Figure \ref{income}.
There curvature continues to monotonically increase at a significant
rate up to quite high incomes: at \$50k it is $.87$, at \$100k it
is $1.19$ and by \$200k it has leveled off at $1.31$, near its asymptotic
value of $1+\frac{1}{\alpha}=\frac{4}{3}$. It is this last calibration
that we use to represent the dPln calibration US income distribution
in the paper.

Moreover, the monotone increasing nature of curvature is not only
true in the US data. While we have not been able to prove any general
results about this four-parameter class, we have calculated similar
plots to Figure \ref{income} for every country for which a dPln income
distribution has been estimated, as collected by \citeauthor{dPlnglobal}.
In every case curvature is monotone increasing in income, though in
some cases it levels off at a quite low level of income (typically
when the Gini is high relative to the upper tail estimate). Even this
leveling off seems to us likely to be a bit of an artifact, arising
from the lack of reliable top incomes tax data in many of the developing
countries on which \citeauthor{dPlnglobal} focus. In any case, it
appears that a ``stylized fact'' is that a reasonable model of most
country's income distributions has curvature that is significantly
below unity among the poor, rises above unity for the rich and monotone
increasing over the full range so long as top income inequality is
significant relative to overall inequality.

\end{document}